\documentclass[11pt,twoside, letter]{article}

\usepackage{amsthm, amsmath, amssymb, amsfonts, graphicx, epsfig}
\usepackage{accents}

\usepackage{bbm}
\usepackage{mathrsfs}

\usepackage{cancel}

\usepackage{epstopdf}

\usepackage{graphicx}
\usepackage[font=sl,labelfont=bf]{caption}
\usepackage{subcaption}

\usepackage{float}
\restylefloat{table}
\usepackage{multicol}

\usepackage{enumerate}
\usepackage[colorlinks=true, pdfstartview=FitV, linkcolor=blue,
            citecolor=blue, urlcolor=blue]{hyperref}
\usepackage[usenames]{color}

\usepackage{url}
\makeatletter\def\url@leostyle{%
 \@ifundefined{selectfont}{\def\UrlFont{\sf}}{\def\UrlFont{\scriptsize\ttfamily}}} \makeatother

\usepackage{array}

\usepackage{booktabs}

\usepackage[margin=1.0in, letterpaper]{geometry}

\newtheorem{theorem}{Theorem}
\newtheorem{proposition}[theorem]{Proposition}
\newtheorem{lemma}[theorem]{Lemma}

\theoremstyle{definition}
\newtheorem{definition}[theorem]{Definition}
\newtheorem{example}[theorem]{Example}
\theoremstyle{remark}

\numberwithin{equation}{section}
\numberwithin{theorem}{section}

\definecolor{Red}{rgb}{0.9,0,0.0}
\definecolor{Blue}{rgb}{0,0.0,1.0}

\def\cD{\mathcal{D}}

\def\cF{\mathcal{F}}

\def\cI{\mathcal{I}}

\def\cN{\mathcal{N}}

\def\cP{\mathcal{P}}

\def\bE{\mathbb{E}}

\def\bN{\mathbb{N}}

\def\bP{\mathbb{P}}
\def\bQ{\mathbb{Q}}
\def\bR{\mathbb{R}}

\def\sF{\mathscr{F}}

\newcommand{\1}{\mathbbm{1}}            
\newcommand{\set}[1]{\{#1\}}            
\renewcommand{\mid}{\;|\;}              
\newcommand{\abs}[1]{\left\vert#1\right\vert}   

\DeclareMathOperator{\Cov}{Cov}          
\DeclareMathOperator{\Var}{Var}          

\DeclareMathOperator{\var}{\mathrm{V}@\mathrm{R}}           

\usepackage{nth}

\definecolor{darkgreen}{rgb}{0,0.6,0}
\definecolor{midgreen}{rgb}{0,0.5,0}
\usepackage[authormarkup=none]{changes}
\definechangesauthor[color=darkgreen]{TS}
\definechangesauthor[color=midgreen]{MP}
\title{Fair Estimation of Capital Risk Allocation}

\makeatletter
\def\and{%
  \end{tabular}%
  \begin{tabular}[t]{c}}%
\def\@fnsymbol#1{\ensuremath{\ifcase#1\or a\or b\or c\or
   d\or e\or f\or g\or h\or i\else\@ctrerr\fi}}
\makeatother

\author{
        Tomasz R. Bielecki\,\thanks{Department of Applied Mathematics, Illinois Institute of Technology
       \newline \hspace*{1.45em}  10 W 32nd Str, Building REC, Room 208, Chicago, IL 60616, USA
       \newline \hspace*{1.45em}  Emails: \url{tbielecki@iit.edu} (T.R. Bielecki), and \url{cialenco@iit.edu} (I. Cialenco)
       \newline \hspace*{1.45em}  URLs: \url{http://math.iit.edu/\~bielecki}  and \url{http://math.iit.edu/\~igor}
        \vspace{0.5em}} ,
\and
         Igor Cialenco,\,\footnotemark[1] \newline
\and
        Marcin Pitera\,\thanks{ Institute of Mathematics, Jagiellonian University,  Lojasiewicza 6, 30-348 Cracow, Poland
         \newline \hspace*{1.45em}  Email: \url{marcin.pitera@im.uj.edu.pl},  URL: \url{http://www2.im.uj.edu.pl/MarcinPitera/}
         \vspace{0.5em}} ,
\and
        Thorsten Schmidt\,\thanks{ Department of Mathematical Stochastics, University of Freiburg, Eckerstr.1, 79104 Freiburg, Germany
         \newline \hspace*{1.45em} Email: \url{thorsten.schmidt@stochastik.uni-freiburg.de},
         \newline \hspace*{1.45em} URL: \url{http://www.archiv.stochastik.uni-freiburg.de/homepages/schmidt/}
         }
        }

\date{ {\small 
 First Circulated: February 26, 2019 \\ 
 This Version: November 20, 2019
}}

\begin{document}

\maketitle

{\footnotesize
\begin{tabular}{l@{} p{350pt}}
  \hline \\[-.2em]
  \textsc{Abstract}: \ &
  In this paper we develop a novel methodology for  \emph{estimation} of risk capital allocation. The methodology is rooted in the theory of risk measures. We work within a general, but tractable class of law-invariant coherent risk measures, with a particular focus on  expected shortfall. We introduce the concept of \emph{fair} capital allocations and provide explicit formulae for fair capital allocations in case when the constituents of the risky portfolio are jointly normally distributed.
 The main focus of the paper is  on the problem of approximating fair portfolio allocations in the case of not fully known law of the portfolio constituents. We define and study the concepts of fair allocation estimators and asymptotically fair allocation estimators.
  A substantial part of our study is devoted to the problem of estimating fair risk allocations for expected shortfall. We study this problem  under normality as well as in a nonparametric setup.   We derive several estimators, and prove their fairness and/or asymptotic fairness.
  Last, but not least, we propose two backtesting methodologies that are oriented at assessing the performance of the allocation  estimation procedure. The paper closes with a substantial numerical study of the subject.
    \\[0.5em]

\textsc{Keywords:} \ & capital allocation, fair capital allocation, asymptotic fairness, expected shortfall, risk measures, Euler principle, value-at-risk, tail-value-at-risk, backtesting capital allocation.  
\\[1em]
  \hline
\end{tabular}
}


\section{Introduction}

The measurement and the management of risk is without doubt of highest importance in the financial and the insurance industries.
Arguably, the theory and applications of risk measures are  most useful for this purpose.  For early applications in the insurance context see \cite{Buehlmann1970,Gerber1974},  and for a historical perspective in the financial context see  \cite{guill2016bankers}. The seminal article \cite{ArtznerDelbaenEberHeath1999} placed
risk measurements on an axiomatic foundation paving the way to coherent risk measures which have been treated in numerous works since then.  We refer to \cite{Delbaen2000,FoellmerSchied2011,McNFE15Text} for an in-depth treatment of the topic.

The application of risk measures to portfolio management naturally leads to the problem of allocating  portions of the  risk capital to the constituents of the portfolio, i.e.~to the \emph{risk allocation problem.} There are a number of different approaches to risk capital allocation,
depending on the one hand on the class of the used risk measures, and on the other hand on the used allocation principles.
The \emph{Euler principle}, often used in risk management practice, is one example, see e.g. \cite{Tasche2004,Tasch2007}. For coherent risk measures, the Euler principle coincides with the axiomatic approach proposed in \cite{Kalkbrener2005}.
For the more general case of convex risk measures we refer to  \cite{Tsanakas2009, McNFE15Text} and references therein.

Risk measures as we consider them here are mathematical tools which require as inputs probability distributions of the underlying risk factors. In practical applications one is typically confronted with the fact that these probability distributions are not fully specified. For example, let $X$  represent a P\&L, which is a function of some underlying risk factors, and let $\rho$ be the risk measure used to measure the riskiness of $X$, so that the desired quantity to compute is the risk $\rho(X)$. Since the probability laws of the risk factors are not fully specified, then one needs to approximate $\rho(X)$, perhaps by estimating this quantity exploiting historical data. As a consequence, the risk allocations,
which are usually computed in terms of risk measures, need to be approximated, in particular by estimation.

The problem of estimation of risk has, to a great extent, been  neglected in the literature. In the recent paper \cite{PiteraSchmidt2018}  a new statistical methodology for efficient estimation of risk capital $\rho(X)$ was proposed. The methodology introduced in that paper is based on the key concept, which the authors call unbiased estimation of risk also introduced in \cite{PiteraSchmidt2018}, and is based on economic principle.\footnote{The concept of unbiased estimation of risk must not be confused with the classical concept of unbiased estimator.}  Inspired by the ideas from \cite{PiteraSchmidt2018}, in this paper we develop a novel methodology for \emph{estimation of capital risk allocation}.\footnote{In this paper we will occasionally write capital allocation or risk allocation in place of capital risk allocation.} We work within a general, but tractable class of coherent risk measures, the so-called weighted value-at-risk measures introduced in \cite{Cherny2006}, with focus on the expected shortfall risk measure, which is broadly accepted in the risk management practice.

The underlying key concept introduced in this paper is the \textit{fair capital risk allocation}, which builds upon the robust representation of coherent risk measures.  Our concept of fairness aligns well with what has been done in some of the existing literature. In particular, it implies fairness in the sense of fuzzy games  introduced in \cite{Delbaen2000}. The fair capital risk allocation can be also viewed as version of the Euler principle of risk  allocation. The fair allocation principle used here has been also applied in \cite{BCF2018} in the context of allocation of the total default fund among the clearing  members of a CCP. For additional insight about fair risk allocation we refer to the recent work~\cite{CoculescuDelbaen2019}. We provide explicit formulae for fair capital allocations in case when the constituents of the portfolio are jointly normally distributed.

The major focus of the paper is  on the problem of approximating fair portfolio allocations when the law of the portfolio constituents is not fully known.  Motivated by the concept of the fair capital allocation, we define and study the concepts of \textit{fair allocation estimators} and \textit{asymptotically fair allocation estimators}. A substantial portion of our study is devoted to the problem of estimating the risk allocation under expected shortfall and normality. In addition we consider a nonparametric approach to this problem. We derive several estimators, and prove their fairness and/or asymptotic fairness. Last, but not least, we propose two backtesting methodologies that are oriented at assessing the performance of the allocation estimation procedure. Finally, we perform relevant numerical studies. The results of the numerical studies that we have conducted so far are encouraging for practical use of the estimation and backtesting of the capital allocation.

This work is a first step towards developing formal methodologies for estimating and backtesting of fair capital allocation. As such, it has potential to open new theoretical and practical research avenues.

\section{The fair allocation principle}
Let $(\Omega,\sF,\bP)$ be an atomless probability space, and let $\bE$ be the expectation under $\bP$.  In what follows, all needed integrability and regularity assumptions are taken for granted.

We consider a random vector $X=(X_1,\ldots,X_d)$ whose components are interpreted as discounted future profits and losses (P\&Ls). The marginal random variable $X_i$ (margin -- for short) might correspond to the $i$th clearing member of a central clearing counterparty (CCP), to the $i$th position in the portfolio, to the $i$th trader portfolio in a trading desk, or to the $i$th desk in the financial institution portfolio. In the following, we will refer to $X$ as \emph{portfolio} and to $X_i$ as the $i$th \textit{portfolio margin} or the $i$th \textit{portfolio constituent}.

Let  $L^1:=L^1(\Omega,\cF,\bP)$ and let $\rho: L^{1}\to \bR \cup \{+\infty\}$ be a  \textit{normalized monetary risk measure}. That is: $\rho$ is  \emph{monotone}, i.e.~$\rho(U)\leq \rho(V)$ for all $U,V\in L^1$ such that $U\geq V$; $\rho$ is \emph{cash-additive}, i.e.~$\rho(U+c)=\rho(U)-c$ for all $c \in \bR$ and all $U\in L^1$; $\rho$ is \emph{normalized}, i.e.~$\rho(0)=0$.

The riskiness of the portfolio $X$ is measured by applying the risk measure $\rho$ to the \textsl{aggregated  portfolio P\&L} denoted by
\[
S:=\sum_{i=1}^{d}X_i.
\]
We call the quantity $\rho(S)$ the \emph{aggregated risk}, or total risk, of the portfolio $X$.

Our objective is to study the issue of allocating  the aggregated risk of the portfolio to the individual constituents of the portfolio. Specifically, we intend to find
a vector $a=(a_1,\ldots,a_d)\in \bR^{d}$, called a \emph{risk allocation}, such that the following \textit{balance condition} holds
\begin{equation}\label{eq:1}
\rho(S)=\sum_{i=1}^d a_i.
\end{equation}
The component $a_i$ is interpreted as the risk contribution of $X_i$ to the aggregated risk, and therefore $X_i+a_i$ is interpreted as the $i$th \textit{secured margin} of portfolio $X$. Correspondingly, we call $X+a$ the {\it secured portfolio}, and $S+ \sum_{i=1}^{d}a_i$ the \textit{secured aggregated position}.

Stated as such, the risk allocation problem is ill--posed. Indeed, any collection of numbers $a_1,\ldots,a_d$ satisfying the balance condition \eqref{eq:1} constitutes a risk allocation. In order to deal with a meaningful risk allocation problem we need to impose additional conditions, that reflect some additional and desired features of the portfolio allocation. With this in mind, we impose an additional condition on $a$, which we will call the \textit{fairness condition}.

Towards this end, we require more structure on the risk measure $\rho$. We additionally assume that  the monetary risk measure $\rho$ is  finite, law-invariant,  comonotonic  and coherent; see~\cite{Kusuoka2001} for details. In view of \cite[Theorem 2(iii)]{Shapiro2013} we conclude that $\rho$ is a  \emph{weighted value-at-risk measure},\footnote{Following the traditional nomenclature, we use the name `weighted value-at-risk measure', although a more appropriate name would be `weighted expected shortfall'.} so that it admits representation (1.1) in \cite{Cherny2006} for a  fixed probability measure $\nu$  on $[0,1]$. Specifically, for a continuously distributed random variable $Y$,
\begin{equation}\label{eq:rho}
\rho(Y)=\rho_\nu (Y):=\int_{[0,1]}\textrm{ES}_\alpha(Y)\nu(d\alpha), \quad Y\in L^1,
\end{equation}
where $\textrm{ES}_\alpha$ is the \emph{Expected Shortfall}\footnote{For a formal definition of expected shortfall in the context of this paper see \eqref{def:ex}.} (ES) risk measure (sometimes also called {\it tail value-at-risk} or {\it conditional value-at-risk}) for reference level $\alpha \in [0,1]$. Moreover,  $\rho$ admits a robust-type representation of the form 
\begin{equation}\label{eq:robust}
\rho(Y)=\sup_{\bQ\in \cD}\bE_{\bQ}[-Y],
\end{equation}
where $\cD$ is a determining family of probability measures absolutely continuous with respect to $\bP$. As shown in~\cite[Theorem 6.3]{Cherny2006}, for any $Y\in L^1$ there exists a unique {\it minimal extreme measure} $\bQ_Y\in \cD$ such that\footnote{Note that the set of {\it extreme measures}, i.e. the set of measures that satisfy \eqref{eq:maximiser}, might contain more than one element. The term {\it minimal} corresponds to the minimal element with respect to the convex stochastic order; see~\cite{Cherny2006} for details.}
\begin{equation}\label{eq:maximiser}
\rho(Y)=\bE_{\bQ_Y}[-Y].
\end{equation}
Sometimes, we refer to $\bQ_Y$ as  {\it the worst-case scenario measure} (for position $Y$). We denote by $Z_Y$ the associated Radon-Nikodym derivative $d\bQ_Y/d\bP$. In particular, as shown in \cite{Cherny2006} (cf. formula (6.2) there), if $Y$ has a continuous distribution then we have
\begin{equation}\label{eq:g(Y)}
Z_Y=g(Y), \quad \textrm{and} \quad \rho(Y) = \bE[-g(Y)Y],
\end{equation}
for some Borel function $g$. For example if $\rho=\textrm{ES}_\alpha$ is the expected shortfall at level $\alpha$, then we have
\begin{equation}\label{ZY}
Z_Y=\frac{1}{\alpha}\1_{\{Y<q_Y(\alpha)\}},
\end{equation} where $q_Y(\alpha)$ is the $\alpha$--quantile of $Y$.

In what follows, for simplicity, we write $\bE_S$ instead of $\bE_{\bQ_S}$. The value $\bE_{S}\left[X_i +a_i\right]$  represents the average performance of the secured margin $X_i +a_i$ under the extremal measure $\bQ_S$.
The following \emph{fairness condition} selects risk allocations which are comparable under the extremal measure of the aggregated portfolio P\&L.

\begin{definition}\label{def:fair}
The capital allocation $a=(a_1,\dots,a_d)$ is called \emph{fair}, if
\begin{equation}\label{eq:fair.th}
\bE_S\left[X_i +a_i\right]=\bE_S\left[X_j +a_j\right], \quad i,j=1,\ldots,d.
\end{equation}
\end{definition}

\bigskip

The economic intuition behind this definition is as follows: the worst-case-scenario $\bQ_S$ is, in our setting, the determining scenario of the capital allocation for the portfolio through $\rho(S)=\bE_S[-S]$ resulting from Equation \eqref{eq:maximiser}. A \emph{fair} capital allocation is meant to  create secured positions $X_i+a_i$, $1 \le i \le d$,  so that the averages  of all secured positions with respect to the worst-case-scenario $\bQ_S$  are all equal.

Since $\rho$ is a monetary risk measure, the extremal measures for $S$ and $S+c$, $c \in \bR$, coincide. Thus, for any fair capital allocation $a$ satisfying the balance condition in \eqref{eq:1} we have
\begin{align}\label{density}
    0 = \rho \bigg(\sum_{i=1}^{d}(X_i +a_i)\bigg)
      = -\bE_S\bigg[\sum_{i=1}^{d}(X_i +a_i)\bigg]
      = -\sum_{i=1}^{d}\bE_S\left[X_i +a_i\right],
\end{align}
and consequently the risk allocations are given by
\begin{equation}\label{Eq:ai}
    a_i=-\bE_{S}[X_i]= -\bE[Z_S X_i], \quad i=1,\ldots,d.
\end{equation}
In view of \eqref{eq:g(Y)}, we also have that
\begin{equation}\label{eq:aig}
  a_i = -\bE\big[g\big(\sum_{k=1}^{d}X_k\big)X_i\big], \quad i=1,\ldots,d.
\end{equation}

First, we note that the fair risk allocation is unique, which is due to the existence and 	uniqueness of the extreme measure $\bQ_S$.
Secondly, we also note that the concept of fairness introduced in Definition \ref{def:fair} is actually equivalent to the concept of \emph{Euler risk allocation}. This observation is readily demonstrated by \eqref{Eq:ai}. However, it is the characterization of the fairness property of risk allocation as presented in \eqref{eq:fair.th} that underlies the notion of fair allocation estimator given in Definition \ref{def:def}, which is the key definition in this paper. That is why we defined fairness of risk allocation via \eqref{eq:fair.th} rather than via  \eqref{Eq:ai}.

We also note that the above notion of fairness implies \emph{fairness in the sense of fuzzy games} introduced in \cite{Delbaen2000}. Indeed, this follows from Theorems 17 and 18 therein taking representation \eqref{eq:robust} into account. The fair allocation principle of Definition~\ref{def:fair} has been applied in \cite{BCF2018} in the context of allocation of the total default fund among the clearing  members of a CCP.

The following example illustrates the concept of fair allocation.

\begin{example}[Mean risk allocation]\label{ex:mean}
Consider expectation for measuring risk, i.e.~$\rho(Y)=\bE[-Y]$, in which case $\mathcal{D}=\{\bP\}$. Then, clearly, for any $X=(X_1,\ldots,X_d)$, the capital allocation $\mathbf{a}=(a_1,\ldots,a_d)$ given as
\[
a_i= -\bE[X_i], \quad i=1,\ldots,d,
\]
is fair.
\end{example}

\subsection{Risk allocation under normality}\label{S:normal}

As an  example where explicit formulae can be obtained, we study the case of normally distributed profits and losses. In this regard, let us assume that the vector $X$ is normally  distributed under $\bP$ with mean $\mu$ and covariance matrix $\Sigma$ and fix $i \in \{1,\dots,d\}$. Then, $(X_i,S)$ is bivariate normal, and the conditional expectation $\bE[X_i|S]$ takes the form
\[
\bE[X_i|S]=\beta_i S +\alpha_i,
\]
with $ \beta_i= \frac{\Cov(X_i,S)}{\Var(S)}$,  and $\alpha_i = \mu_i -\beta_i\sum_{j=1}^d\mu_j$. Since this conditional expectation is the $L^2:=L^2(\Omega,\cF,\bP)$ orthogonal projection of $X_i$ on the linear space spanned by $S$ we obtain
\[X_i=\beta_i S +\alpha_i+\epsilon_i,\]
where $S$ and $\epsilon_i$ are independent under $\bP$, and $\bE [\epsilon_i]=0$.
For any weighted value-at-risk measure $\rho$, Equation \eqref{Eq:ai} implies that a fair capital allocation is given by
\begin{align}\label{eq:ai.normal}
a_i &= -\bE_{S}[X_i] = -\alpha_i - \beta_i \bE_S[S] - \bE_S[\varepsilon_i] \nonumber\\
    &= - \alpha_i + \beta_i \rho(S) - \bE[Z_S\varepsilon_i] \nonumber \\
    &= - \alpha_i + \beta_i \rho(S) - \bE[g(S)\varepsilon_i]=- \alpha_i + \beta_i \rho(S) - \bE[g(S)]\bE[\varepsilon_i]\nonumber \\
    &=- \alpha_i + \beta_i \rho(S),
\end{align}
where we have used \eqref{eq:g(Y)} in the fourth  equality, independence of $S$ and $\epsilon_i$  under $\bP$ in the fifth equality, and the fact that $\epsilon_i$ has zero mean under $\bP$, in the last equality.
As expected, the total allocated risk is divided among constituents using the regression slope allocations which is typically referred to as the \emph{covariance principle}, see \cite[Section 8.5]{McNFE15Text}.

\medskip \noindent
\textbf{Expected shortfall.}
To be more specific, we consider as an important example the expected shortfall (ES). In this regard, let $\rho=\textrm{ES}_\alpha$  denote ES  under $\bP$ for the level $\alpha \in (0,1)$. Then, for a continuously distributed real valued random variable $Y$  we have
\begin{align}\label{def:ex}
\textrm{ES}_\alpha(Y)=\bE[-Y \mid  Y\leq q_{Y}(\alpha)],
\end{align}
where $q_{Y}(\alpha)$ is an $\alpha$-quantile of $Y$.  Thus, since $S$ is normally  distributed,  \eqref{def:ex} yields
\begin{equation}\label{eq:rhoS}
\textrm{ES}_\alpha\left(S\right)= -\sum_{i=1}^d\mu_i + \frac{1}{\alpha}\sqrt{\Var(S)}\, \phi\big(\Phi^{-1}(\alpha)\big),
\end{equation}
where $\phi$ and $\Phi$ are the density and the cumulative distribution function of the standard normal distribution; see   \cite[Example~2.14]{McNFE15Text}.
Putting together \eqref{eq:ai.normal} and \eqref{eq:rhoS} we see that the capital allocation for ES is given as
\begin{equation}\label{eq:norm.es.true}
a_i= -\mu_i + \frac{\Cov(X_i,S)}{\alpha\sqrt{\Var(S)}}\phi(\Phi^{-1}(\alpha)),\quad\quad\quad i=1,2,\ldots,d\,.
\end{equation}

\section{Fair allocation estimators}
In practice, the  probability distribution under $\bP$ of $X$, the portfolio's P\&L, is not fully specified. Since, in view of \eqref{eq:g(Y)} and \eqref{eq:aig}, we have
\begin{equation}\label{eq:aig-new}
 \rho(S)=-\bE\bigg[g\Big(\sum_{k=1}^d X_k\Big)\sum_{k=1}^d X_k\bigg],\ \ \textrm{and}\ \ a_i = -\bE\bigg[g\Big(\sum_{k=1}^{d}X_k\Big)X_i\bigg], \quad i=1,\ldots,d,
\end{equation}
then, in almost all practically relevant applications, neither the aggregated risk $\rho(S)$ nor the fair risk allocation $a$ are  known, and thus need to be estimated.  Hence, appropriate estimation procedures have to be developed, in particular estimation procedures based on the historical data about realizations of the portfolio. This will involve estimating, in some  way, the probability distribution of $X$ under $\bP$.

In the following, we set the relevant statistical framework and propose efficient procedures to deal with this estimation issue. We refer to $X$ as to the \emph{population}. Historical information about $X$ is given in terms of a random sample of size $n$ drawn from $X$, which we denote by $X^1,\dots,X^n$, so that  $X^1,\dots,X^n$ are independently drawn copies of the random variable $X$. Our aim is to estimate the aggregated risk $\rho(S)$ using the information contained in  the sample. Towards this end we let
\[
\mathbf{X}^n:=\set{X^j= (X^j_1,\ldots,X^j_d),\ j=1,\ldots,n},
\]
represent the random sample,  and let us denote its realization by
\begin{equation}
\mathbf{x}^n:=\set{x^j= (x^j_1,\ldots,x^j_d),\ j=1,\ldots,n},
\end{equation}
where $x^j_k$ corresponds to the $j$-th observed (realized) value of the portfolio's $k$th margin.

The formal statistical setup for this situation is as follows: consider a family of probability measures $\cP:=(\bP^\theta)_{\theta\in\Theta}$ on $(\Omega,\sF)$, where $\Theta$ denotes the parameter space.  To avoid unnecessary technical difficulties, we assume that all measures in $\cP$ are equivalent.
Furthermore, we assume that for any $\theta\in\Theta$ the random sample $X^1,\dots,X^{n}$ is i.i.d.~under $\bP^{\theta}$. Moreover, we assume that $\bP =\bP^{\theta_0}$ for some (unknown) parameter $\theta_0\in\Theta$. We will denote by $\rho^\theta$ and, respectively $\bE^\theta$, the  risk measure $\rho$, and respectively the expectation, under the probability measure $\bP^\theta$. Similarly to the notation $\bQ_Y$ and $Z_Y$, corresponding to the reference measure $\bP$, we will use notation  $\bQ^\theta_Y$ and $Z^\theta_Y$ with regard to the reference measure $\bP^\theta$.

Given the random sample $\mathbf{X}^n$, the allocation $a$ is estimated using
an \emph{allocation estimator} $\hat A^n=(\hat A^n_1,\ldots,\hat A^n_d)$ defined as
\begin{equation}\label{estimator}
  \hat A^n = \eta_n (\mathbf{X}^n),
\end{equation}
 for some measurable function $\eta_n:\bR^{d\times n}\to \bR^d$.  

Next, we define a property that should be satisfied by any reasonable \emph{allocation estimator}.

\begin{definition}\label{def:def}
An allocation estimator $\hat A^n$ is called {\it fair} if, for all $\theta  \in \Theta$,
\begin{equation}\label{eq:fair}
\bE^\theta\Big[ Z_{S,\hat A^n}^\theta (X_i +\hat A^n_i)\Big] = 0,\quad i=1,\ldots,d,
\end{equation}
where $Z^{\theta}_{S,\hat A^n} := Z^{\theta}_{S+\sum_{i=1}^{d}\hat A^n_i}$.
\end{definition}
We emphasize that $\hat{A}^n$ is a random variable, and $Z^{\theta}_{S,\hat A^n}$ is the Radon-Nikodym derivative corresponding to  $S+\sum_{i=1}^{d}\hat A^n_i$.

We stress that the definition of the fair allocation estimator requires that property \eqref{eq:fair} is satisfied for all populations  from the population space $\Theta$, that is for all $\theta  \in \Theta$.

Intuitively,  the above definition means that  an allocation estimator is fair if it mimics the balanced fairness condition \eqref{Eq:ai} for all relevant scenarios (given by probability distributions $\bP^\theta, \ \theta \in \Theta$). In particular, the aggregated risk estimator obtained from a  {\it fair} allocation estimator $\hat A$ by summation turns out to be  {\it unbiased} in the sense of \cite[Definition 4.1]{PiteraSchmidt2018}, namely,  for any $\theta\in\Theta$ we get
\begin{equation}\label{eq:unbiased}
    \rho^\theta\Big(S + \sum_{i=1}^d\hat A^n_i\Big)
         = -\sum_{i=1}^{d}\bE^{\theta}\Big[ Z^\theta_{S,\hat A^n} (X_i +\hat A^n_i)\Big]
         = 0.
\end{equation}

 Equality~\eqref{eq:unbiased} guarantees that the secured aggregated portfolio position $S+\sum_{i=1}^d\hat A_i$ is acceptable in the sense that it bears no risk, while Equality~\eqref{eq:fair}  ensures that the average performance of the secured marginal positions under the worst-case scenario measure for the secured portfolio $S$ are the same and that the joint position is secured. In particular, for $d=1$, the definitions of {\it fairness} and {\it unbiasedness} coincide.

It should be noted that  \eqref{eq:unbiased} means that a fair allocation estimator charges an adequate amount of capital to secure the portfolio. This is a consequence of  \eqref{eq:fair}, which means that a fair allocation estimator applies an adequate amount of capital charge to each position constituent.

We end this section with a simple example to illustrate the concept of fairness.

\begin{example}\label{ex:1-0}
Consider the mean risk allocation given in Example~\ref{ex:mean}. This leads to the  family of risk measures $\rho^{\theta}(\cdot)=-\bE^{\theta}[\,\cdot\,]$, $\theta\in\Theta$.
Then,  the risk allocation estimator
\[
\hat M^{n}_i=-\frac 1n \sum_{j=1}^{n}X^j_i,\qquad \textrm{for } i=1,2,\ldots,d\,,
\]
is a fair allocation estimator. Indeed, note that here, for each $\theta\in\Theta$, the extremal measure coincides with the original probability measure $\bP^\theta$, i.e. $Z_{S,\hat M^{n}}^\theta\equiv 1$. Thus, for $i\in\{1,2,\ldots d\}$ we obtain
\[
\bE^{\theta}\Big[ Z_{S,\hat M^{n}}^\theta (X_i +\hat M^{n}_i)\Big]=
\bE^{\theta}\Big[X_i -\frac 1n \sum_{j=1}^n X_i^j\Big]=0.
\]
\end{example}

\subsection{Estimating  capital allocation under expected shortfall and normality}

Following Section~\ref{S:normal}, we study the case where  the $d$-dimensional random vector $X$ is normally distributed under every $\bP^{\theta}$, and we assume that the risk is measured by the expected shortfall $\textrm{ES}^\theta_\alpha$, at a fixed level $\alpha\in (0,1)$.
In what follows, for the random sample $\mathbf{X}^n$, we will use the notation   $S^j  := \sum_{i=1}^d X_i^j$, $j=1,\ldots,n$, and we set\footnote{To ease  the notation, we will drop the superscript $n$ in the following. So, we will write  $\hat\mu_i$ rather than $\hat\mu_i^n$, etc.  }
\begin{align*}
\hat\mu_i & := \tfrac{1}{n}\textstyle \sum_{j=1}^n X_i^j,\\
\hat\mu_S & :=\textstyle \tfrac{1}{n}\sum_{j=1}^{n} S^j = \sum_{i=1}^{d}\hat\mu_i  \\
\hat\sigma_S^2 &:= \tfrac{1}{n}\textstyle \sum_{j=1}^n(S^j - \hat \mu_S)^2,\\
{\widehat \Cov}_{X_i,S} &:=\tfrac{1}{n}\textstyle \sum_{j=1}^n (X_i^j - \hat \mu_i)(S^j - \hat \mu_S),
\end{align*}
to denote the sample mean of the $i$th constituent, the sample mean of the portfolio, the sample variance of the portfolio, and the sample covariance of the $i$th constituent and the portfolio, respectively.

Motivated by the Representation~\eqref{eq:ai.normal} we define the allocation estimator $\hat B=(\hat B_1,\ldots,\hat B_d)$ as
\begin{equation}\label{eq:ai.normal3}
\hat B_i:= -\hat \alpha_i+ \hat \beta_i\,\hat R(S),\qquad i=1,\ldots,d\,,
\end{equation}
where
$ \hat\beta_i= \tfrac{1}{{\hat\sigma}_S^2}{\widehat\Cov}_{X_i,S}$\, and\, $\hat\alpha_i =\hat\mu_i-\hat\beta_i\hat\mu_S$\, are the estimators of the slope and intercept regression coefficient from the $L^2$ orthogonal projection of the $i$th margin of $X$ onto $S$, and where $\hat R(S)$ is an unbiased risk estimator (in the sense of \cite{PiteraSchmidt2018}) for the Expected Shortfall of the secured position $S$.
It has been shown in \cite[Example 5.4]{PiteraSchmidt2018} that $\hat R(S)$ under normality can be represented as \begin{equation} \label{eq:hatR}
\hat R(S) =-\hat\mu_S+\hat {\sigma}_S b_n,
\end{equation}
where $b_n\in\mathbb{R}$ is deterministic, and depends only on the sample size $n$, and risk level $\alpha\in (0,1)$.
Consequently, the estimator becomes
$$
\hat B_i = -\hat \mu_i + \frac{\widehat \Cov_{X_i,S} }{\hat \sigma_S}b_n,\qquad i=1,\ldots,d.
$$

Before we show that $\hat B$ satisfies the fairness property, we show an important conditional un\-biasedness  property  of the estimators  $ \hat\beta_i$ and $\hat\alpha_i $, in the usual statistical sense.
Towards this end, for $i=1,2,\ldots,d$, we use
\begin{align*}
\beta_i^\theta &:={\Cov^\theta(X_i,S)} \cdot (\Var^\theta(S))^{-1},\\
\alpha_i^\theta & := \bE^\theta(X_i) -\beta^\theta_i\sum_{k=1}^d\bE^\theta(X_k),
\end{align*}
to denote the true regression coefficients of the $L^2$--orthogonal projection of $i$th margin of $X$ onto $S$ under $\bP^\theta$, for $\theta\in\Theta$; see Section~\ref{S:normal}. Note that, in view of our assumption that for any $\theta\in\Theta$ the random sample $X^1,\dots,X^{n}$ is i.i.d.~under $\bP^{\theta}$, we get $\beta_i^\theta={\Cov^\theta(X^j_i,S^j)} \cdot (\Var^\theta(S^j))^{-1}$ and $\alpha_i^\theta= \bE^\theta(X^j_i) -\beta^\theta_i\sum_{k=1}^d\bE^\theta(X^j_k)$, for $j=1,\ldots, n$.

\begin{proposition}\label{prop2.3}
For any $\theta\in\Theta$ it holds that
\begin{equation}\label{eq:unbiased.ab}
\bE^\theta\Big[\hat\beta_i\,|\, \hat\mu_S,\hat\sigma_S\Big] =\beta_i^\theta\qquad \textrm{ and }\qquad \bE^\theta\Big[\hat\alpha_i\,|\, \hat\mu_S,\hat\sigma_S\Big] =\alpha_i^\theta, \qquad i=1,\dots,d.
\end{equation}
\end{proposition}
\begin{proof}
Recall from Section \ref{S:normal} that under normality, for $j\in \{1,\ldots, n\}$, $i\in \{1,\ldots,d\}$, and $\theta\in\Theta$, we have
\begin{equation}\label{eq:residual.k}
X^j_i=\alpha^\theta_i+\beta^\theta_i S^j+\epsilon^{j,\theta}_i,
\end{equation}
where $\epsilon_i^{j,\theta}$ is a zero mean Gaussian random variable independent of $S^j$.
As a simple consequence of \eqref{eq:residual.k} we obtain that $\epsilon_i^{j,\theta}$ is independent of $\mu_S$ and $\sigma_S$ under $\bP^\theta$ for all $\theta \in \Theta$.
Then, by definition,
\begin{align}
  \cI_1:=  \bE^\theta\bigg[\hat\beta_i\,|\, \hat\mu_S,\hat\sigma_S\bigg] \label{temp:387}
       & = \tfrac{1}{\hat\sigma_S^2}\bE^\theta\bigg[\frac{1}{n}\sum_{j=1}^{n}(X^j_i-\hat\mu_i)(S^j-\hat\mu_S) \,|\, \hat\mu_S,\hat\sigma_S\bigg]\\
       & = \tfrac{1}{\hat\sigma_S^2}\bE^\theta\bigg[\frac{1}{n}\sum_{j=1}^{n}X^j_i S^j-\hat\mu_i\hat\mu_S \,|\, \hat\mu_S,\hat\sigma_S\bigg]. \nonumber
\end{align}
Inserting \eqref{eq:residual.k}, and using that $n^{-1}\sum_{j=1}^n (S^j)^2 = \hat \sigma_S^2 + \hat \mu_S^2$, we obtain
\begin{align*}
\cI_1
& = \tfrac{1}{\hat\sigma_S^2}\bE^\theta\bigg[\frac{1}{n}\sum_{j=1}^{n}(\alpha_i^\theta +\beta_i^\theta S^j+\epsilon_i^{j,\theta})S^j-\hat\mu_i\hat\mu_S \,|\, \hat\mu_S,\hat\sigma_S\bigg]\\
&=\tfrac{1}{\hat\sigma_S^2}\bE^\theta\bigg[\alpha_i^\theta \hat\mu_S +\beta_i^\theta (\hat\sigma_S^2+\hat\mu_S^2)+\frac{1}{n}\sum_{j=1}^{n}\epsilon^{j,\theta}_i S^j-\hat\mu_i\hat\mu_S \,|\, \hat\mu_S,\hat\sigma_S\bigg]\\
&{=\tfrac{1}{\hat\sigma_S^2}\bE^\theta\bigg[\alpha_i^\theta \hat\mu_S +\beta_i^\theta (\hat\sigma_S^2+\hat\mu_S^2)+\frac{1}{n}\sum_{j=1}^{n}\bE^\theta\big[\epsilon^{j,\theta}_i S^j\,|\, S^j,\hat\mu_S,\hat\sigma_S\big]-\hat\mu_i\hat\mu_S \,|\, \hat\mu_S,\hat\sigma_S\bigg]}\\
&{=\tfrac{1}{\hat\sigma_S^2}\bE^\theta\bigg[\alpha_i^\theta \hat\mu_S +\beta_i^\theta (\hat\sigma_S^2+\hat\mu_S^2)+\frac{1}{n}\sum_{j=1}^{n} S^j\bE^\theta\big[\epsilon^{j,\theta}_i\,|\, S^j,\hat\mu_S,\hat\sigma_S\big]-\hat\mu_i\hat\mu_S \,|\, \hat\mu_S,\hat\sigma_S\bigg]}\\
&{=\tfrac{1}{\hat\sigma_S^2}\bE^\theta\bigg[\alpha_i^\theta \hat\mu_S +\beta_i^\theta (\hat\sigma_S^2+\hat\mu_S^2)+\frac{1}{n}\sum_{j=1}^{n} S^j\bE^\theta\big[\epsilon^{j,\theta}_i\big]-\hat\mu_i\hat\mu_S \,|\, \hat\mu_S,\hat\sigma_S\bigg]}\\
&=\tfrac{1}{\hat\sigma_S^2}\bE^\theta\bigg[\alpha_i^\theta \hat\mu_S +\beta_i^\theta (\hat\sigma^2_S+\hat\mu_S^2)-\hat\mu_i\hat\mu_S \,|\, \hat\mu_S,\hat\sigma_S\bigg]
=\beta_i^\theta+\tfrac{\hat\mu_S}{\hat\sigma_S^2} \, \bE^\theta\bigg[\alpha_i^\theta +\beta_i^\theta \hat\mu_S-\hat\mu_i \,|\, \hat\mu_S,\hat\sigma_S\bigg].
\end{align*}
We use again \eqref{eq:residual.k} and obtain
\begin{align}\label{eq:mui}
    \hat \mu_i = \frac 1 n \sum_{j=1}^n X_i^j = \alpha_i^\theta + \beta_i^\theta \frac 1 n \sum_{j=1}^n S^j + \eta^\theta,
\end{align}
{with  $\eta^\theta= \frac{1}{n}\sum_{j=1}^n \epsilon_i^{j,\theta}$ satisfying $\bE^\theta[\eta^\theta \,|\, \hat\mu_S,\hat\sigma_S]=0$, so that}
\begin{align}
    \bE^\theta\big[ \hat \mu_i \,|\, \hat\mu_S,\hat\sigma_S \big] = \alpha_i^\theta +\beta_i^\theta \hat\mu_S,
\end{align}
and hence $\cI_1 = \beta_i^\theta$ yielding our first claim. With this result and using \eqref{eq:mui}, we obtain
\begin{align*}
\bE^\theta\left[\hat\alpha_i\,|\, \hat\mu_S,\hat\sigma_S\right] & =\bE^\theta\left[\hat\mu_i-\hat\beta_i\hat\mu_S\,|\, \hat\mu_S,\hat\sigma_S\right]\\
& =\bE^\theta\left[\hat\mu_i-\beta^\theta_i\hat\mu_S\,|\, \hat\mu_S,\hat\sigma_S\right]
 =\alpha_i^\theta
\end{align*}
which concludes the proof of~\eqref{eq:unbiased.ab}. \end{proof}

Proposition \ref{prop2.3} shows that we can estimate the portfolio risk expressed through $\hat\mu_S$ and $\hat\sigma_S$ without impacting the statistical unbiasedness property of the regression coefficients; cf. Equation \eqref{eq:hatR}. Consequently, the risk allocation estimation procedure could be split into two independent steps. First, we estimate the aggregated portfolio risk, and then we estimate the proper allocation of the risk within portfolio constituents. Now, we use this property to show that the allocation estimator given in~\eqref{eq:ai.normal3} satisfies the fairness property.

\begin{theorem}\label{th:normal.fair}
Assume that the allocation estimator $\hat B=(\hat B_1,\dots,\hat B_d)$ is given by \eqref{eq:ai.normal3} with $\hat R(S)$ as in \eqref{eq:hatR}. Then, the capital allocation $\hat B$ is fair.
\end{theorem}

\begin{proof} In what follows we will simply write $\hat R$ instead of $\hat R(S)$.
 We note that for any $\theta \in \Theta$ the Radon-Nikodym density $Z^\theta_{S,\hat B}$ is $\sigma(S+\sum_{i=1}^d \hat B_i)$-measurable; see \cite[Proposition~6.2]{Cherny2006} and recall that
$\hat B_i = - \hat \alpha_i + \hat \beta_i \hat R$. Moreover, since
\begin{align*}
    \sum_{i=1}^d \hat \beta_i &= \frac{1}{\hat \sigma_S^2} {\sum_{i=1}^d {\widehat \Cov}_{ X_i, S}} = \frac{\hat \sigma_S^2}{\hat \sigma_S^2} =1
\end{align*}
we  obtain that
\begin{align*}
    \sum_{i=1}^d \hat \alpha_i &= \sum_{i=1}^d \hat \mu_i - \hat \mu_S \cdot \sum_{i=1}^d \hat \beta_i = \hat \mu_S - \hat \mu_S = 0.
\end{align*}
Consequently, as expected,
\begin{align} \label{sumhatAi}
    \sum_{i=1}^d \hat B_i = \hat R
\end{align}
and Equation \eqref{eq:hatR} yields that $Z^\theta_{S,\hat B}$ is $\sigma(\hat \mu_S,\hat \sigma_S,S)$-measurable.
With a view towards \eqref{eq:fair},  we compute
\begin{align*}
    \bE^\theta\Big[ Z^\theta_{S,\hat B} \hat \alpha_i \Big] &= \bE^\theta\Big[ Z^\theta_{S,\hat B}  \bE^\theta[ \hat \alpha_i \,|\, \hat\mu_S,\hat\sigma_S,S] \Big] \\
    &=  \bE^\theta\Big[ Z^\theta_{S,\hat B}  \bE^\theta[ \hat \alpha_i \,|\, \hat\mu_S,\hat\sigma_S] \Big]
     =  \bE^\theta\Big[ Z^\theta_{S,\hat B}   \alpha_i^\theta \Big],
\end{align*}
by Proposition \ref{prop2.3}. Analogously,
\begin{align*}
    \bE^\theta\Big[ Z^\theta_{S,\hat B} \hat \beta_i \Big] &=
       \bE^\theta\Big[ Z^\theta_{S,\hat B}   \beta_i^\theta \Big]
\end{align*}
and we obtain
\begin{align}
\bE^\theta\left[Z^\theta_{S,\hat B}\left(X_i +\hat B_i\right)\right] &=\bE^\theta\left[Z^\theta_{S,\hat B}\left(X_i -\hat\alpha_i+\hat\beta_i\hat R\right)\right]\nonumber\\
& =\bE^\theta\left[Z^\theta_{S,\hat B}\left(X_i -\alpha^\theta_i+\beta^\theta_i\hat R\right)\right].\label{eq:1.un}
\end{align}
Next, using \eqref{eq:unbiased} and \eqref{sumhatAi} yields that
\begin{equation}\label{eq:regr.s}
    0 = \bE^\theta\bigg[Z^\theta_{S,\hat B}\Big(S+\sum_{i=1}^{d}\hat B_i \Big)\bigg]=\bE^\theta\bigg[Z^\theta_{S,\hat B}\Big(S+\hat R\Big)\bigg].
\end{equation}
This result, together with representation \eqref{eq:residual.k} for $j=n+1$, and letting $X^{n+1}=X$, imply  that
\begin{align}
\eqref{eq:1.un}
    &= \bE^\theta\left[Z^\theta_{S,\hat B}\left(X_i -\alpha^\theta_i-\beta^\theta_iS\right)\right]
     =\bE^\theta\left[Z^\theta_{S,\hat B}\epsilon_i^\theta\right]
     =\bE^\theta\big[Z^\theta_{S,\hat B}\big]\bE^\theta\big[\epsilon_i^\theta\big]
     =0,\label{eq:ex2:normal}
\end{align}
where we used the fact that $(\epsilon_i^\theta,S)$ is bivariate normal with uncorrelated margins, so that $\epsilon_i^\theta$ is independent of $S$, and consequently from $Z^\theta_{S,\hat B}$. This concludes the proof.
\end{proof}

\section{Asymptotic fairness}

We now introduce the definition of fairness for a sequence of estimators, $(\hat A^n)_{n\in\bN}$, and we define the notion of asymptotic fairness.

\begin{definition}\label{def:asymptotic}
 A sequence of allocation estimators $(\hat A^n)_{n\in\bN}$ will be called \emph{fair} at $n\in\bN$, if $\hat A^n$ is fair. If fairness holds for all $n\in\bN$, we call the sequence $(\hat A^n)_{n\in\bN}$ fair. The sequence $(\hat A^n)_{n\in\bN}$ is called \emph{asymptotically fair} if
\begin{equation}\label{eq:fair.asymptotic}
\bE^\theta\Big[ Z_{S,\hat A^n}^\theta (X_i +\hat A^n_i)\Big] \xrightarrow{\,n\to\infty\,} 0,\quad i=1,2,\ldots,d,\ \textrm{and}\ \theta  \in \Theta.
\end{equation}
\end{definition}

In view of Theorem~\ref{th:normal.fair} it is clear   that the sequence of capital allocation estimators $(\hat B^n)_{n\in\bN}$ defined in \eqref{eq:ai.normal3}, for varying $n$, is a fair sequence.\footnote{Recall that the superscript $n$ is omitted in \eqref{eq:ai.normal3} for the ease of notation.}

In the rest of the section we assume that the risk allocation is done using ES with reference level $\alpha$.

\subsection{Asymptotic fairness of capital allocation estimators under normality}

 Using \eqref{eq:norm.es.true}, we now define a sequence $\hat C^n=(\hat C_1^n,\ldots,\hat C_d^n), n\in\bN,$ of ``plug-in type'' capital allocation estimators as
\begin{equation}\label{eq:normal.plugin}
\hat C_i^n :=-\hat \mu_i + \frac{{\widehat \Cov}_{X_i,S}}{\alpha\hat\sigma_S}\phi(\Phi^{-1}(\alpha)).
\end{equation}
The sequence $(\hat C^n)_{n\in\bN}$ is not fair, in general, but it is asymptotically fair, as proven below.

\begin{proposition}\label{pr:norm.asympt}
The sequence $(\hat C^n)_{n\in\bN}$ is asymptotically fair.
\end{proposition}

\begin{proof}
Set $\hat F^n:=-\hat \mu_S + \hat {\sigma}_S\frac{\phi(\Phi^{-1}(\alpha))}{\alpha}$ and note that $\hat C_i^n =-\hat\alpha^n_i+\hat \beta^n_i \hat F^n,\  i=1,2,\ldots,d.$

Proceeding analogously  to  the proof of Theorem~\ref{th:normal.fair}, with $\hat B$ replaced by $\hat C^n$  and with $\hat R$ replaced by $\hat F^n$, we see that in order to prove proposition it is enough to show that for any $\theta\in\Theta$ we have
\begin{equation}
\bE^\theta\bigg[Z^\theta_{S+\hat F^n}\Big(S+\hat F^n\Big)\bigg] \xrightarrow{\,n\to\infty\,} 0.
\end{equation}
Now, note that
\[
\bE^\theta\bigg[Z^\theta_{S+\hat F^n}\Big(S+\hat F^n\Big)\bigg] =\rho_\theta(S+\hat F^n),
\]
and, in the terminology of  \cite{PiteraSchmidt2018},  $\hat F^n$ is the standard Gaussian expected shortfall plug-in estimator for $S$. Consequently, noting that for $d=1$ the definition of asymptotic fairness coincides with the definition of asymptotic unbiasedness given in \cite[Definition 6.1]{PiteraSchmidt2018}, and using \cite[Proposition 6.4]{PiteraSchmidt2018} we conclude the proof.
\end{proof}

\subsection{Asymptotic fairness of non-parametric capital allocation estimators}

We assume throughout this section that the population $X$, and hence the aggregated portfolio $S$, are continuous random variables under any $\theta\in\Theta$. Given that the ES is used to determine the risk allocation, and taking \eqref{ZY} and \eqref{Eq:ai} into account, we consider two natural non-parametric expected shortfall capital allocation estimators
\begin{align}
\check D^n_i&:=-\frac{\sum_{k=1}^{n}X_i^k\1_{\{S^k+\hat{\var}^n_{\alpha}\leq 0\}}}{n\alpha},\quad i=1,\ldots,d,\\
\hat D^n_i&:=-\frac{\sum_{k=1}^{n}X_i^k\1_{\{S^k+\hat{\var}^n_{\alpha}\leq 0\}}}{\sum_{k=1}^{n}\1_{\{S^k+\hat{\var}^n_{\alpha}\leq 0\}}},\quad i=1,\ldots,d,
\end{align}
where $\hat\var^n_{\alpha}:=-S^{(\lfloor n\alpha\rfloor+1)}$, with $S^{(j)}$ denoting the $j$th order statistics, and $\lfloor z \rfloor$ denoting the largest integer less or equal than $z$.

\begin{proposition}\label{pr:nonpar.fair}
The sequences $(\hat D^n)_{n\in\bN}$ and $(\check D^n)_{n\in\bN}$ are asymptotically fair.
\end{proposition}

We will show only that $\hat D^n_i$ is asymptotically fair. The proof for $\check{D}_i^n$ follows by similar arguments.
Before we prove Proposition~\ref{pr:nonpar.fair}, let us introduce supplementary notation and a lemma that will be useful for the proof. For any $\theta\in\Theta$ we use $a^{\theta}=(a^{\theta}_1,\ldots,a^{\theta}_n)$ to denote the true expected shortfall allocation for $X$ under $\theta$ and so we have (cf.  \eqref{ZY})
\[
Z^\theta_{S,a^\theta}=\frac{1}{\alpha}\1_{\left\{S+\sum_{i=1}^{d}a^{\theta}_i\leq q^\theta_{S+\sum_{i=1}^{d}a^{\theta}_i}(\alpha)\right\}},
\]
where  $q^{\theta}_{S+\sum_{i=1}^{d}a^{\theta}_i}(\alpha)$ denotes the true $\alpha$-quantile of $S+\sum_{i=1}^{d}a^{\theta}_i$ under $\mathbb{P}^\theta$.
Similarly, we have
\[
Z^\theta_{S,\hat D}=\frac{1}{\alpha}\1_{\left\{S+\sum_{i=1}^{d}\hat D^n_i\leq q^\theta_{S+\sum_{i=1}^{d}\hat D^n_i}(\alpha)\right\}}.
\]

\begin{lemma}\label{lm:Z.ok}
For any $\theta\in\Theta$ we get $Z_{S,\hat D^n}^\theta\xrightarrow{\,\bP^{\theta}\,} Z^\theta_{S,a^\theta}$, as $n\to\infty$.
\end{lemma}
\begin{proof}
Let us fix $\theta\in\Theta$. For brevity we will use the notation $r :=\sum_{i=1}^{d}a^\theta$ and $R_n :=\sum_{i=1}^{d} \hat D^n_i$.
First,  using classical trimmed-mean convergence arguments (see e.g. \cite{Stigler1973}) we will show that
\begin{equation}\label{eq:Zto0.3}
R_n \xrightarrow{\,\bP^{\theta}\,} r,\quad n\to\infty.
\end{equation}
Let $I_n :=\sum_{k=1}^n \1_{\{S^k-q^{\theta}_{S}(\alpha) \leq 0\}}$ and $a_n:=\lfloor n\alpha\rfloor +1$, $n\in\bN$. Since $a_n=\sum_{k=1}^{n}\1_{\{S^k+\hat{\var}^n_{\alpha}\leq 0\}}$, we get
\begin{align*}
R_n & = -\frac{\sum_{k=1}^{n}S^k\1_{\{S^k+\hat{\var}^n_{\alpha}\leq 0\}}}{\sum_{k=1}^{n}\1_{\{S^k+\hat{\var}^n_{\alpha}\leq 0\}}}=-\frac{1}{a_n}\sum_{k=1}^{a_n}S^{(k)}=-\frac{1}{a_n}\sum_{k=1}^{I_n}S^{(k)}+\epsilon_n,
\end{align*}
where $\epsilon_n:=-\frac{1}{a_n}\left(\1_{\{a_n> I_n\}} \sum_{k=I_n+1}^{a_n}S^{(k)}-\1_{\{a_n < I_n\}} \sum_{k=a_n+1}^{I_n}S^{(k)}\right)$. Next, we will show that $\epsilon_n \xrightarrow{\,\mathbb{P}^\theta\,}0$.
Due to the consistency of the empirical quantiles, we have that $S^{(a_n)}\xrightarrow{\mathbb{P}^\theta}q^{\theta}_{S}(\alpha)$ and $S^{(I_n)}\xrightarrow{\mathbb{P}^\theta}q^{\theta}_{S}(\alpha)$, as $n\to \infty$. Hence, noting that
\begin{align*}
0 & \leq  \abs{\epsilon_n} \leq \abs{\frac{I_n-a_n}{a_n}} \max\left\{\abs{S^{(a_n)}},\abs{S^{(I_n)}}\right\},
\end{align*}
it is sufficient to prove that $\left| \frac{I_n-a_n}{a_n}\right|\xrightarrow{\,\mathbb{P}^\theta\,}0$.  For this, we observe that
$$
\frac{I_n-a_n}{a_n}=\frac{n}{a_n}\left(\frac{1}{n}I_n-\alpha\right)+\frac{n\alpha-a_n}{a_n}.
$$
Since $\lim_{n\to\infty}\frac{n\alpha-a_n}{a_n} =0$, $\lim_{n\to\infty}\frac{n}{a_n} =\frac{1}{\alpha}$, and, by the Law of Large Numbers, $\left(\tfrac{1}{n}I_n-\alpha\right)\xrightarrow{\,\mathbb{P}^\theta\,}0$, we have that $\epsilon_n \xrightarrow{\,\mathbb{P}^\theta\,}0$.
Also, by the Law of Large Numbers we get at once that
\[
\frac{1}{a_n}\sum_{k=1}^{I_n}S^{(k)}=\frac{n}{a_n}\left(\frac{1}{n}\sum_{k=1}^{n}S^{k}\1_{\left\{S^k \leq q^{\theta}_{S}(\alpha)\right\}}\right)\,\xrightarrow{\,\,\,\mathbb{P}^\theta\,\,\,}\, \frac{1}{\alpha}\bE\left[S\1_{\{S\leq q^{\theta}_{S}(\alpha)\}}\right] =-r,
\]
which concludes the proof of \eqref{eq:Zto0.3}.


Next, for a fixed $\epsilon \in (\frac{1}{\alpha},0)$, we get
\begin{align}
\bP^{\theta}\left[|Z_{S,\hat D^n}-Z^\theta_{S,a^\theta}|>\epsilon \right] &=\bP^{\theta}\left[|Z_{S,\hat D^n}-Z^\theta_{S,a^\theta}|\neq 0 \right]\nonumber\\
& = \bP^{\theta}\left[ \{S+R_n \leq q^{\theta}_{S+R_n}(\alpha)\}\cap \{S+r >q^{\theta}_{S+r}(\alpha)\} \right]\label{eq:Zto0.1}\\
&\quad +\bP^{\theta}\left[ \{S+R_n > q^{\theta}_{S+R_n}(\alpha)\}\cap \{S+r \leq q^{\theta}_{S+r}(\alpha)\} \right].\label{eq:Zto0.2}
\end{align}
We want to show that \eqref{eq:Zto0.1} and \eqref{eq:Zto0.2} go to zero as $n\to\infty$. For brevity, we show the proof only for \eqref{eq:Zto0.1}; the proof for \eqref{eq:Zto0.2} is analogous. For any $\epsilon_2 >0$ we get
\begin{align}
\eqref{eq:Zto0.1} &= \bP^{\theta}\left[ \{q^{\theta}_{S+r}(\alpha) <S+r \leq q^{\theta}_{S+r-(r-R_n)}(\alpha)+(r-R_n)\} \right]\nonumber\\
&\leq \bP^{\theta}\left[ \{q^{\theta}_{S+r}(\alpha) <S+r \leq q^{\theta}_{S+r-(r-R_n)}(\alpha)+|r-R_n|\} \right]\nonumber\\
&\leq \bP^{\theta}\left[ \{ |r-R_n| \geq \epsilon_2\}\right]+\bP^{\theta}\left[ \{q^{\theta}_{S+r}(\alpha) <S+r \leq q^{\theta}_{S+r-(r-R_n)}(\alpha)+\epsilon_2\}\right]\label{eq:Zto0.4}.
\end{align}
Using \eqref{eq:Zto0.3}, and recalling that convergence in probability implies convergence in distribution which in turn implies convergence of quantiles (at continuity points) for $n\to\infty$ we get
\begin{equation}\label{eq:Zto0.5}
\bP^{\theta}\left[ \{ |r-R_n| \geq \epsilon_2\}\right]\to 0\quad \textrm{and}\quad  q^{\theta}_{S+r-(r-R_n)}(\alpha) \to q^{\theta}_{S+r}(\alpha).
\end{equation}
Combining \eqref{eq:Zto0.4} with \eqref{eq:Zto0.5}, noting that the choice of $\epsilon_2$ was arbitrary, and that $S$ is continuous, we conclude the proof.
\end{proof}

Now, we are ready to prove Proposition~\ref{pr:nonpar.fair}.
\begin{proof}[Proof of Proposition~\ref{pr:nonpar.fair}]
Let us fix $\theta\in\Theta$ and $i\in \{1,\ldots,d\}$. We want to show that
\[
\bE^\theta\Big[ Z_{S,\hat D^n}^\theta (X_i +\hat D^n_i)\Big]\to 0,\quad n\to\infty.
\]
Noting that
\[
\bE^\theta\Big[ Z_{S,\hat D^n}^\theta (X_i +\hat D^n_i)\Big]
=\bE^\theta\Big[ Z^\theta_n(X_i +\hat D^n_i)\Big]+\bE^\theta\Big[ Z^\theta_{S,a^\theta} (X_i +\hat D^n_i)\Big],
\]
where $Z^\theta_n:=Z_{S,\hat D^n}^\theta-Z^\theta_{S,a^\theta}$, we need to prove that
\begin{equation}\label{eq:cons1}
\bE^\theta\Big[ Z^\theta_n(X_i +\hat D^n_i)\Big]\to 0,\quad n\to\infty,
\end{equation}
and
\begin{equation}\label{eq:cons2}
\bE^\theta\Big[ Z^\theta_{S,a^\theta} (X_i +\hat D^n_i)\Big]\to 0,\quad n\to\infty.
\end{equation}
We start with the proof of~\eqref{eq:cons1}. Noting that for any $n\in\bN$ we have $|Z^\theta_n|\leq \frac{1}{\alpha}$
and
\[
\sum_{k=1}^{n}\1_{\{S^k+\hat{\var}^n_{\alpha}\leq 0\}} =\lfloor n\alpha\rfloor +1,
\]
we get
\begin{align}
\left|\bE^\theta\Big[Z^\theta_n (X_i +\hat D^n_i)\Big]\right| & \leq \bE^\theta\Big[ |Z^\theta_n| (|X_i| +|\hat D^n_i|)\Big]\nonumber\\
&\leq  \frac{1}{\alpha}\left(\bE^\theta\Big[ \1_{\{\left|Z^\theta_n\right|\neq 0\}} |X_i|\Big]+\bE^\theta\Big[ \1_{\{\left|Z^\theta_n\right|\neq 0\}} |\hat D^n_i|\Big]\right)\nonumber\\
&\leq  \frac{1}{\alpha}\left(\bE^\theta\Big[ \1_{\{\left|Z^\theta_n\right|\neq 0\}} |X_i|\Big]+\frac{1}{\lfloor n\alpha\rfloor +1}\bE^\theta\Big[ \1_{\{\left|Z^\theta_n\right|\neq 0\}}\sum_{k=1}^{n}|X_i^k|\1_{\{S^k+\hat{\var}^n_{\alpha}\leq 0\}}\Big]\right)\nonumber\\
&\leq  \frac{1}{\alpha}\left(\bE^\theta\Big[ \1_{\{\left|Z^\theta_n\right|\neq 0\}} |X_i|\Big]+\frac{1}{\lfloor n\alpha\rfloor +1}\bE^\theta\Big[ \1_{\{\left|Z^\theta_n\right|\neq 0\}}\sum_{k=1}^{n}|X_i^k|\Big]\right)\nonumber\\
&\leq  \frac{1}{\alpha}\left(\bE^\theta\Big[ \1_{\{\left|Z^\theta_n\right|\neq 0\}} |X_i|\Big]+\frac{n}{\lfloor n\alpha\rfloor +1}\bE^\theta\Big[ \1_{\{\left|Z^\theta_n\right|\neq 0\}}|X_i^1|\Big]\right).\label{eq:cons3}
\end{align}
Now, noting that  $\1_{\{\left|Z^\theta_n\right|\neq 0\}}=\1_{\{\left|Z^\theta_n\right|< \frac{1}{2\alpha}\}}$ and using Lemma~\ref{lm:Z.ok} we get
\[
\bP^{\theta}\left[|Z^\theta_n|\neq 0\right]\to 0,\quad n\to\infty.
\]
Combining this with \eqref{eq:cons3}, noting that $|X_i|$ and $|X_i^1|$ are integrable, and $\frac{n}{\lfloor n\alpha\rfloor +1}\to \frac{1}{\alpha}$ as $n\to\infty$, we conclude the proof of~\eqref{eq:cons1}.

Next, we prove~\eqref{eq:cons2}. Recalling that $a^\theta$ is a true allocation for $X$ under $\theta$ we get
\[
\bE^\theta\Big[ Z^\theta_{S,a^\theta} (X_i +\hat D^n_i)\Big]=\bE^\theta\Big[ Z^\theta_{S,a^\theta} (X_i +a_i^{\theta})\Big]+\bE^\theta\Big[ Z^\theta_{S,a^\theta} (\hat D^n_i-a_i^\theta)\Big]=\bE^\theta\Big[ Z^\theta_{S,a^\theta} (\hat D^n_i-a_i^\theta)\Big].
\]
Consequently, noting that $Z^\theta_{S,a^\theta}$ and $\hat D^n_i$ are independent under $\bP^{\theta}$ we get
\begin{align}
\bE^\theta\Big[ Z^\theta_{S,a^\theta} (X_i +\hat D^n_i)\Big] &=\bE^\theta\Big[ Z^\theta_{S,a^\theta}\Big] \bE^\theta\Big[\hat D^n_i-a_i^\theta\Big]\nonumber\\
& = -\bE^\theta\left[\frac{\sum_{k=1}^{n}X_i^k\1_{\{S^k+\hat{\var}^n_{\alpha}\leq 0\}}}{\sum_{k=1}^{n}\1_{\{S^k+\hat{\var}^n_{\alpha}\leq 0\}}}+a_i^\theta\right]\nonumber\\
& =  -\frac{1}{\lfloor n\alpha\rfloor +1}\bE^\theta\left[\sum_{k=1}^{n}(X_i^k+a_i^\theta)\1_{\{S^k+\hat{\var}^n_{\alpha}\leq 0\}}\right]\nonumber\\
& =  -\frac{n}{\lfloor n\alpha\rfloor +1}\bE^\theta\left[ (X_i^1+a_i^\theta)\1_{\{S^1+\hat{\var}^n_{\alpha}\leq 0\}}\right]\nonumber\\
& =-\frac{n}{\lfloor n\alpha\rfloor +1}\bE^\theta\left[ (X_i^1+a_i^\theta)\left(\1_{\{S^1+\hat{\var}^n_{\alpha}\leq 0\}}-\1_{\{S^1\leq q^\theta_{S}(\alpha)\}}\right)\right].\nonumber\\
& \leq\frac{n}{\lfloor n\alpha\rfloor +1}\bE^\theta\left[ |X_i^1+a_i^\theta|\1_{A_n}\right],\label{eq:new.2nd.1}
\end{align}
where $A_n:=\{q^\theta_{S}(\alpha)<S^1\leq -\hat{\var}^n_{\alpha}\} \cup \{-\hat{\var}^n_{\alpha}<S^1\leq q^\theta_{S}(\alpha)\}$, and where in the last equality we used the property $\bE^\theta\left[ (X_i^1+a_i^\theta)\1_{\{S^1\leq q^\theta_{S}(\alpha)\}}\right]=0$. By similar reasoning as in \eqref{eq:Zto0.4}, we get
\begin{align*}
\bP^\theta[A_n]
& \leq\bP^\theta[|q^\theta_{S}(\alpha)+\hat{\var}^n_{\alpha}|>\epsilon] +\bP^\theta[|S^1+\hat{\var}^n_{\alpha}|<\epsilon],
\end{align*}
for any $\epsilon>0$. Since $-\hat{\var}^n_{\alpha}$ is a consistent estimator of $q^\theta_{S}(\alpha)$, we conclude that $\bP^\theta[|q^\theta_{S}(\alpha)+\hat{\var}^n_{\alpha}|>\epsilon] \to 0$, as $n\to\infty$. Consequently, as the choice of $\epsilon$ was arbitrary and $S^1$ is continuous, we obtain that $\bP^\theta[A_n] \to 0$, as $n\to\infty$. Combining this with \eqref{eq:new.2nd.1}, and since $ |X_i^1+a_i^\theta|$ is integrable, and $\frac{n}{\lfloor n\alpha\rfloor +1}\to \frac{1}{\alpha}$ as $n\to\infty$, the proof of~\eqref{eq:cons2} is complete.

\end{proof}

\section{Backtesting and numerical examples}\label{sec:backtesting}

In this section we  analyze the proposed fair capital allocation methodology via examples using simulated data and real market data. It goes without saying that any quantitative methodology used for measuring and allocating risk relies on an adopted formal model. It also goes without saying that actual results of risk measurement and/or risk allocation need to be tested for their adequacy. Often, testing adequacy of the results of risk measurement is done in practice using backtesting, and we will use this approach in testing the estimation procedures of fair capital allocation introduced in the previous sections.

\emph{Backtesting}, applied for risk measurement in the financial context,  can be summarized as follows: given a time series of capital forecasts, one compares these forecasts with the realized losses; the accumulated performance is the key ingredient of the backtesting.
Backtesting might be also treated as a specific case of assessment of quality of a {\it point forecast}, which aims at assessing whether the forecasted capital is sufficient; see \cite{Ziegel2014,SchKatGne2015,NolZie2017}. In particular, backtesting value-at-risk goes back to~\cite{Kupiec73} and recently has gained a lot of practical and theoretical interest; see \mbox{\cite{AcerbiSzekely2014,PiteraSchmidt2018}} for further details on this topic and the related literature.
Undoubtedly, similar backtesting procedure should be developed for testing the adequacy of risk capital allocation methodologies.

We focus our attention on assessing the performance of a statistical capital allocation methodology when the underlying reference risk measure is expected shortfall at the fixed level $\alpha\in (0,1]$, used in computing of the values of our estimators. For this purpose we propose two backtesting frameworks: \begin{itemize}\addtolength{\itemsep}{-0.5\baselineskip}
\item absolute deviation from fairness backtesting;
\item risk level shifts adjustments backtesting.
\end{itemize}
The backtesting framework adopted for assessment of adequacy of estimators of capital allocations, say $\widetilde {A} = (\widetilde A_1, \ldots, \widetilde A_d)$, that were created using some capital allocation methodology,\footnote{We refer to such methodology as to an Internal Capital Allocation Model (ICAM).} uses as its input the observations of past P\&Ls.
The key ingredient to both backtesting methods is the estimation of
\begin{equation}\label{eq:fair.bt-total}
 \sum_{i=1}^d\bE^{\theta_0}\Big[ Z_{S,\hat A}^{\theta_0}(X_i +\widetilde A_i)\Big],
\end{equation}
and the estimation of
\begin{equation}\label{eq:fair.bt}
\bE^{\theta_0}\Big[ Z_{S,\hat A}^{\theta_0} (X_i +\widetilde A_i)\Big], \ i=1,2,\ldots,d.
\end{equation}
We assume that the length of the backtesting window is $m$ days. With each day $k=1,\ldots,m,$ we associate the P\&Ls $X^k_i$ and allocation estimators $\widetilde A^k_i$, $i=1,\ldots,d.$ The estimators $\widetilde A^k_i$ can be obtained in various ways. One way is to proceed in accordance to what was  proposed previously in this paper. Specifically,  to produce allocation estimators $\widetilde A^k_i$ on day $k$ one uses market observations from the previous $n$ days. We denote these observations as $\mathbf{X}^{k,n}=(X_i^{k-n},\ldots,X_i^{k},\ i=1,\ldots,d)$. Based on these observations, and following \eqref{estimator}, we compute the estimators of the allocations as $\widetilde A^{k}=\eta_n(\mathbf{X}^{k,n})$.

The realizations of $X^k_i$ and $\widetilde A^k_i$ are denoted as $x^k_i$ and $\widetilde a^k_i$, respectively.  We set $y:=(y^1,\dots,y^m)$, where $y^k= (y_1^k,\ldots ,y_d^k)$, $k=1,\ldots,m$, and
$ y_i^k:=x_i^k+ \widetilde a^k_i,\quad i=1,2,\ldots,d$. 
We also let $\xi^k:=\sum_{i=1}^{d}y_i^k,\quad k=1,2,\ldots,m$, to  denote the realized aggregated secured position on day $k$, and we set $\xi:=(\xi^1,\ldots,\xi^m)$.

In order to proceed we introduce the following functions of $\beta \in (0,1]$,

\begin{equation}\label{eq:fair.bt-total-hist}
 G_{\beta}(\tilde{A}):=-\frac{\sum_{k=1}^{m}\xi^k\1_{\{\xi^k+\hat{\var}_{\beta}(\xi)\leq 0\}}}{\sum_{k=1}^{m}\1_{\{\xi^k+\hat{\var}_{\beta}(\xi)\leq 0\}}},
\end{equation}
and
\begin{equation}\label{eq:fair.bt-hist}
G^i_{\beta}(\tilde{A}):=-\frac{\sum_{k=1}^{m}y_i^k\1_{\{\xi^k+\hat{\var}_{\beta}(\xi)\leq 0\}}}{\sum_{k=1}^{m}\1_{\{\xi^k+\hat{\var}_{\beta}(\xi)\leq 0\}}}, \quad i=1\ldots,d,
\end{equation}
with $\hat{\var}_{\beta}$ being the empirical value-at-risk at level $\beta\in (0,1]$. Note that $y^i_k$s are computed using as the reference risk measure ES at the fixed risk level $\alpha$. If no confusions arise, we will write  $G_\beta$, respectively $G_\beta^i$, instead of $G_{\beta}(\tilde{A})$, respectively $ G_{\beta}^i(\tilde{A})$.

Now, similarly to the derivation of $\hat D^n_i$, we estimate the expectation in \eqref{eq:fair.bt-total} as $-G_{\alpha}$, and we estimate \eqref{eq:fair.bt} as $-G^i_{\alpha}$.

\medskip\noindent
\textbf{Deviation from fairness backtesting.} 
If the capital allocation methodology is fair, then the obtained empirical values $G_\alpha^i, \ i=1,\ldots,d$, should be close to zero, for the fixed reference level $\alpha$; the bigger the obtained estimate, the bigger the potential (true) deviation from fairness for the $i$th margin. \textit{The deviation from fairness backtest} assesses proximity to zero of $G_\alpha^i, \ i=1,\ldots,d$. A comprehensive study of properties of $G_\alpha^i$s, such as `how far from zero is an acceptable value' is beyond the scope of this manuscript. Nevertheless, the following backtesting methodology is one way to address this question.

\medskip\noindent
\textbf{Risk level shift backtesting.} Instead of measuring the deviation from fairness directly, it is natural to find the reference risk level $\beta\in(0,1]$ that makes $G_\beta^i$ closest to zero; equivalently, we want to answer the question
by how much one needs to shift the reference risk level $\alpha$ to make the position acceptable.  This approach hinges on duality-based performance measurement introduced in~\cite{PiteraMoldenhauer2018}. It should be noted that this approach is different from the elicitability-based backtests as it focuses on capital conservativeness assessment rather than the general forecast fit; cf.~\cite{NolZie2017}. Formally, for the estimators of capital allocation $\tilde{A}$, we define
\begin{align}
\Upsilon (\tilde{A}) &:=\inf\set{\beta\in (0,1]\, :\, G_{\beta}(\tilde{A})\leq 0}, \label{eq:R.bt} \\
W_{-}^{i}(\tilde{A}) &:= \inf\{\epsilon\in [0,\alpha]\, :\,G^i_{\alpha}(\tilde{A})\cdot G^i_{\alpha-\epsilon}(\tilde{A})\leq 0 \}, \label{eq:W+}\\
W_{+}^{i} (\tilde{A}) &:= \inf\{\epsilon\in [0,1-\alpha]\, :\, G^i_{\alpha}(\tilde{A})\cdot G^i_{\alpha+\epsilon}(\tilde{A}) \leq 0 \}, \label{eq:W-}
\end{align}
where in \eqref{eq:W+} we use the convention $\inf \emptyset =\alpha$, and correspondingly,  in  \eqref{eq:W-} we put $\inf \emptyset =1-\alpha$. Similar to $G^i_\beta, G_\beta$, we may simple write $\Upsilon$, and $W_\pm^i$. Note that $G_\beta$ is a monotone decreasing function in $\beta$, while $G_\beta^i$ generally speaking is not monotone. Hence, the quantities $W^\pm$ are defined as the smallest shift in the reference risk level from $\alpha$, to the right or to the left, that makes the $i$th secured position acceptable. Thus,
the closer $\Upsilon$ is to the initial reference risk level $\alpha$ the better is the total risk estimation procedure. Similarly, the closer $W^\pm$ are to zero, the better is the risk allocation procedure. One can look at $W$ as the performance index that is dual to the ES family; see \cite[Proposition 4.3]{PiteraMoldenhauer2018} for more details.

Finally, by combining the left and right minimal shifts, we define the the minimal shift estimator as
\begin{equation}\label{eq:performance2}
W^i (\tilde{A}) :=
\begin{cases}
-W_{-}^{i},  & \textrm{if } W_{-}^{i} < W_{+}^{i} \\
\phantom{-}W_{+}^{i},  & \textrm{if } W_{-}^{i} \geq W_{+}^{i}
\end{cases},
\quad i=1,2,\ldots,d.
\end{equation}

\smallskip \noindent
Before moving to numerical examples, several comments on backtesting procedure are in order.
\begin{enumerate}[(a)]\addtolength{\itemsep}{-0.5\baselineskip}
\item It goes without saying that the results produced by the deviation from fairness and the risk level shift approaches should be compared with each other for consistency and reality check.
\item It is worth mentioning that the two proposed backtesting methodologies can be applied to any  ICAM, not necessarily those discussed in this paper.
\item  Our study of the backtesting procedure of the estimation of the risk capital allocation is  preliminary.  A thorough investigation of the statistical properties of $ G_\alpha^i$ and $W^i$ is deferred to future studies.
\end{enumerate}

\noindent Next we will illustrate the performance of the capital allocation estimators $\hat B^n$, $\hat C^n$, and $\hat D^n$ on simulated data by applying the two backtesting procedures described above. For brevity and to ease the notation, we will write $\hat B^n$, $\hat C^n$, and $\hat D^n$ as $\hat B$, $\hat C$, and $\hat D$, respectively. 

For simulations, we consider two cases of probability distributions of the P\&Ls vector $X$ - the Gaussian distribution and the Student's $t$-distribution. We also fix the reference level $\alpha=0.05$. All numerical evaluations are performed using R statistical software; the source codes are available from the authors upon request.

\begin{example}[Gaussian P\&Ls]\label{ex:1}
We assume that the portfolio $X$ of eight (discounted) P\&Ls follows an eight dimensional Gaussian distribution $\cN(\mu, \Sigma)$, with the (true) mean
\[
\mu= (0.000786,  0.001549,  0.001660, 0.000195, 0.000650,  0.000413, -0.000401, -0.001146),
\]
and the (true) variance-covariance matrix
{\small
\[
\Sigma=\left[
\begin{tabular}{rrrrrrrrrr}
0.000226 & 0.000174 & 0.000104 & 0.000066 & 0.000069 & 0.000019 & -0.000077 & -0.000135 \\
0.000174 & 0.000346 & 0.000135 & 0.000068 & 0.000091 & 0.000022 & -0.000082 & -0.000195 \\
0.000104 & 0.000135 & 0.000257 & 0.000065 & 0.000084 & 0.000034 & -0.000093 & -0.000111 \\
0.000066 & 0.000068 & 0.000065 & 0.000133 & 0.000048 & 0.000025 & -0.000058 & -0.000064 \\
0.000069 & 0.000091 & 0.000084 & 0.000048 & 0.000137 & 0.000034 & -0.000065 & -0.000081 \\
0.000019 & 0.000022 & 0.000034 & 0.000025 & 0.000034 & 0.000061 & -0.000022 & -0.000031 \\
-0.000077 & -0.000082 & -0.000093 & -0.000058 & -0.000065 & -0.000022 & 0.000149 & 0.000085 \\
 -0.000135 & -0.000195 & -0.000111 & -0.000064 & -0.000081 & -0.000031 & 0.000085 & 0.000202 \\
\end{tabular}
\right].
\]
}
For the purpose of obtaining the above mean vector and the variance-covariance  matrix we used values of daily  returns of eight stocks from S\&P 500 index, namely: AAPL, AMZN, BA, DIS, HD, KO, JPM, and MSFT; these data were taken for the period from January 2015 till December 2018. We will use this sample again in Example \ref{ex:4}. The first six stocks represent long positions in our portfolio and the last two represent short positions; this gives the negative entries in $\mu $ and $\Sigma$. The positions in each stock are equally weighted with nominal (absolute) value \$1.

We took the learning period of $n=500$ days, and the backtesting period of $m=5,\!000$ days. Below, we present the results  for the Gaussian plug-in estimator $\hat C$ and the non-parametric estimator $\hat D$; we omit results for estimators $\hat B$ and $\check D$, since, due to large size of the learning period, the results are almost identical to $\hat C$ and $\hat D$, respectively. Additionally, for comparison, we present results for the true allocations $a$; these allocations were obtained by plugging-in true mean and covariance matrix into \eqref{eq:norm.es.true}.

\begin{figure}[htp!]
\begin{center}
\scalebox{0.25}{\includegraphics{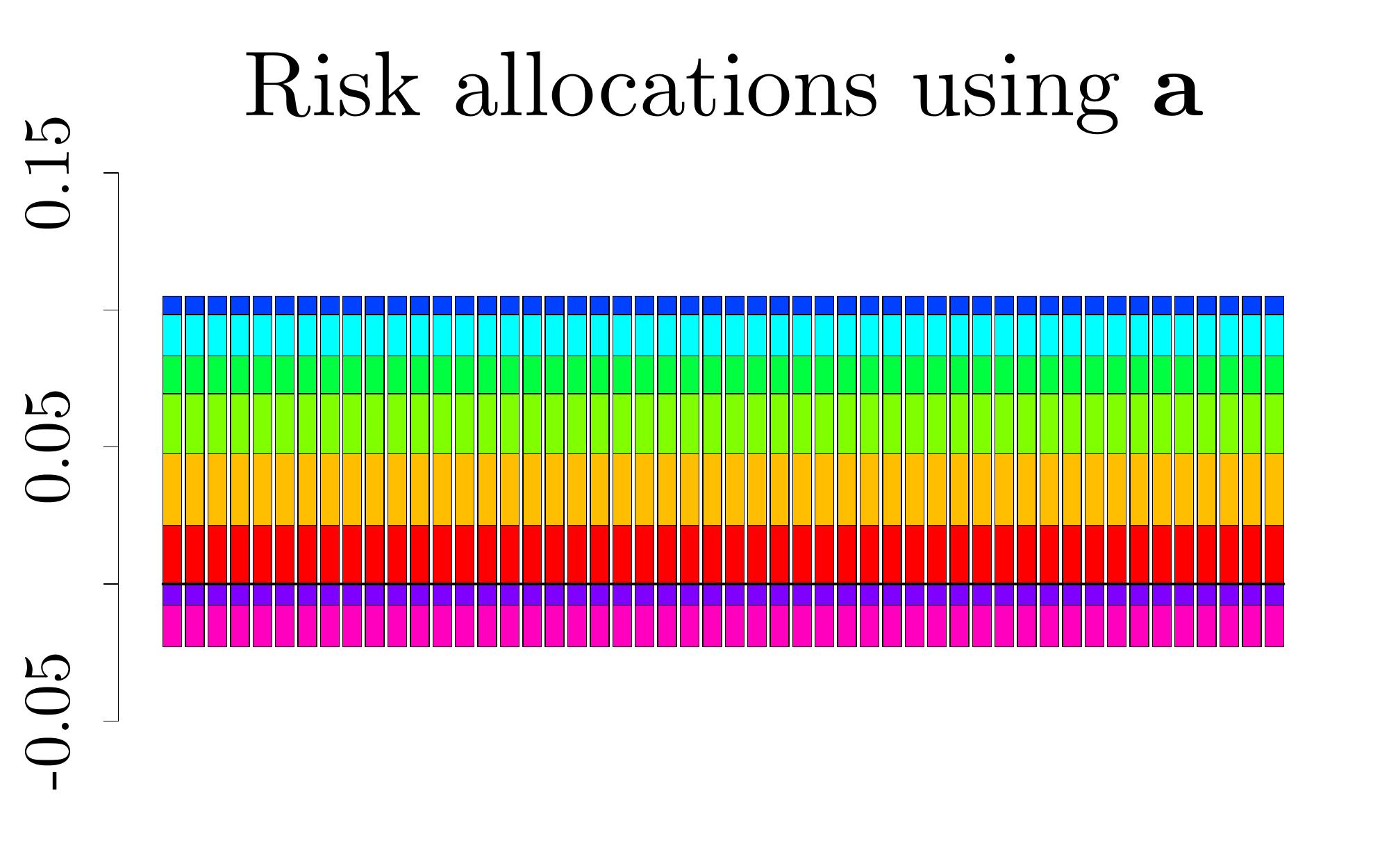}}
\scalebox{0.25}{\includegraphics{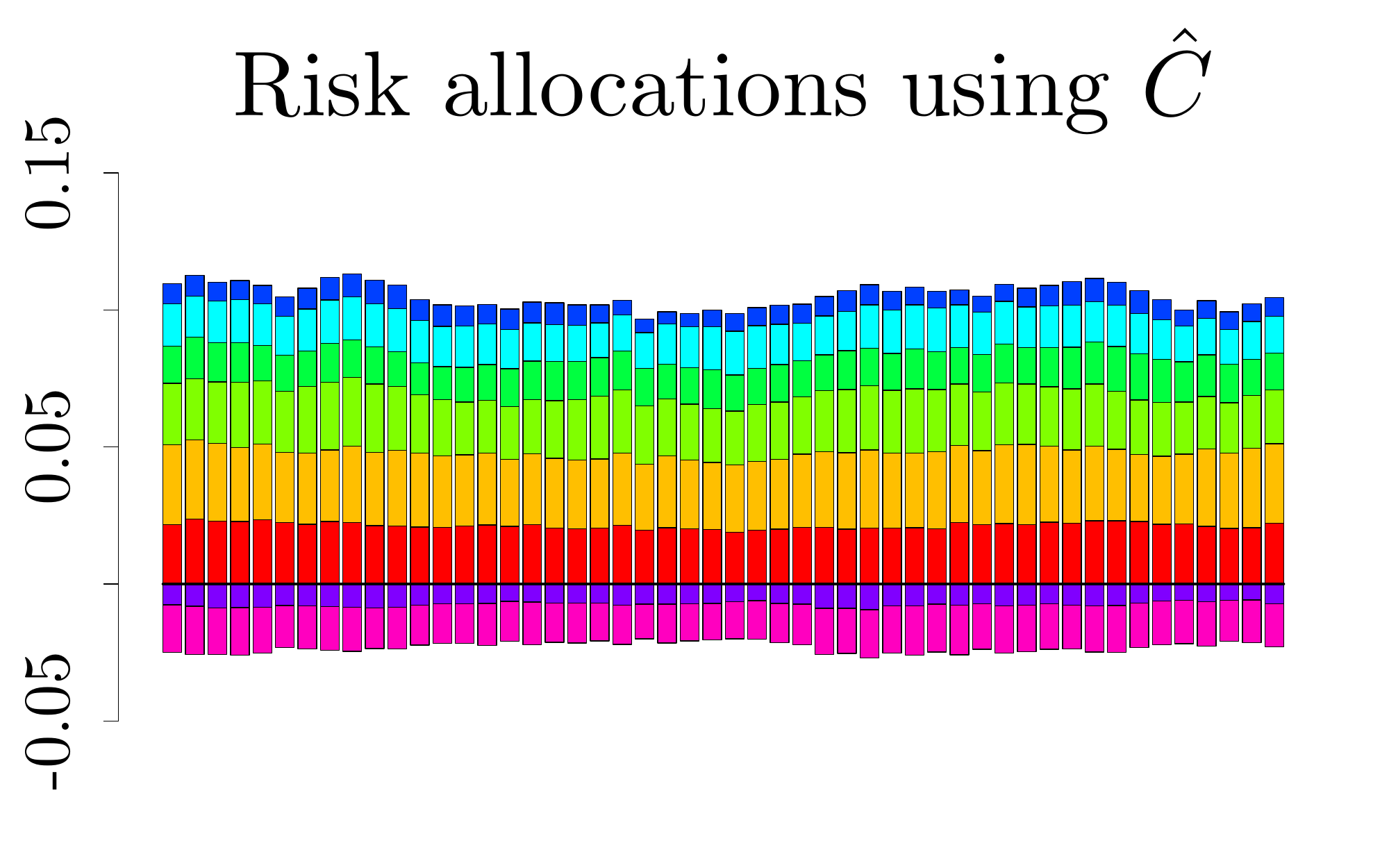}}
\scalebox{0.25}{\includegraphics{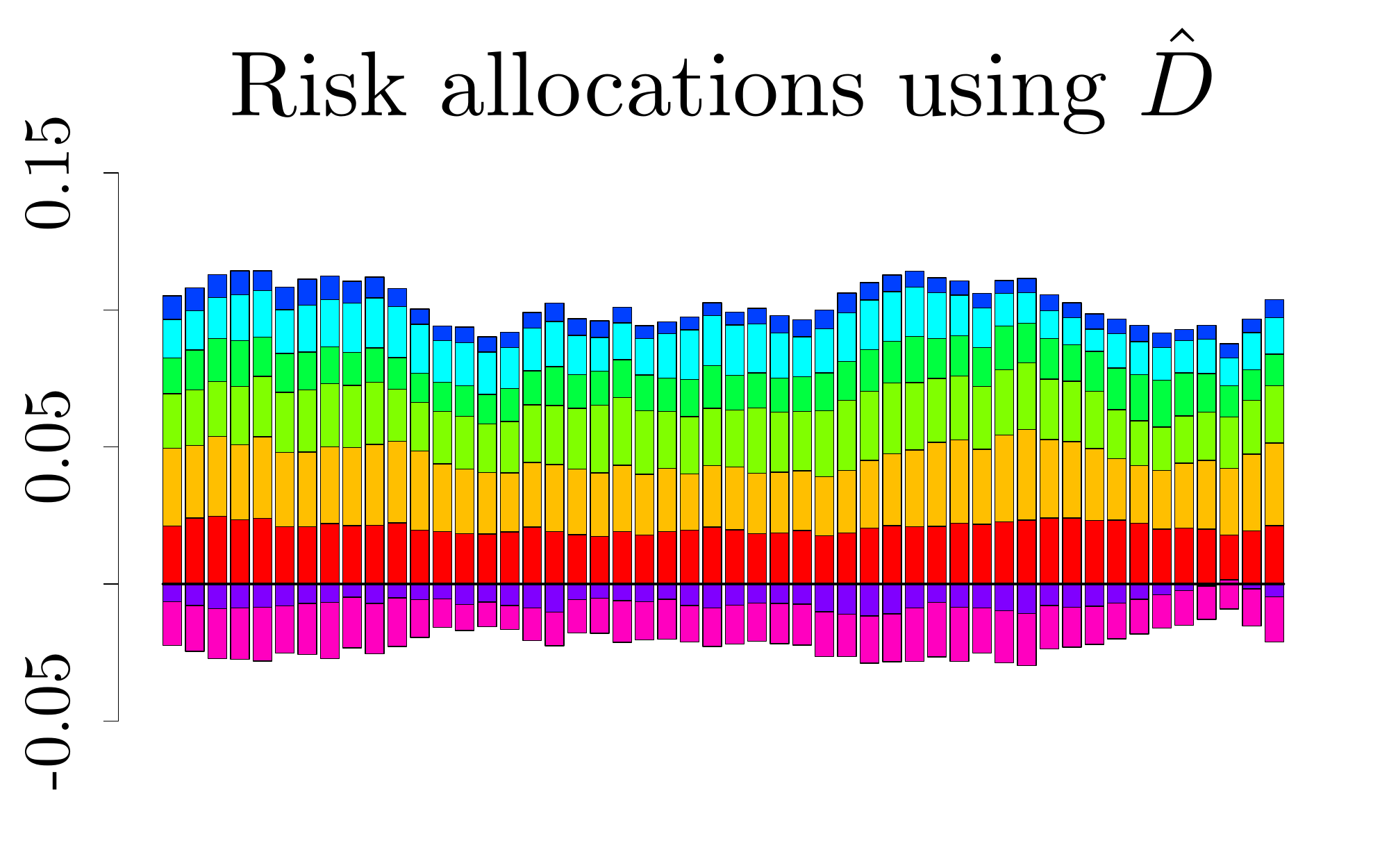}}\\
\scalebox{0.25}{\includegraphics{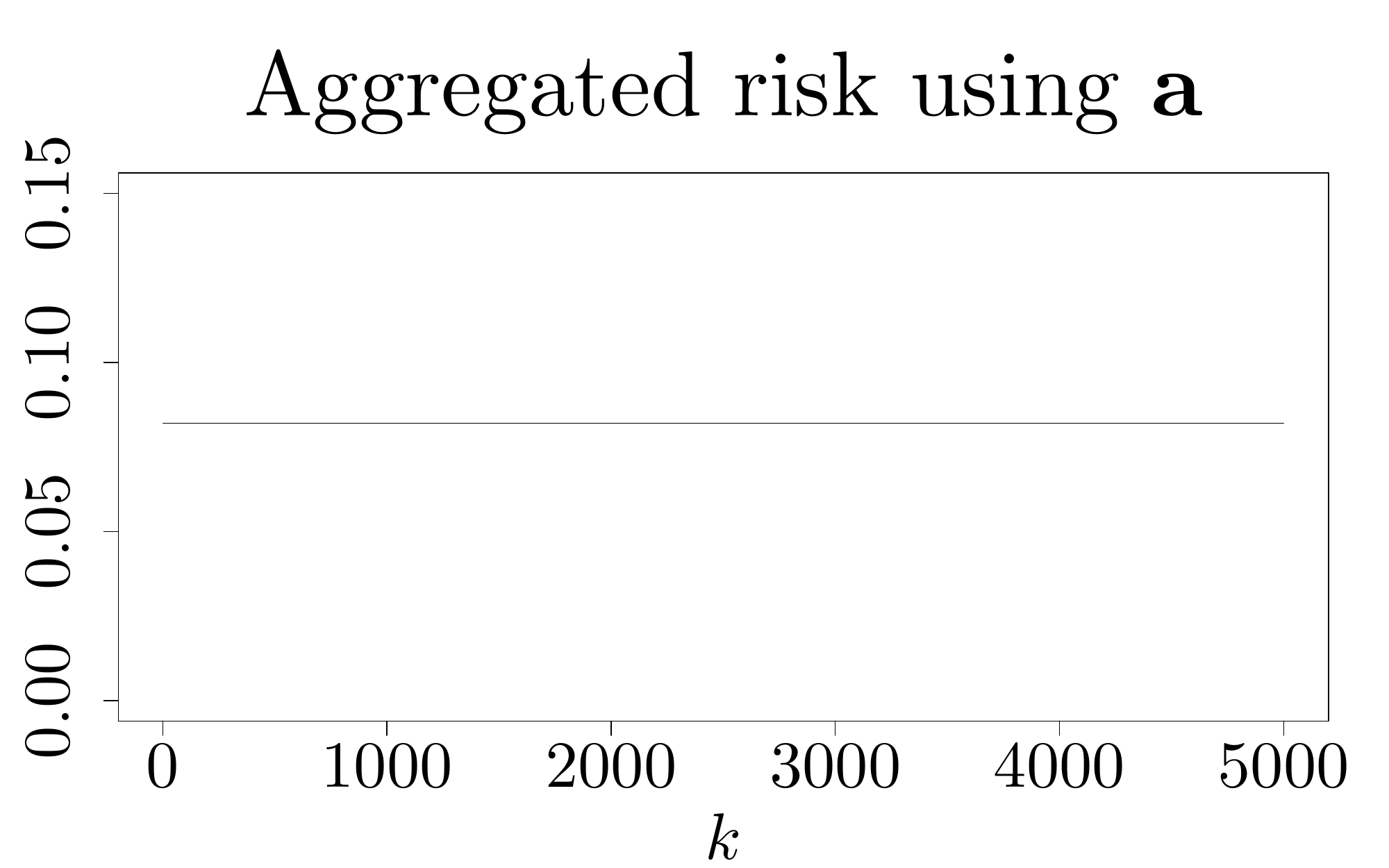}}
\scalebox{0.25}{\includegraphics{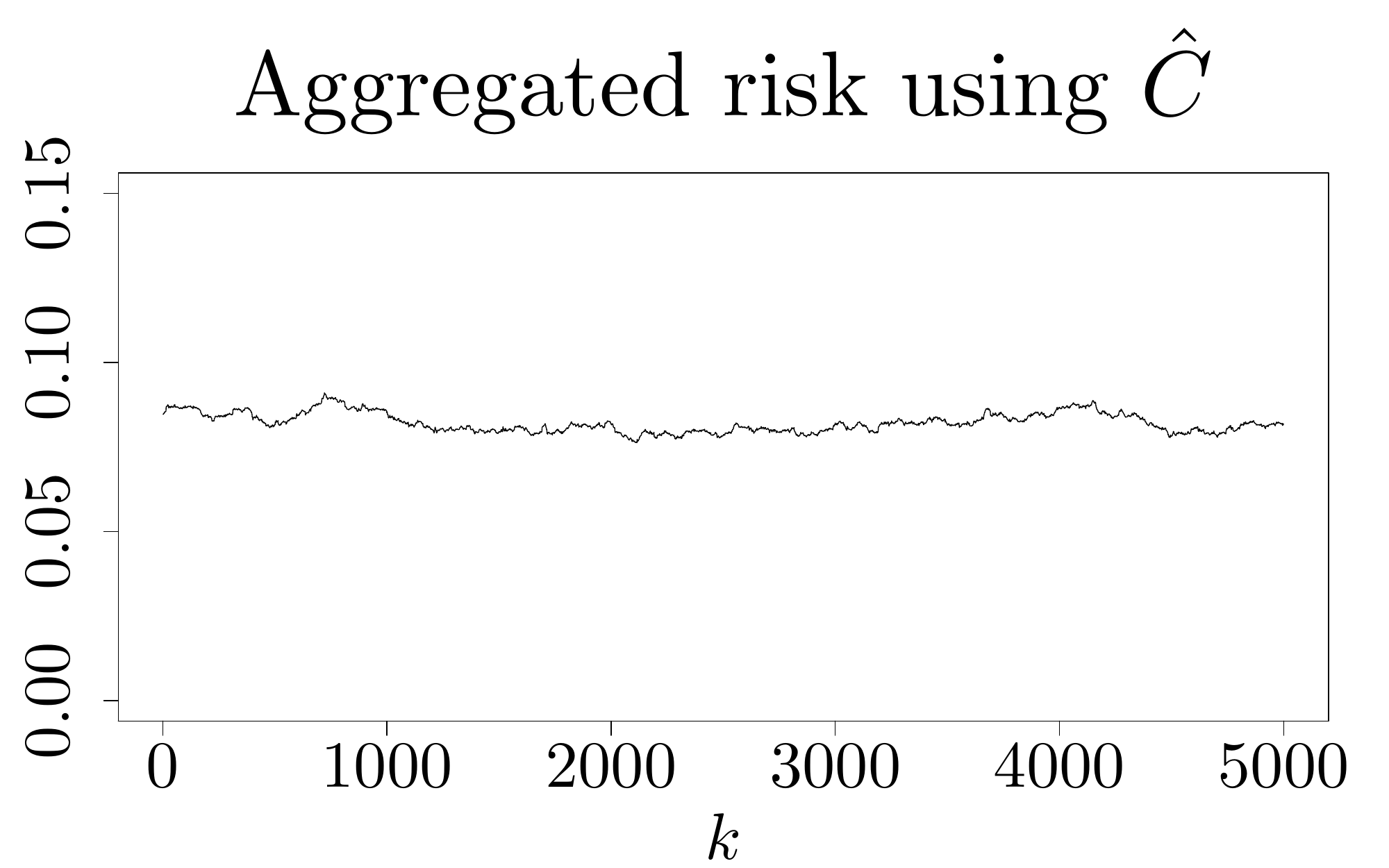}}
\scalebox{0.25}{\includegraphics{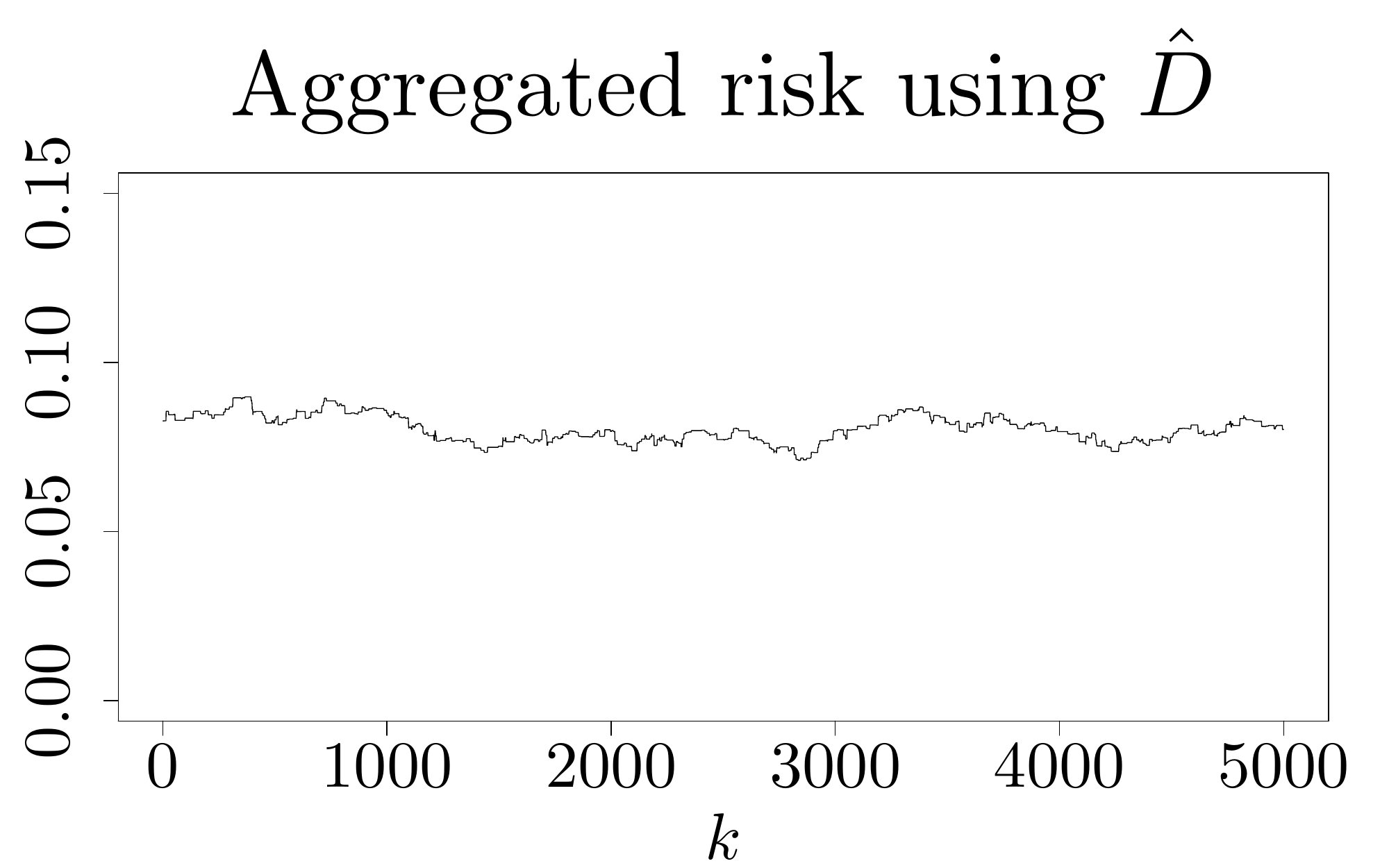}}
\end{center}
\caption{Example~\ref{ex:1}.
Top row: estimated risk allocations for the eight portfolio constituents (indexed by color) at each backtesting day, $k=1,\ldots,m$ for the true allocation $a$, and the estimated risk allocations $\hat C$ and $\hat D$; the height of each horizontal layer represents the risk allocated to one of the constituents.
Bottom row: estimated aggregated risk at each backtesting day. The estimated risk allocations are close to the reference allocations. The Gaussian plug-in estimator $\hat C$ slightly outperforms the non-parametric method $\hat D$.}
\label{fig:Ex1-1}
\end{figure}

The obtained results  validate, as expected,  the proposed methods.  In Figure~\ref{fig:Ex1-1} we present the values of the risk allocation to each constituent (top row), and the aggregated risk (bottom row). In this example, the fair risk allocation $a$ computed with the true underlying distribution can be considered as  reference for the  backtesting results. The estimated risk allocations using $\hat C$ and $\hat D$ are close to the reference allocations, and as expected, the results computed using the non-parametric method $\hat D$ are not as close to the reference results as those obtained using $\hat C$ that explicitly exploits the Gaussian distribution structure of the data.

Table~\ref{T:normal} contains the summary of the estimated backtesting measures $G_{0.05}, G_{0.05}^i, W^i$ and $\Upsilon$.  First, we note that the values of $G_{0.05}(a), G^i_{0.05}(a)$ and $W^i(a)$ corresponding to backtesting the fair allocation are, as expected, close to zero. In addition, $\Upsilon (a)$ is close to $\alpha=0.05$.  This indicates that the proposed backtesting methodologies are adequate. The obtained values give the benchmark for the following results produced by using $\hat C$ and $\hat D$. We note that indeed, the values of $G_{0.05}, G^i_{0.05}, W^i$ and $\Upsilon$ corresponding to $\hat C$ and $\hat D$ are in the same ballpark as for $a$, indicating that $\hat{C}$ and $\hat D$ are suitable risk allocation methodologies.  We also provide a graphical representation of $G^i_{0.05}$ in Figure~\ref{fig:Ex1-3} (top row), and in Figure~\ref{fig:Ex1-3} (bottom row) we $G_\beta^i(\hat{D}), i=1,\ldots,8$ as function of $\beta$.

\begin{table}[htp!]
\centering
\scalebox{0.8}{
\begin{tabular}{rrrrrrrrrrr} \toprule
& $X_1$ & $X_2$ & $X_3$ & $X_4$ & $X_5$ & $X_6$ & $X_7$ & $X_8$ & & $S$\\
\midrule
 $G^i_{0.05}(a)$ & -0.00038 &  0.00017 & -0.00054 & -0.00039 & -0.00021 & -0.00048 &  0.00076 & -0.00011 & $G_{0.05}(a)$ & -0.00118 \\
    $W^i(a)$ & -0.013 &  0.001 & -0.002 & -0.002 & -0.003 & -0.012 & -0.017 &  0.001 & $\Upsilon(a)$ & 0.047 \\
   \midrule
 $G^i_{0.05}(\hat{C})$ & -0.00090 & -0.00037 & -0.00014 & -0.00022 &  0.00043 & -0.00058 &  0.00088 &  0.00008 & $G_{0.05}(\hat{C})$ & -0.00081 \\
    $W^i(\hat{C})$ & -0.009 & -0.004 & -0.002 & -0.002 &  0.004 & -0.016 & -0.013 & -0.005 & $\Upsilon(\hat{C})$ & 0.048 \\
   \midrule
$G^i_{0.05}(\hat{D})$ &  0.00032 &  0.00110 &  0.00031 &  0.00014 & -0.00003 &  0.00010 & -0.00011 & -0.00088 & $G_{0.05}(\hat{D})$ & 0.00094 \\
    $W^i(\hat{D})$ &  0.003 &  0.005 &  0.002 &  0.005 & -0.001 & -0.001 &  0.002 &  0.008 & $\Upsilon(\hat{D})$ & 0.053 \\
   \bottomrule
\end{tabular}
}
\caption{Summary of the estimated backtesting measures for Example~\ref{ex:1}: In the first columns we show $G^i_{0.05}$ and $W^i$, $i=1,\dots,8$, for the true allocation $a$, and the estimated risk allocations $\hat C$ and $\hat D$, corresponding to backtesting the fair allocation. The values are close to zero, indicating that the proposed backtesting methodologies are adequate. The last column shows the aggregated quantities $G_{0.05}$ and the risk level shift $\Upsilon$. Here, $\Upsilon$ is close to $\alpha=0.05$, as expected.}
\label{T:normal}
\end{table}

For convenience, we additionally present several graphical representations of the backtesting metrics. In Figure~\ref{fig:Ex1-2} we plot  $G_\beta$ and $G_\beta^i$ as functions of $\beta$, for the three risk allocation methods $a, \hat{C}, \hat{D}$. All these functions should take zero value around $\beta=\alpha=0.05$, which is clearly the case. Finally, Figure~\ref{fig:Ex1-4} is dedicated to risk level shift backtesting. The top row shows the values of $W^i$ for risk allocations estimated  using $a, \hat{C}$, and $\hat{D}$. The blue dots in the bottom graphs in Figure~\ref{fig:Ex1-4} depict the values of $\alpha\pm W$, all of them being close to the reference risk value $\alpha=0.05$, which again indicates adequacy of risk allocation estimation procedure $\hat D$.

\begin{figure}[t!]
\begin{center}
\scalebox{0.25}{\includegraphics{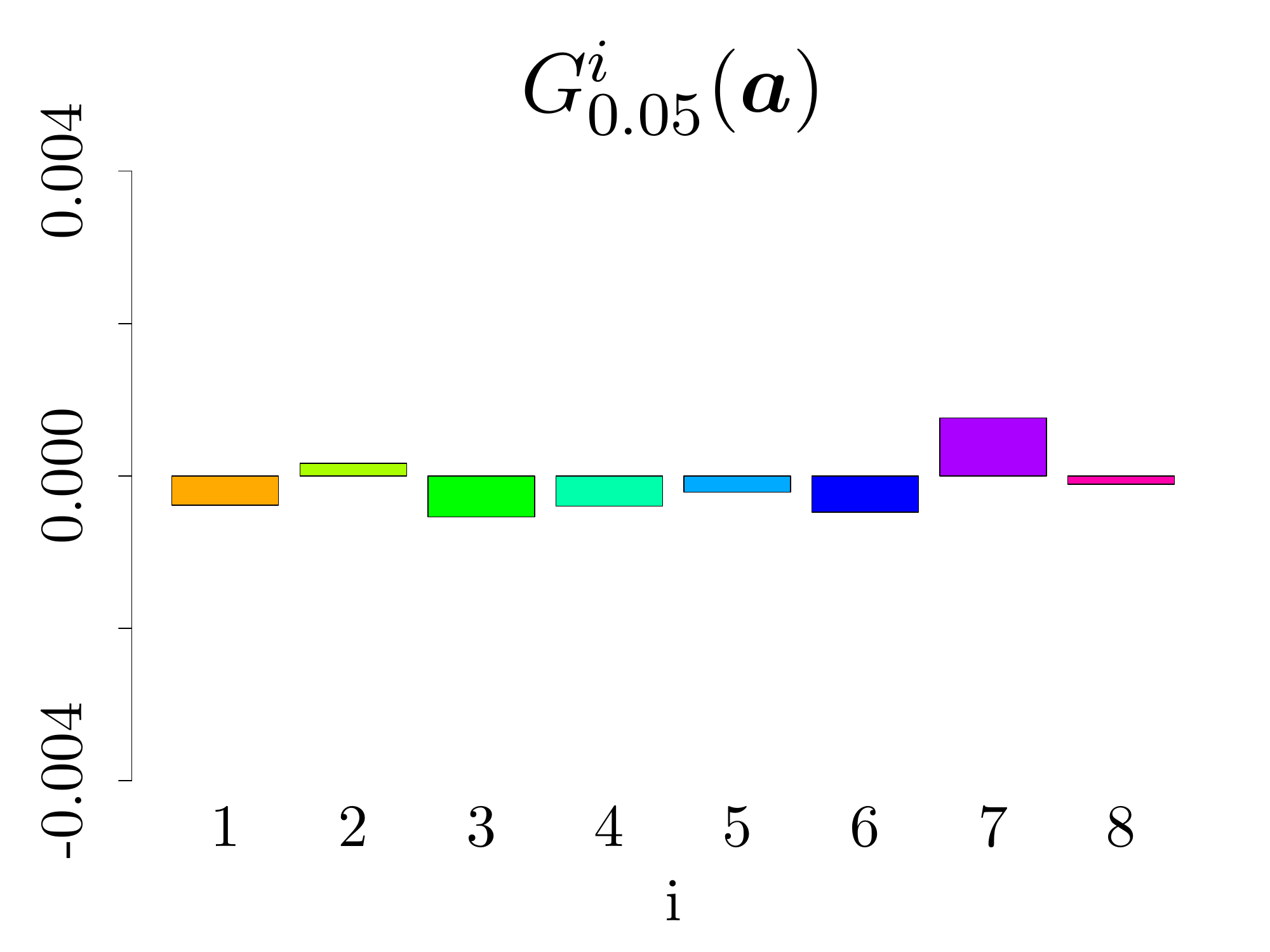}}
\scalebox{0.25}{\includegraphics{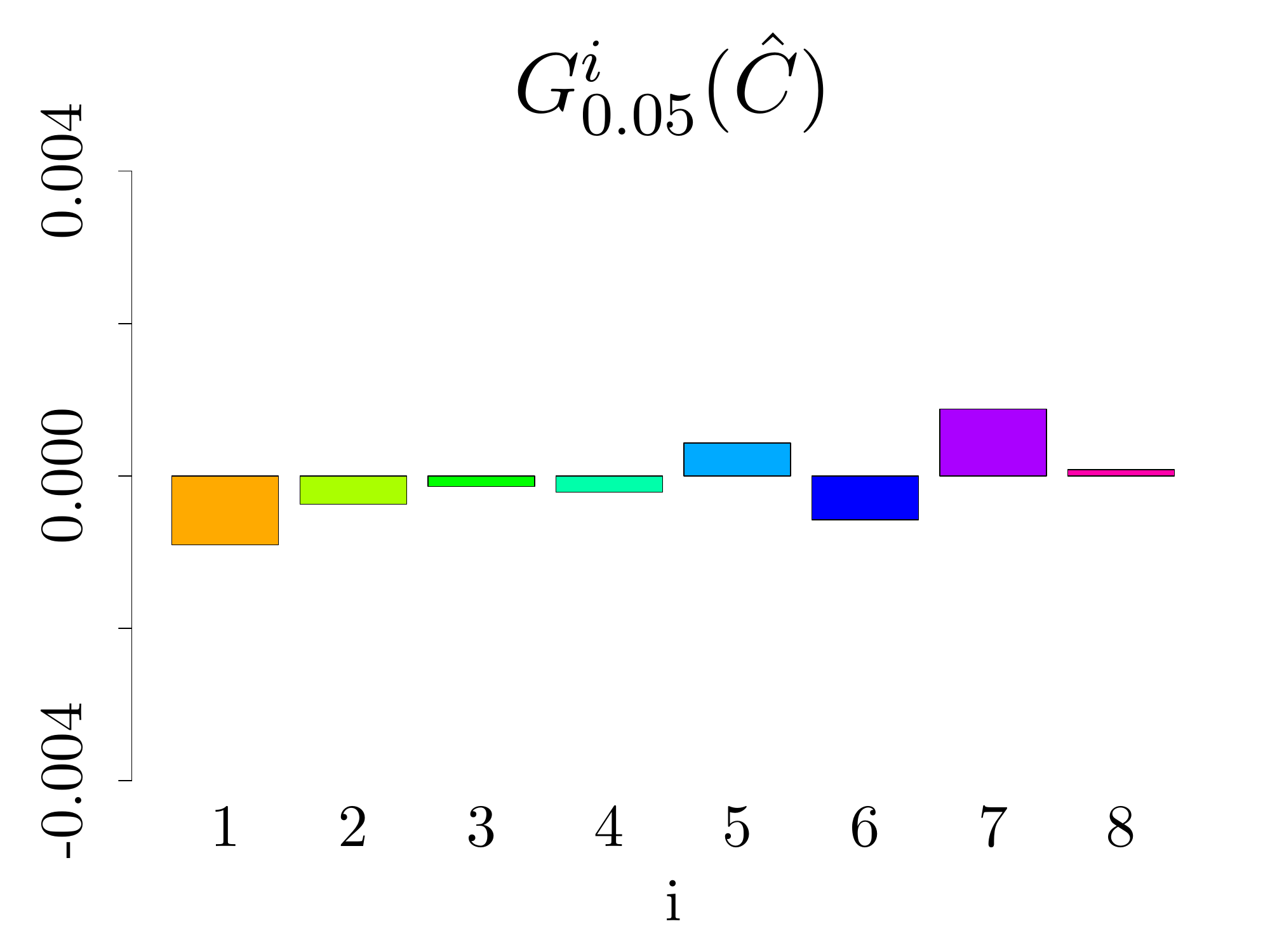}}
\scalebox{0.25}{\includegraphics{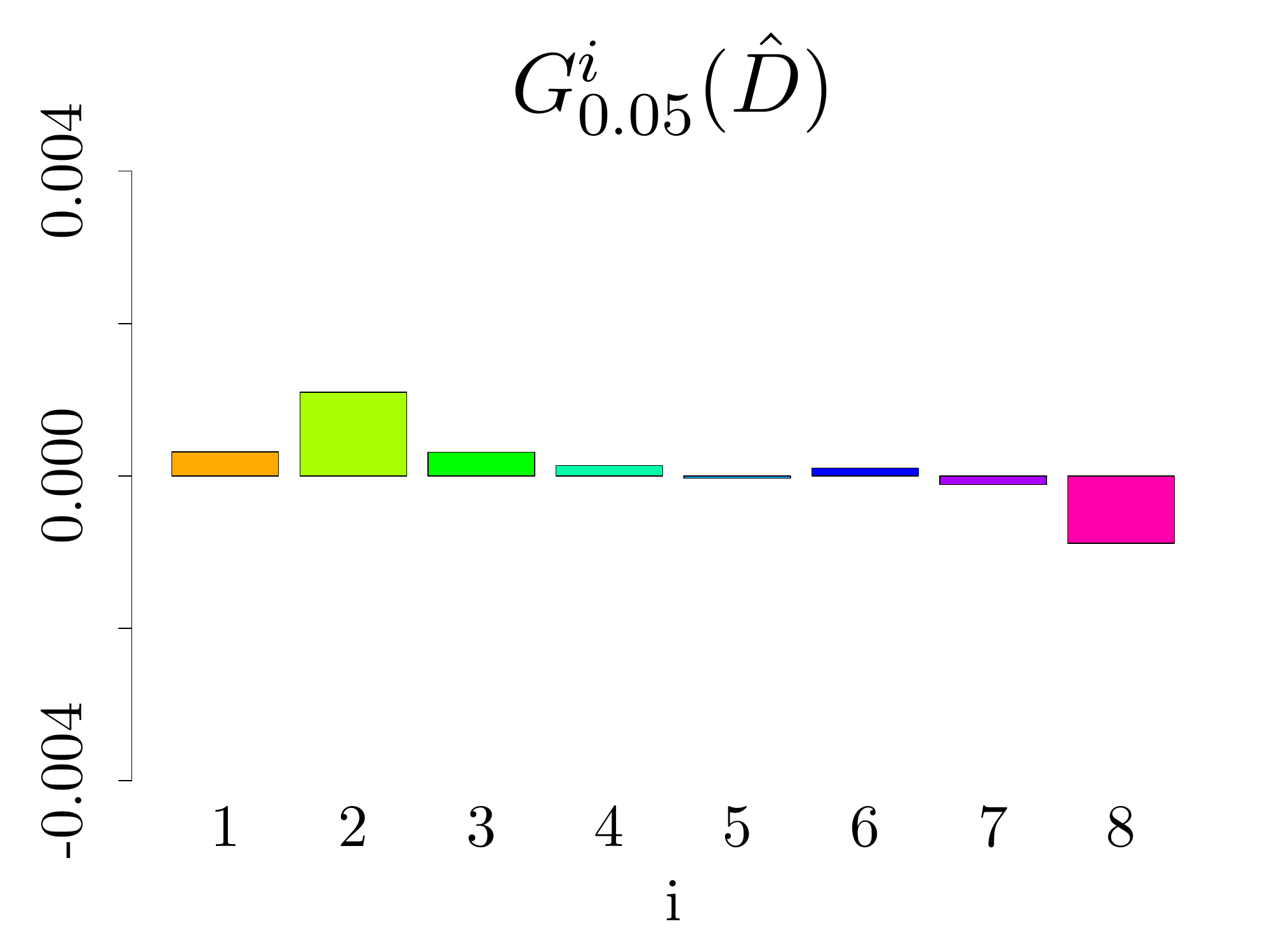}}
\vspace{0.6cm}
\scalebox{0.40}{\includegraphics{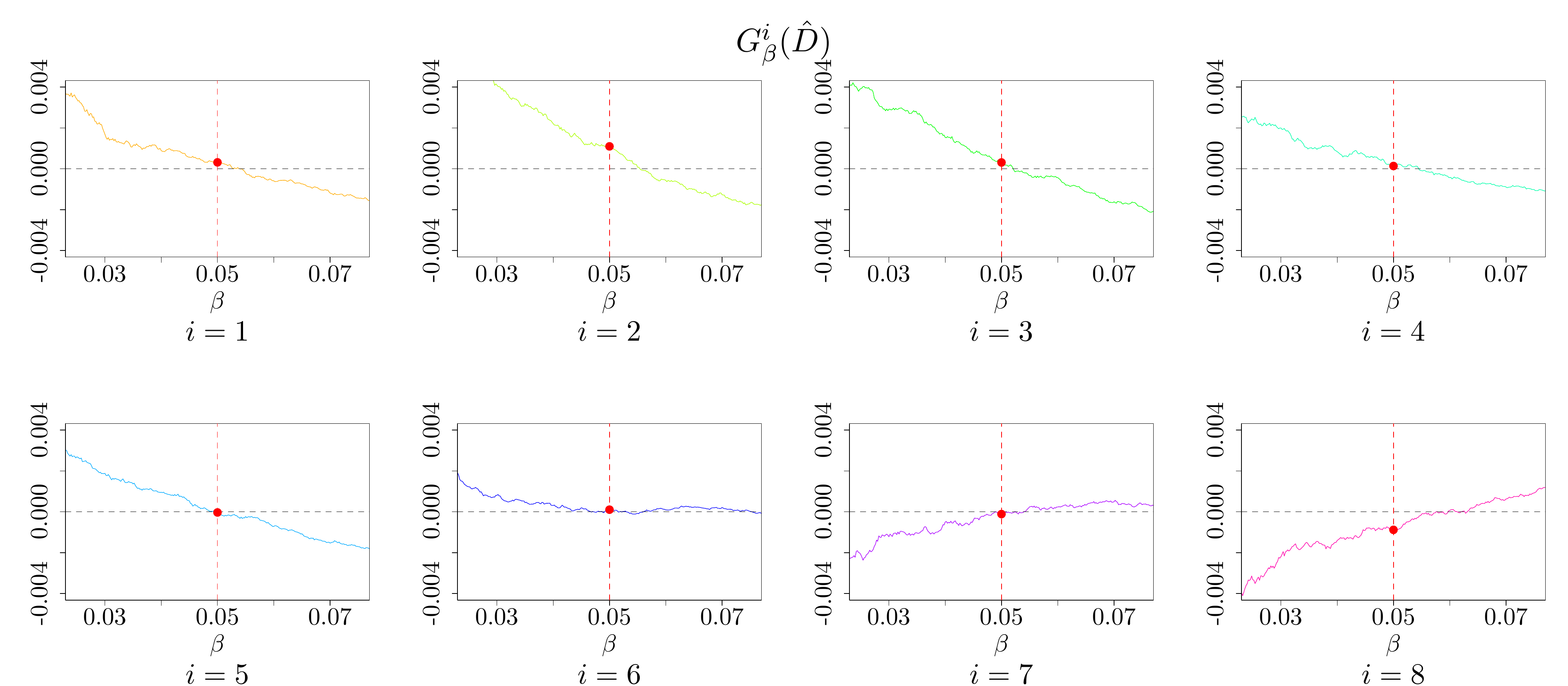}}
\end{center}
\caption{Example~\ref{ex:1}. Graphical representation of the deviation from fairness backtesting method, compare Table \ref{T:normal}: the first row shows  $G^i_{0.05}$, $i=1,\dots,8$, for the true allocation $a$, and the estimated risk allocations $\hat C$ and $\hat D$. The values are close to zero, indicating that the proposed backtesting methodologies are adequate. The second and third row shows
 $G^i_{\beta}$ as function of $\beta$ for each constituent. The red dots in the bottom rows represent the values of $G_{0.05}^i$ using $\hat D$. An aggregated plot together with $G_\beta$ is given in Figure \ref{fig:Ex1-2}.}
\label{fig:Ex1-3}
\end{figure}

\begin{figure}[htp!]
\begin{center}
\scalebox{0.24}{\includegraphics{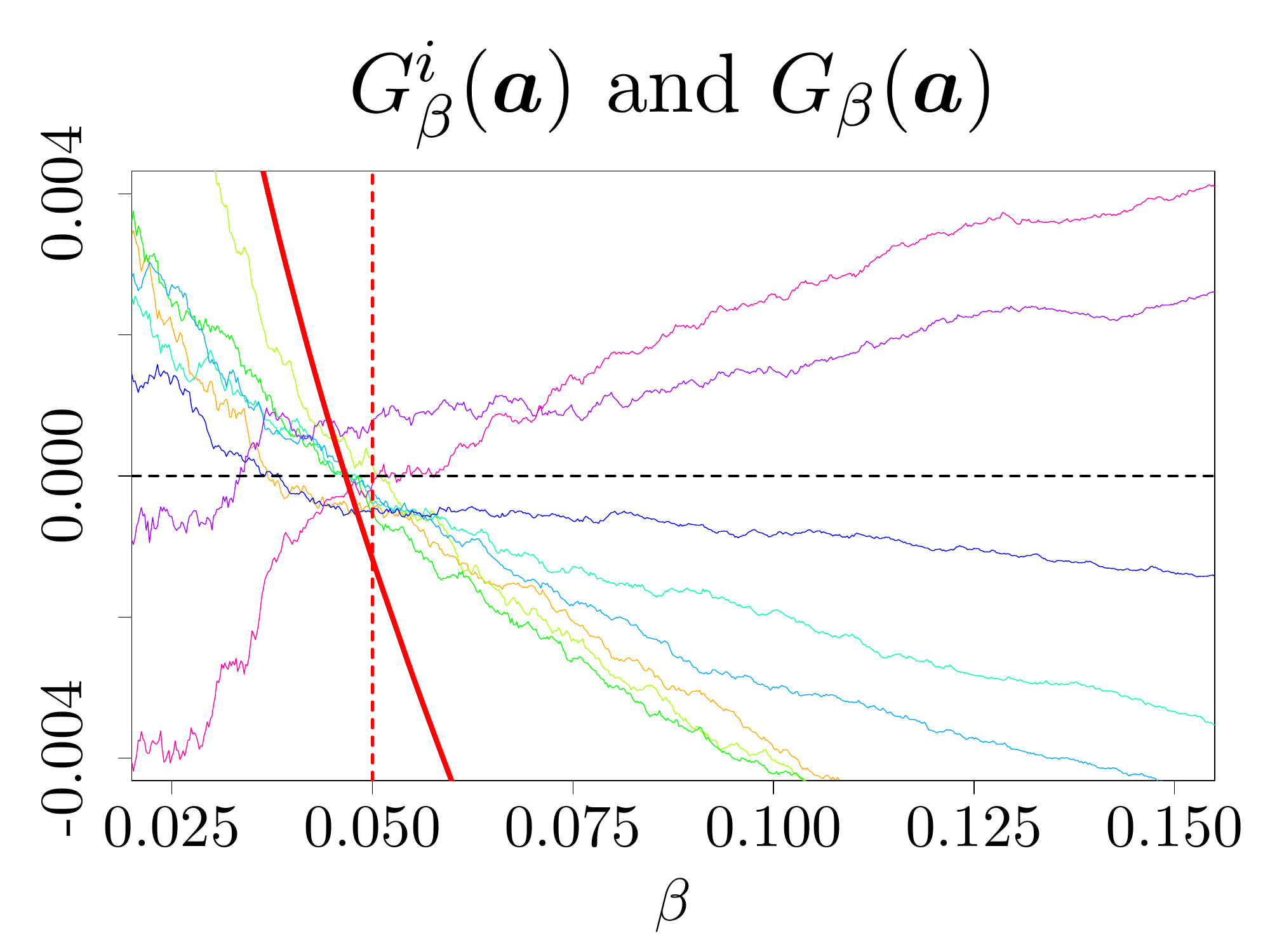}}
\scalebox{0.24}{\includegraphics{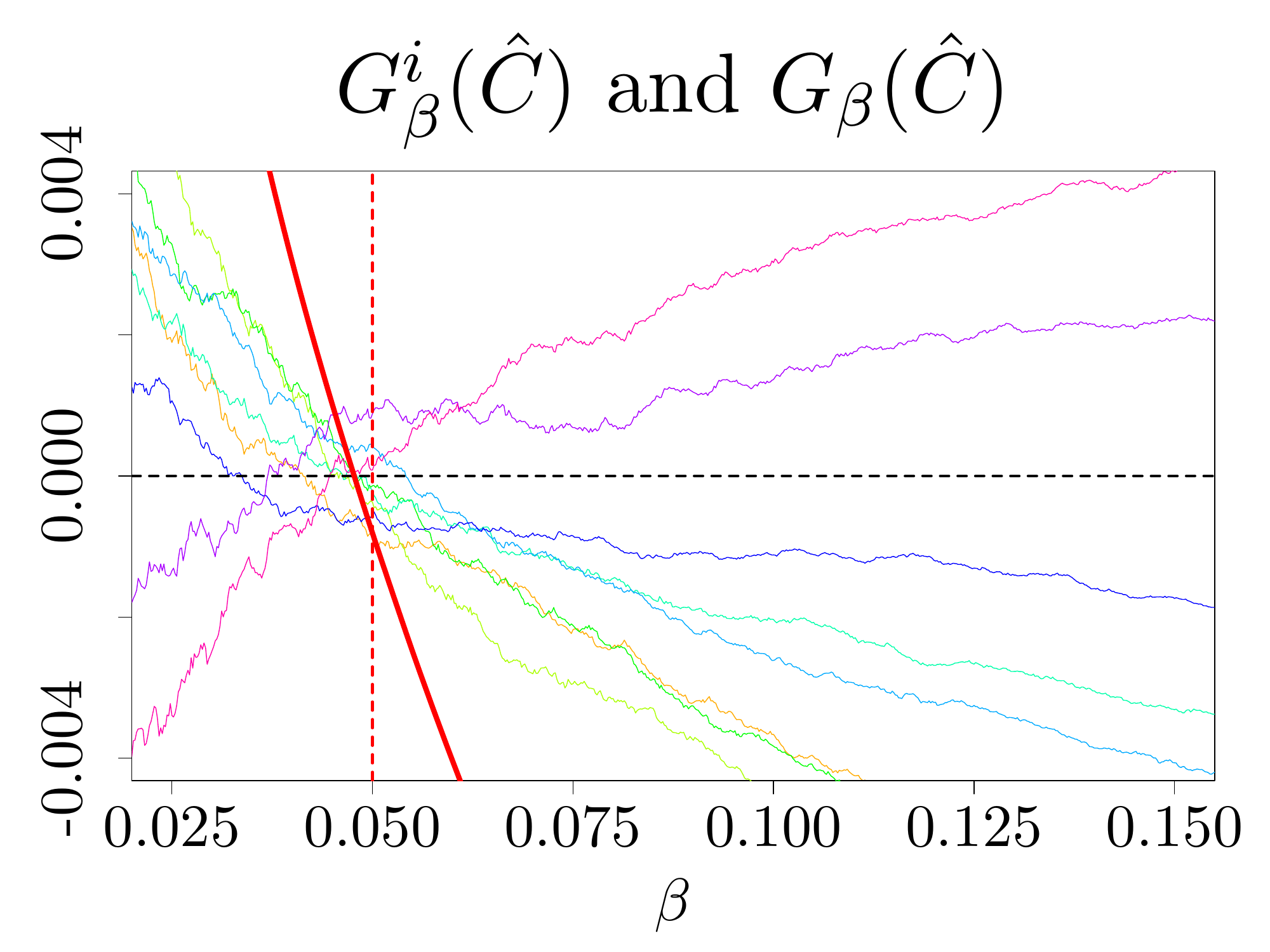}}
\scalebox{0.24}{\includegraphics{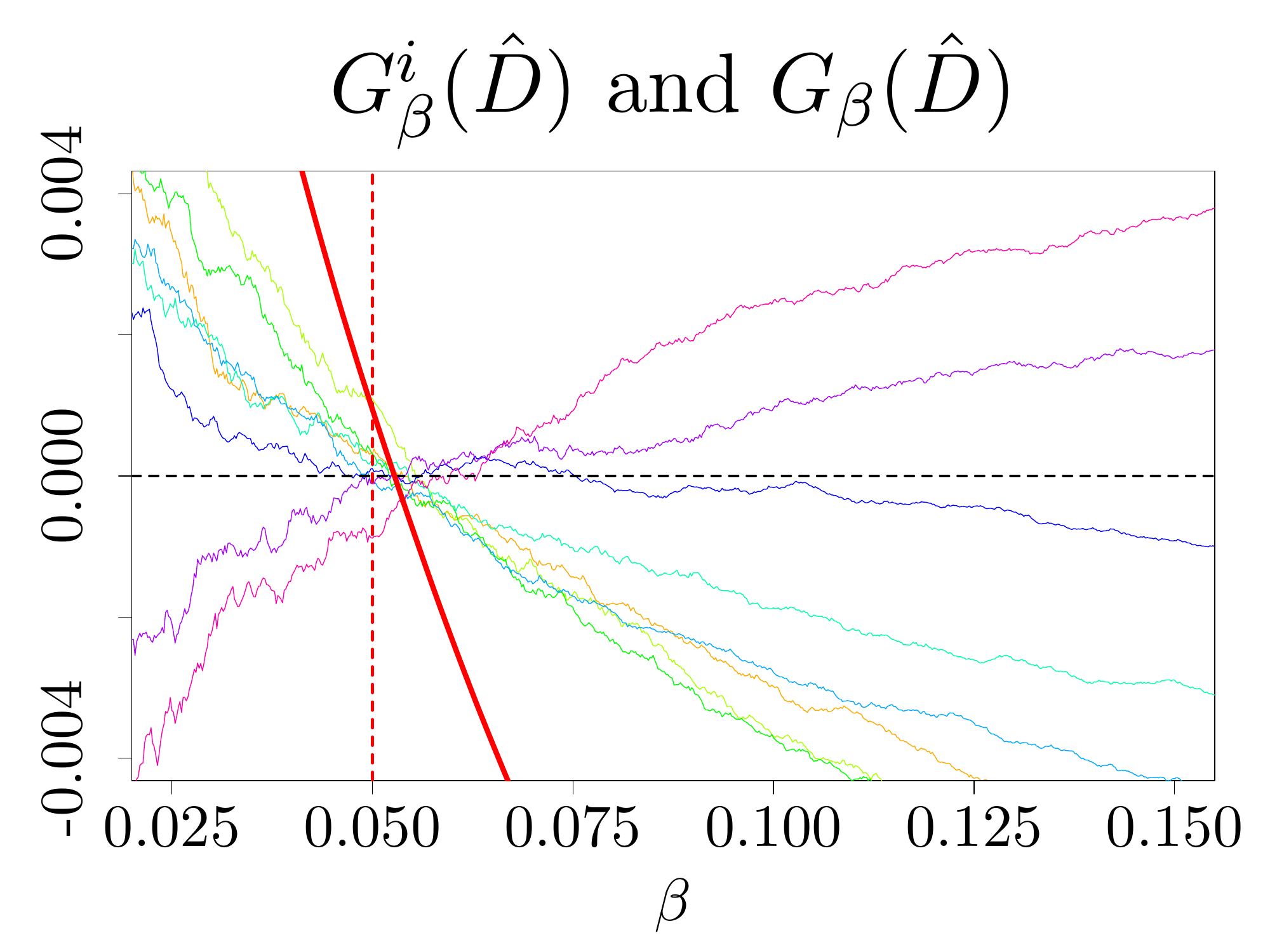}}
\end{center}
\caption{Example~\ref{ex:1}: the estimated backtesting measures as function of the risk level $\beta$ for the true allocation $a$, and the estimated risk allocations $\hat C$ and $\hat D$ (compare Table \ref{T:normal} for values corresponding to $\beta=0.05$).  The measures $G^i_{\beta}$, $i=1,2,\ldots,8$ for the different constituents are indicated by color while the  bold red line represents the backtesting measure $G_{\beta}$ at portfolio level. All these functions should be zero around $\beta=\alpha=0.05$, which is clearly the case.}
\label{fig:Ex1-2}
\end{figure}

 \begin{figure}[htp!]
\begin{center}
\scalebox{0.25}{\includegraphics{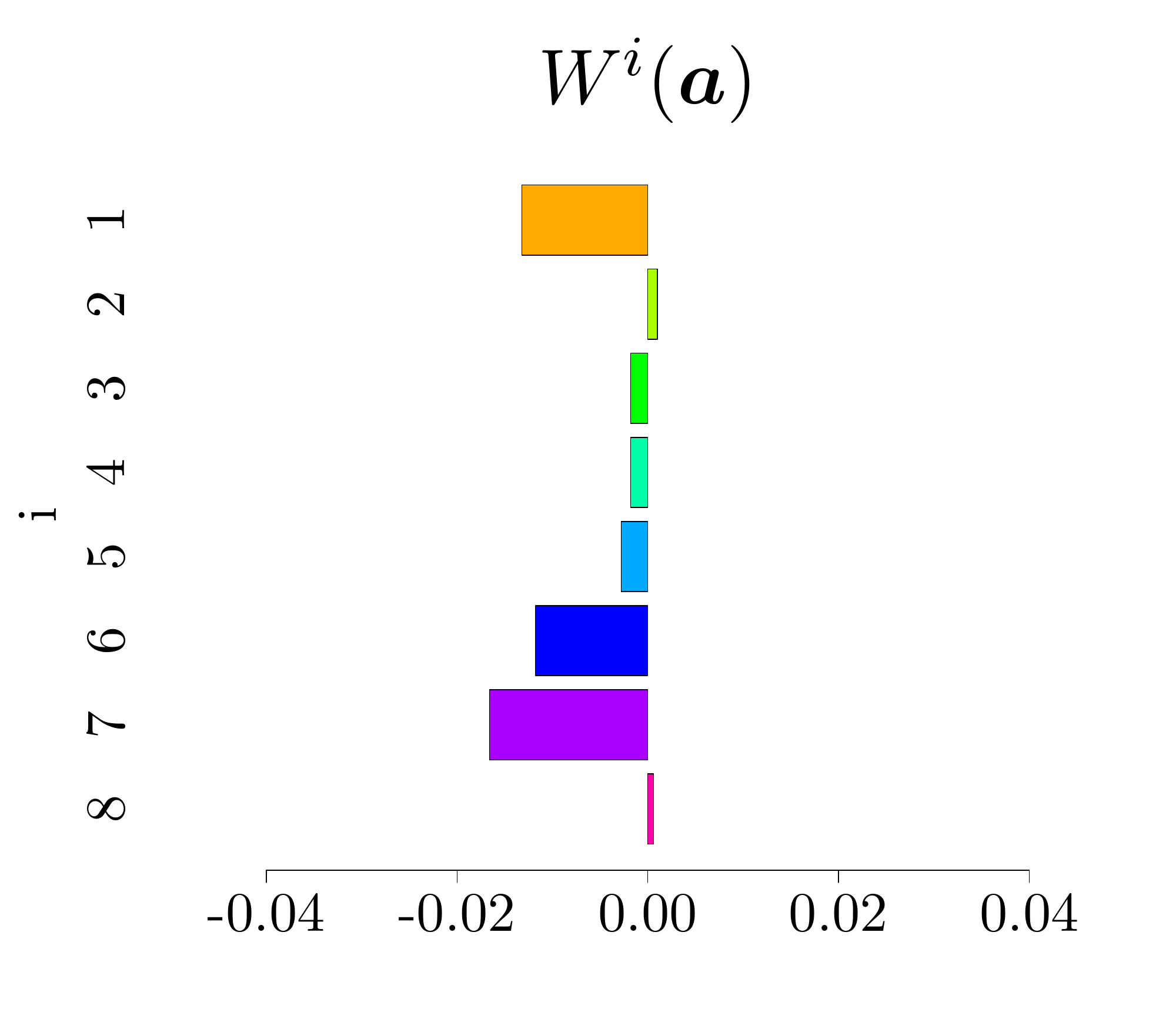}}
\scalebox{0.25}{\includegraphics{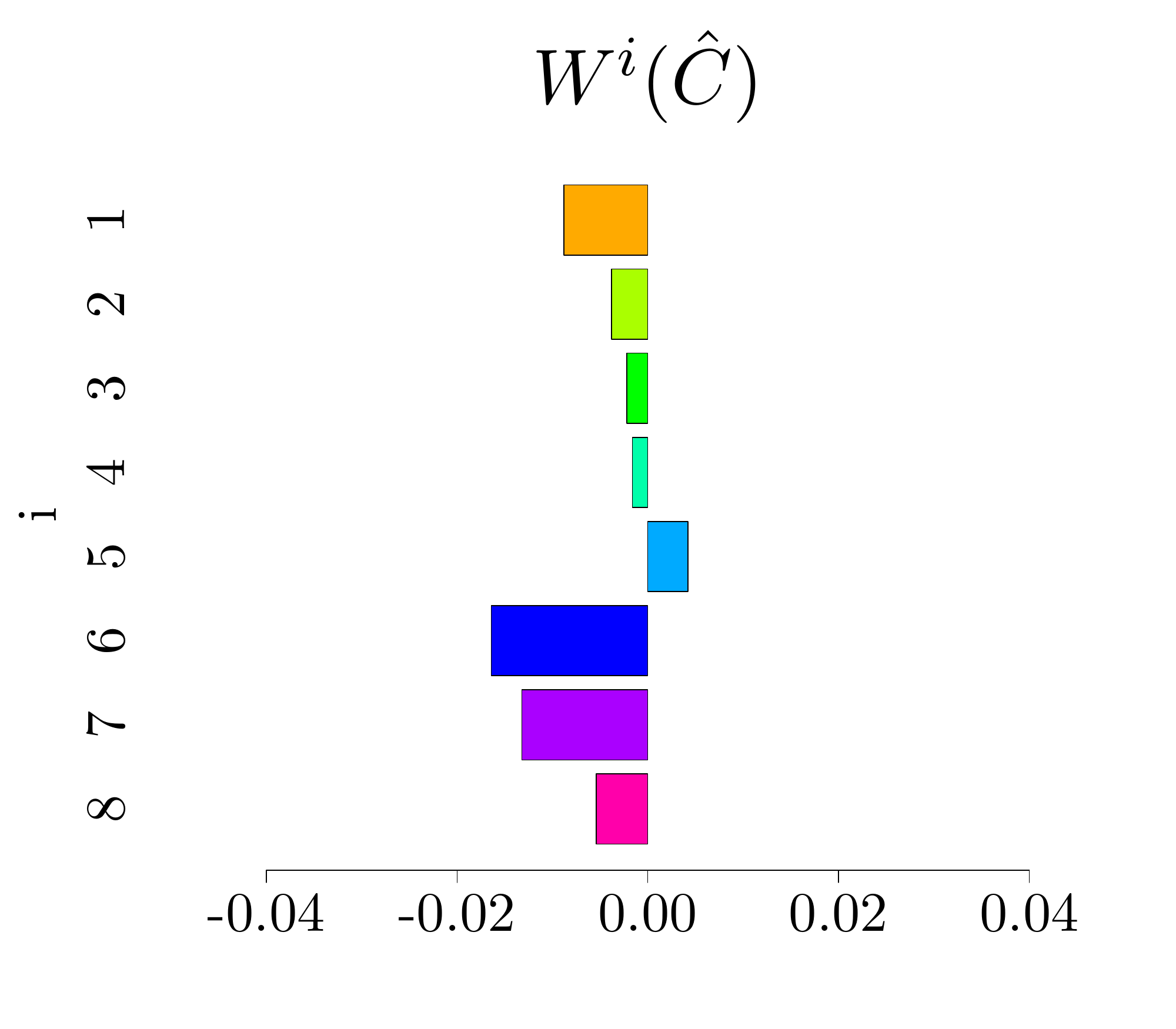}}
\scalebox{0.25}{\includegraphics{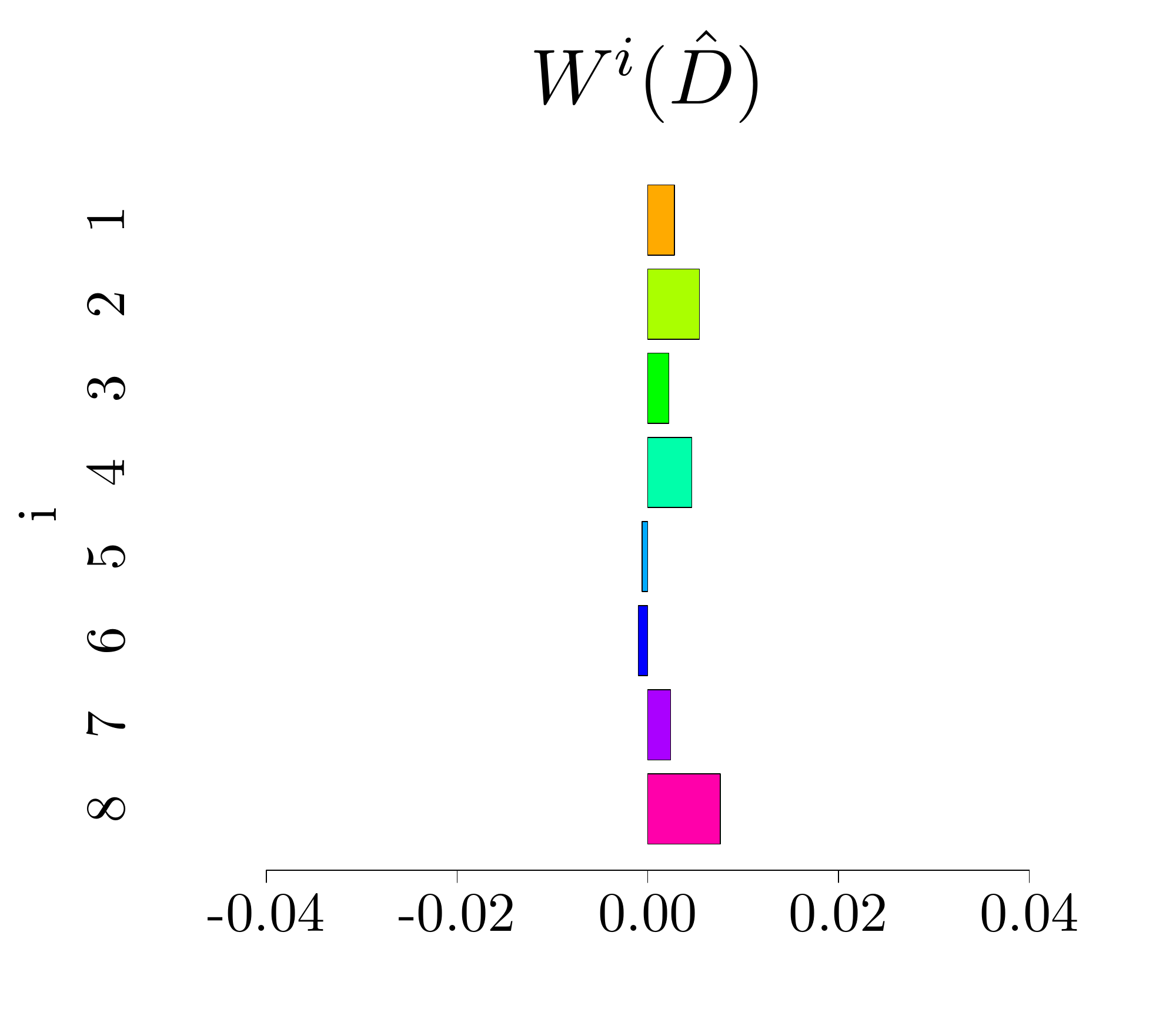}}\\
\vspace{0.6cm}
\scalebox{0.40}{\includegraphics{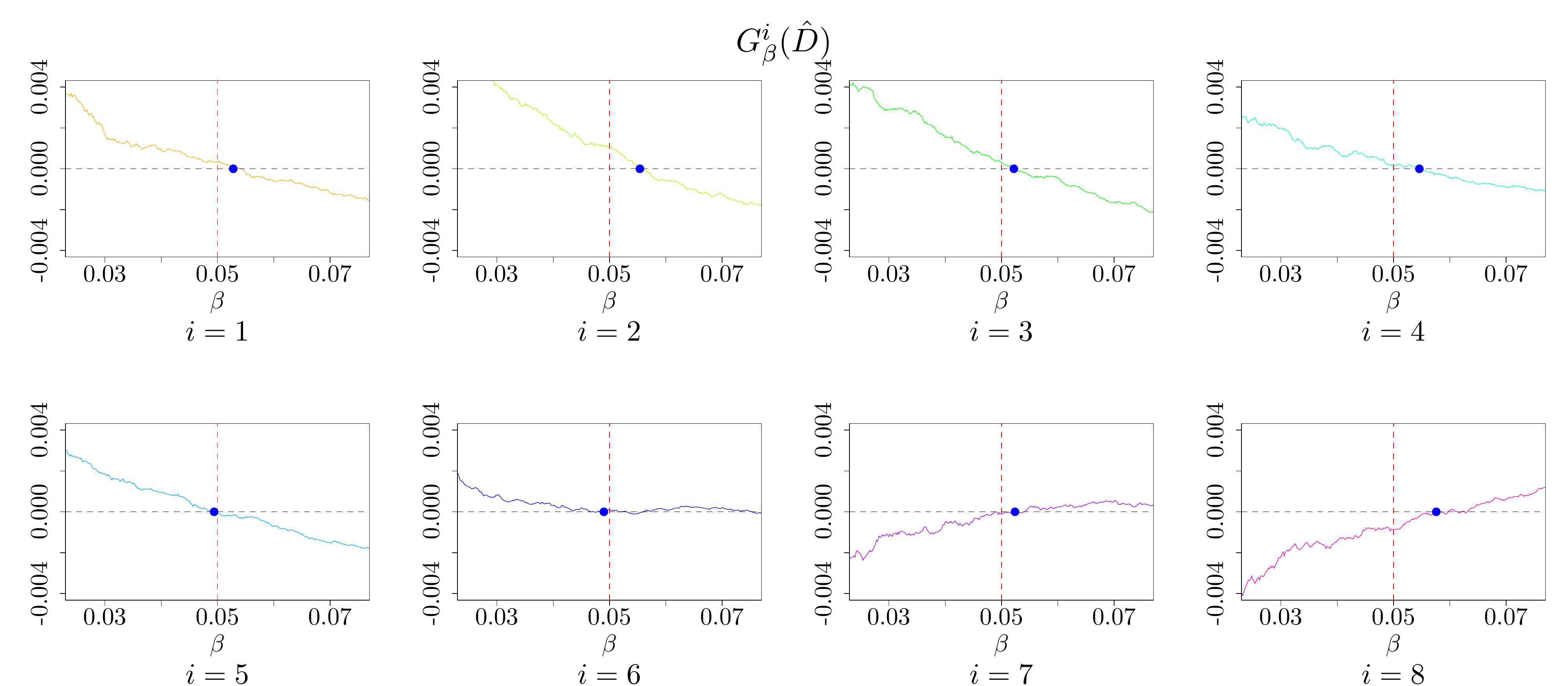}}
\end{center}
\caption{Example~\ref{ex:1}. Graphical representation of the risk level shift backtesting method, compare Table \ref{T:normal}: the first row shows  $W^i$, $i=1,\dots,8$, values being  close to zero. The second and third row shows
 $G^i_{\beta}(\hat D)$ as function of $\beta$ for each constituent, blue dots represent the values of $W^i(\hat{D})$.}
\label{fig:Ex1-4}
\end{figure}

\end{example}


\clearpage

\begin{figure}[htp!]
\begin{center}
\scalebox{0.25}{\includegraphics{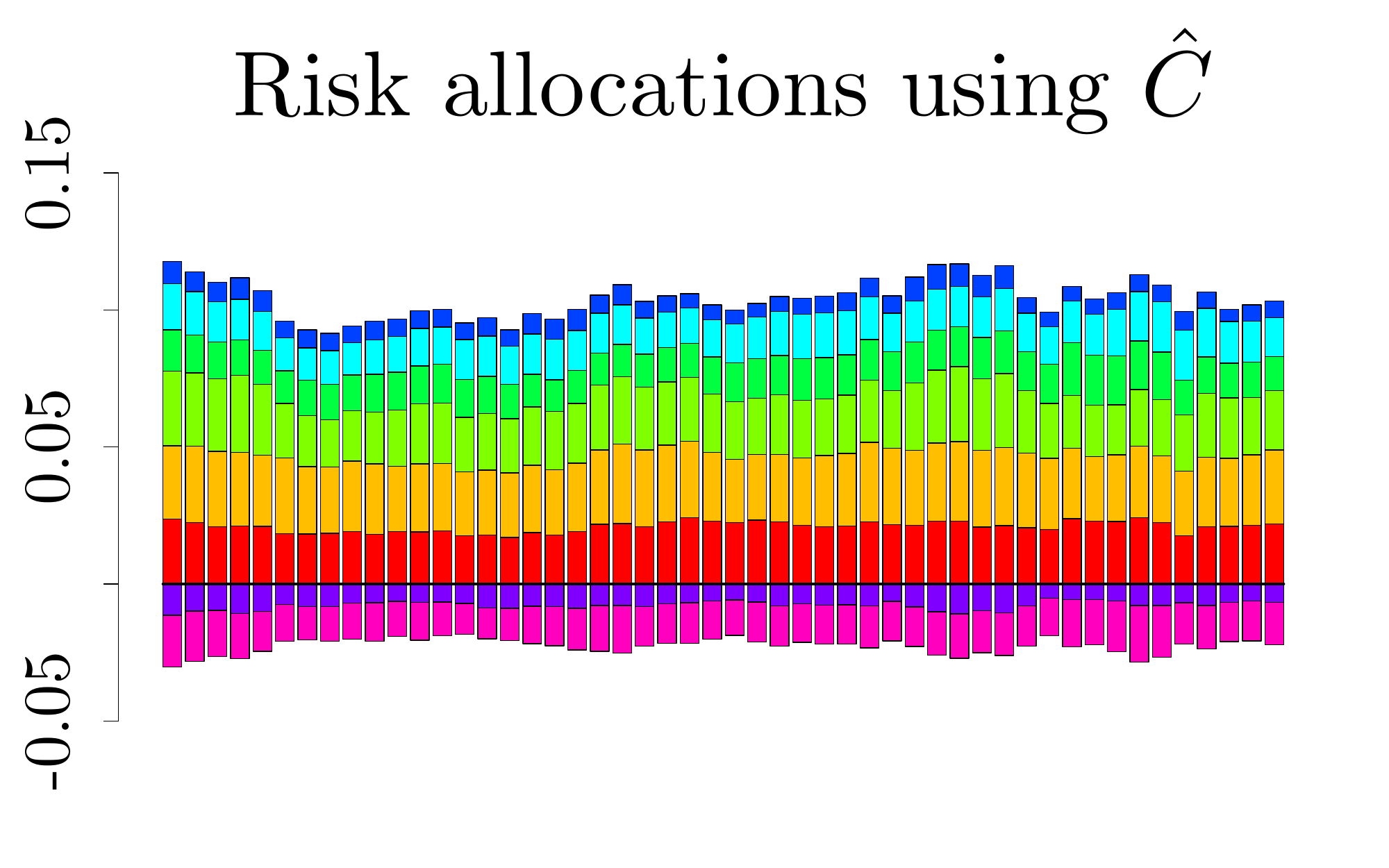}}
\scalebox{0.25}{\includegraphics{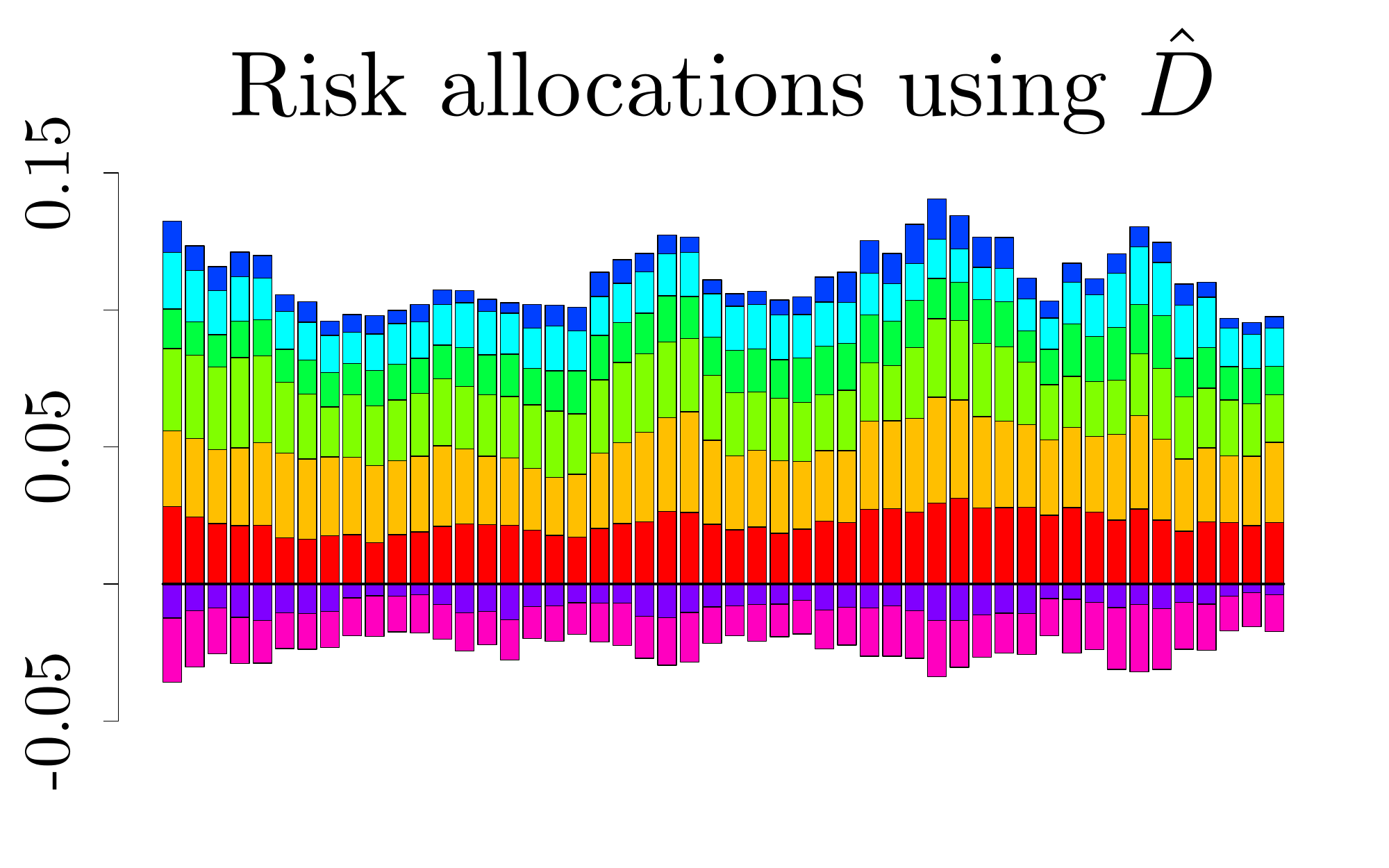}}
\end{center}
\caption{Example~\ref{ex:2}. Estimated risk allocations for the eight portfolio constituents (indexed by colour) at each backtesting day, $k=1,\ldots,m$for the estimated risk allocations $\hat C$ and $\hat D$;  the height of each coloured  horizontal layer represents the risk allocated to one of the constituents. It is apparent that the estimated risk allocation by these two methods are quite different.}
\label{fig:Ex2-1}
\end{figure}

\begin{table}[htp!]
\centering
\scalebox{0.8}{
\begin{tabular}{lllllllllll} \toprule
& $X_1$ & $X_2$ & $X_3$ & $X_4$ & $X_5$ & $X_6$ & $X_7$ & $X_8$ & & $S$\\
   \midrule
 $G^i_{0.05}(\hat{C})$ &  0.00209 &  0.00363 &  0.00152 &  0.00127 &  0.00093 &  0.00088 & -0.00084 & -0.00076 & $G_{0.05}(\hat{C})$ & 0.00872 \\
   $W^i(\hat{C})$ & 0.016 & 0.025 & 0.013 & 0.012 & 0.012 & 0.021 & 0.014 & 0.012 & $\Upsilon(\hat{C})$ & 0.069 \\
   \midrule
  $G^i_{0.05}(\hat{D})$ &  0.00074 &  0.00062 & -0.00030 &  0.00028 &  0.00049 &  0.00014 &  0.00001 &  0.00013 & $G_{0.05}(\hat{D})$ & 0.00212 \\
   $W^i(\hat{D})$ &  0.004 &  0.003 & -0.001 &  0.008 &  0.007 &  0.001 &  0.000 & -0.001 & $\Upsilon(\hat{D})$ & 0.054 \\
   \bottomrule
\end{tabular}
}
\caption{Summary of the estimated backtesting measures for Example~\ref{ex:2}: In the first columns we show $G^i_{0.05}$ and $W^i$, $i=1,\dots,8$, for the estimated risk allocations $\hat C$ and $\hat D$, corresponding to backtesting the fair allocation. The values of $G^i_{0.05}(\hat{C})$ are of one order of magnitude further away from zero than $G^i_{0.05}(\hat{D})$, indicating that indeed risk allocation methodology $\hat D$ is more adequate for this experiment. The last column shows the aggregated quantities $G_{0.05}$ and the risk level shift $\Upsilon$. Here, $\Upsilon(\hat{D})$ is close to $\alpha=0.05$, as expected.}
\label{T:student}
\end{table}

\begin{example}[Student $t$-distributed P\&Ls]\label{ex:2}
Similar to the previous example we consider a portfolio of eight constituents and with discounted P\&L following a $t$-distribution with five degrees of freedom. For comparison reasons, the distribution of $(X^1,\ldots,X^8)$ is modified so that it has the same mean and variance covariance structure as in Example~\ref{ex:1}.

First,  note that there is no available counterpart of $a$ for this setup. Second, as we will show below, since $X$ does not follow a Gaussian distribution, one should not use $\hat C$ to estimate the risk allocation, and only $\hat D$ is an appropriate methodology in estimating risk allocation. In Figure~\ref{fig:Ex2-1}, we present the estimated risk allocations computed using $\hat C$ and $\hat D$, over the entire backtesting period $k=1,\ldots,m$. It is apparent that the estimated risk allocation by these two methods are quite different. Table~\ref{T:student} contains the values of the estimated backtesting metrics, and for the reader's convenience $G^i_{0.05}$ and $W^i$ are represented graphically in Figure~\ref{fig:Ex2-2}. The values of $G^i_{0.05}(\hat{C})$ are of one order of magnitude further away from zero than $G^i_{0.05}(\hat{D})$, indicating that indeed risk allocation methodology $\hat D$ is more adequate for this experiment. We also note that magnitude of  $G^i_{0.05}(\hat{D})$ in this example aligns with the benchmark values from Example~\ref{ex:1}. Similar arguments hold true for $W^i$ and $\Upsilon$.

\clearpage

\begin{figure}[htp!]
\begin{center}
\scalebox{0.25}{\includegraphics{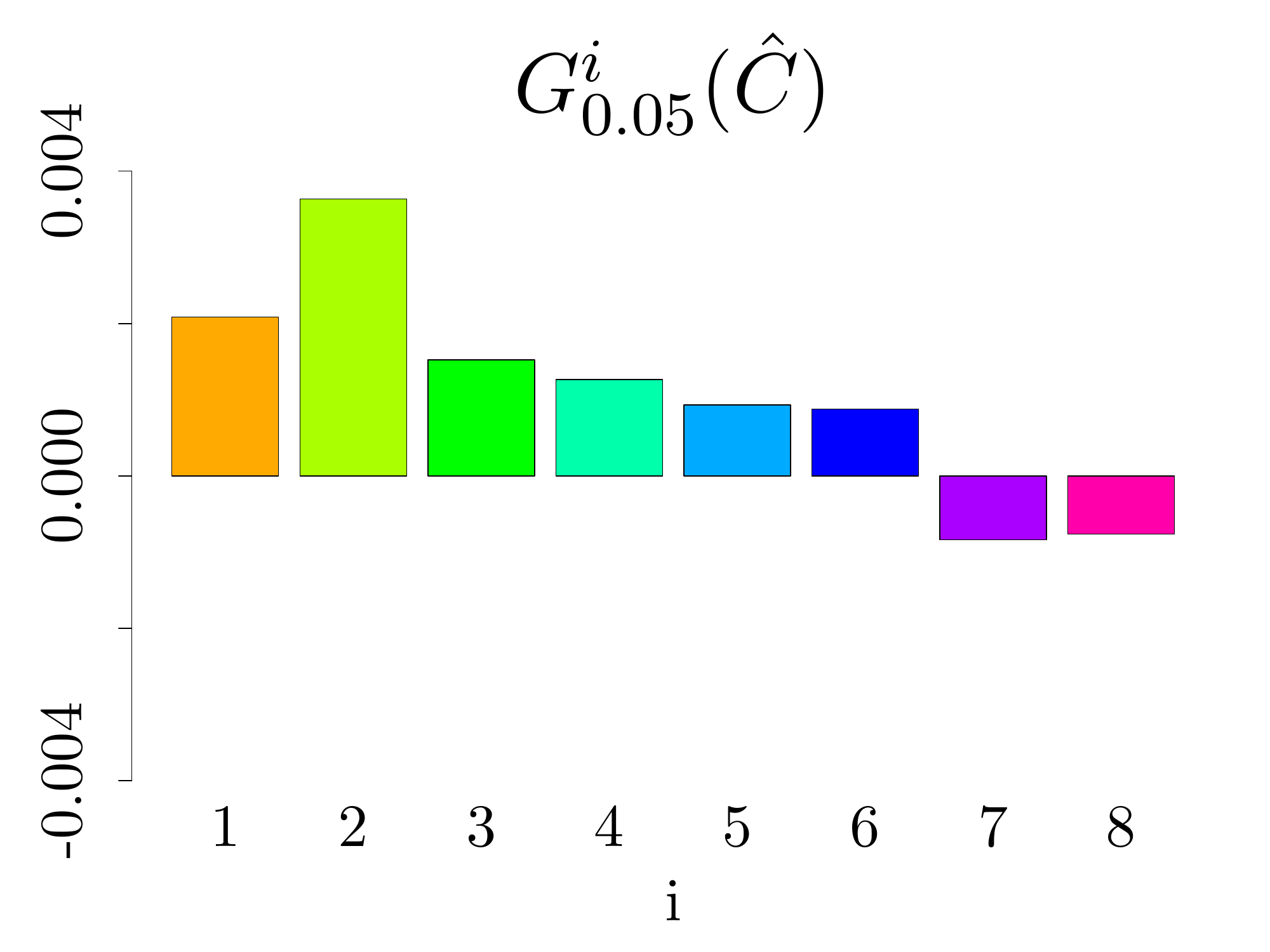}}
\scalebox{0.25}{\includegraphics{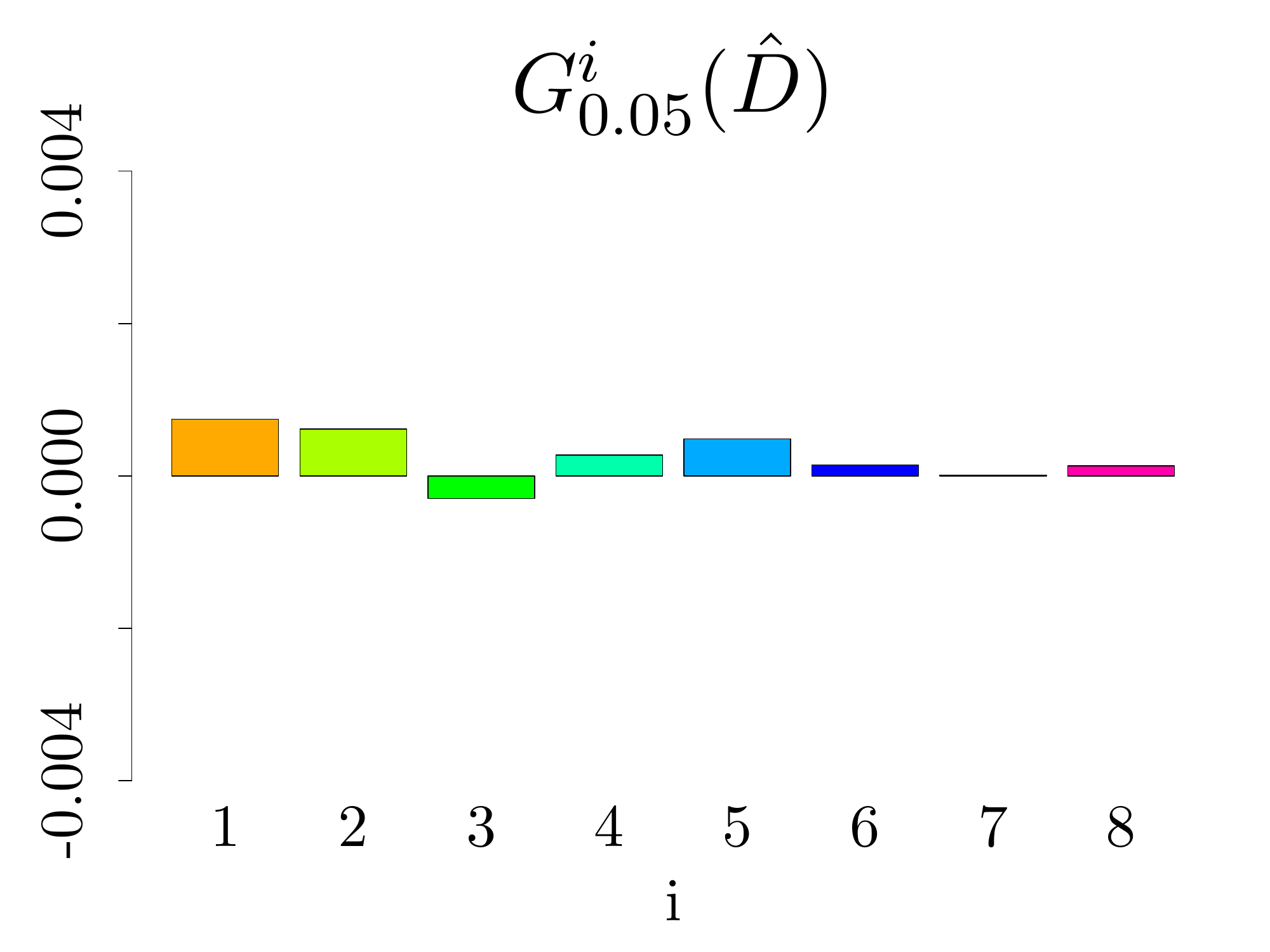}}\\
\scalebox{0.25}{\includegraphics{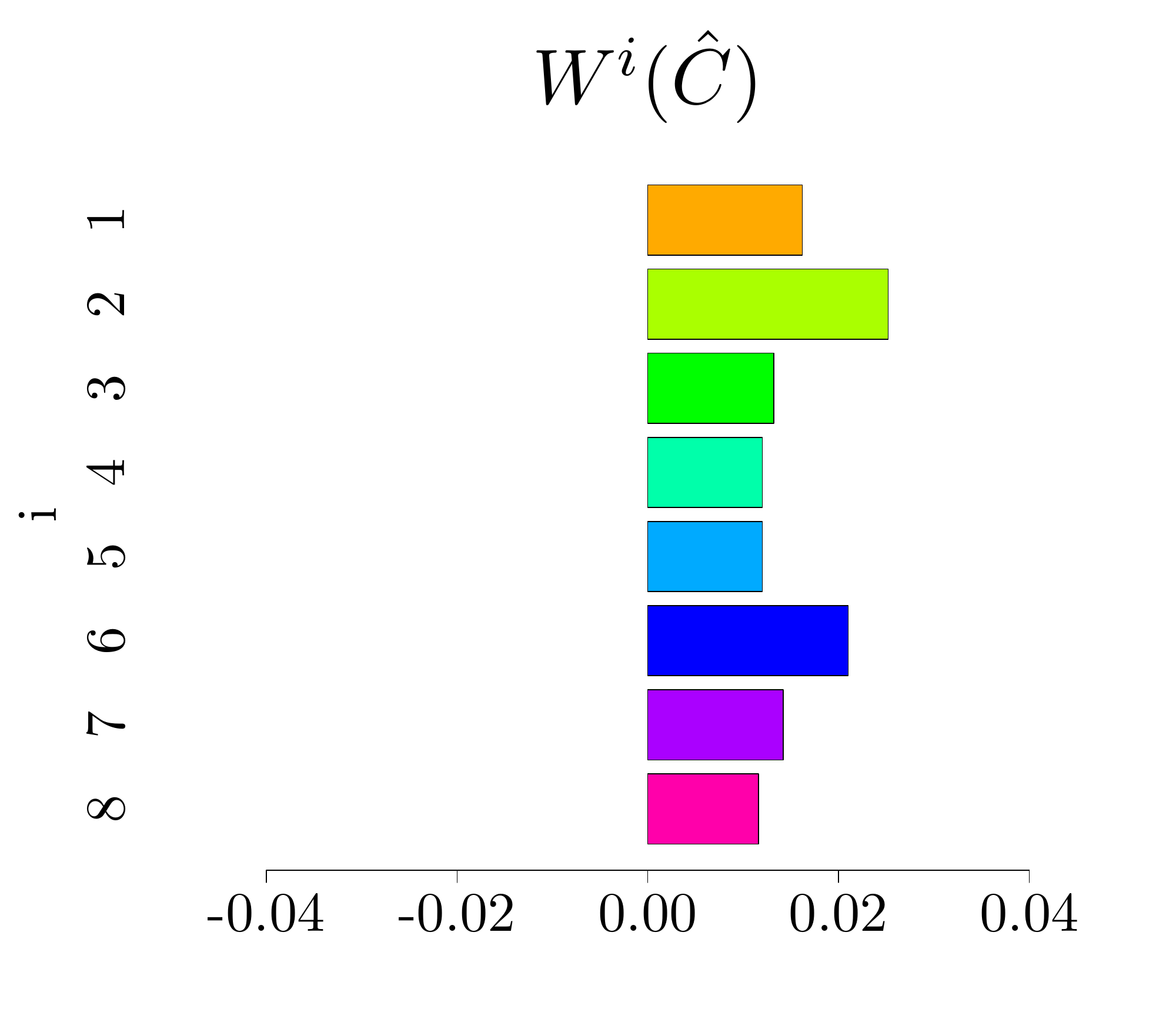}}
\scalebox{0.25}{\includegraphics{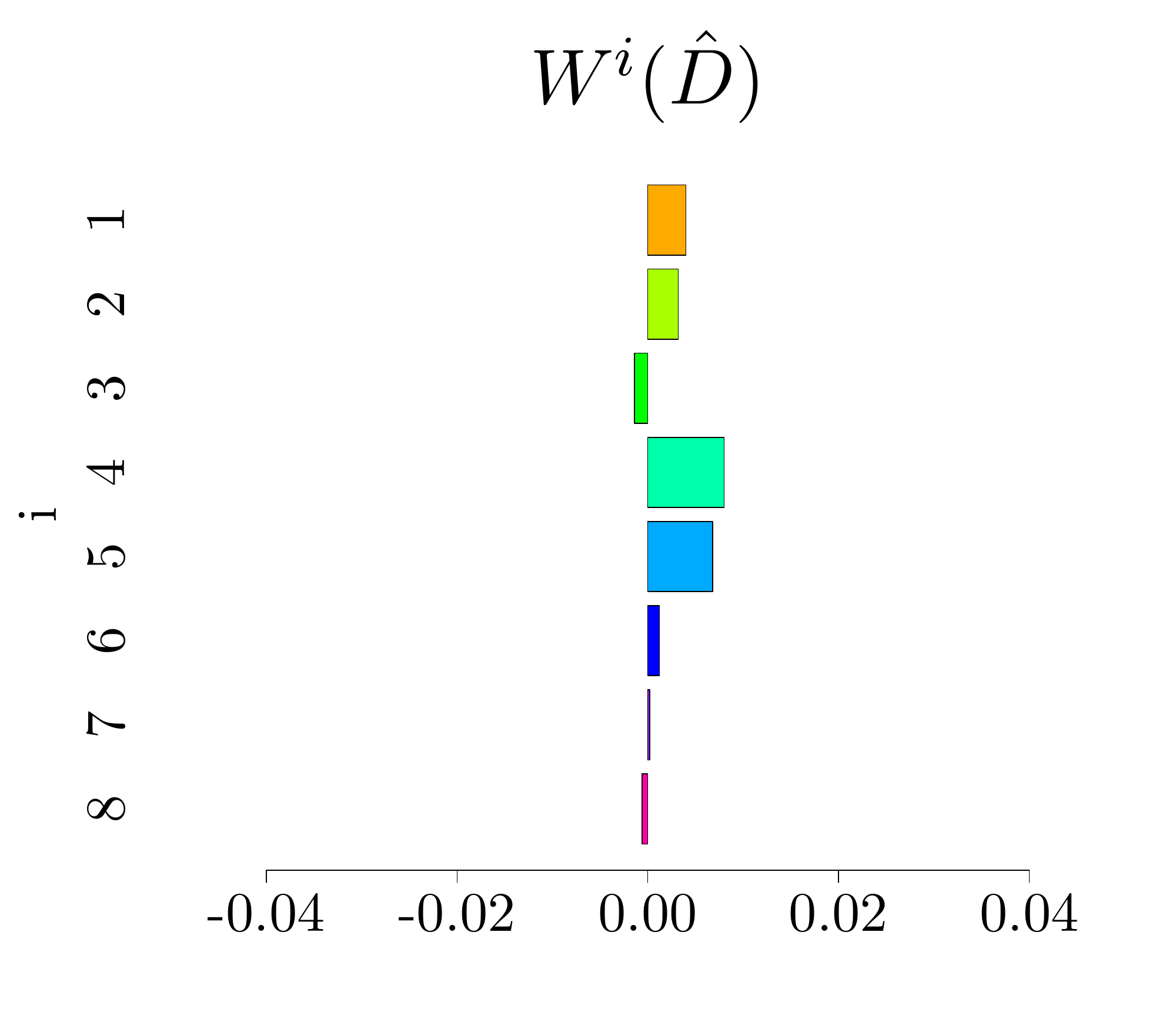}}
\end{center}
\caption{Estimated backtesting measures for Example~\ref{ex:2}, compare Figures \ref{fig:Ex1-1} and \ref{fig:Ex1-2}. The first row represents the deviation from fairness backtesting method and shows $G^i_{0.05}$ for each constituent $i=1,\dots,8$ for the risk allocations $\hat C$ and $\hat D$. The second row represents the risk level shift backtesting method and shows $W^i$, respectively. The values of $G^i_{0.05}(\hat{C})$ are of one order of magnitude further away from zero than $G^i_{0.05}(\hat{D})$, indicating that indeed risk allocation methodology $\hat D$ is more adequate for this experiment. We also note that magnitude of  $G^i_{0.05}(\hat{D})$ in this example aligns with the benchmark values from Example~\ref{ex:1}. Similar arguments hold true for $W^i$.}
\label{fig:Ex2-2}
\end{figure}
\end{example}

\begin{example}[Fairness and asymptotic fairness]\label{ex:3}
In this example we illustrate the fairness and the asymptotic fairness properties.
Again, for the sake of a reference statistic which eases the presentation,   we work under the normality assumption. Moreover,  we consider only the first three  constituents from Example~\ref{ex:1}, that is $(X^1,X^2,X^3)$, because the other constituents show similar behavior. The numerical results presented below confirm that allocations $a$ and $\hat B$ are fair. In addition, these results confirm that  the allocations $\hat C$ and $\hat D$ are asymptotically fair even though they are not fair in this example.

Figure \ref{F:asympt.G} deals with the issue of a short learning period, that is a small sample size, of $n=250$. We see that for allocations $a$ and $\hat B$  the $G^i_{0.05}$'s and  $W^i$'s are getting close to zero with increasing $m$, and that $\Upsilon$ gets close to $0.05$ with increasing $m$,  confirming that these are fair allocations. We also see that $G^i_{0.05}$'s and $W^i$'s stay away from zero, and $\Upsilon$ stays away from $0.05$ with increasing $m$ for allocations $\hat C$ and $\hat D$, indicating that these are not fair allocations.

Figure~\ref{F:consistency} illustrates the asymptotic fairness of $\hat D^n$ with $n\rightarrow \infty$. The left panel shows that $G^i_{0.05}(\hat{D})$ get closer to zero for large $m$ with increasing $n$. Similarly for the right panel, with regard to  $W^i$  and $\Upsilon$.

\begin{figure}[htp!]
\begin{center}
\scalebox{0.30}{\includegraphics{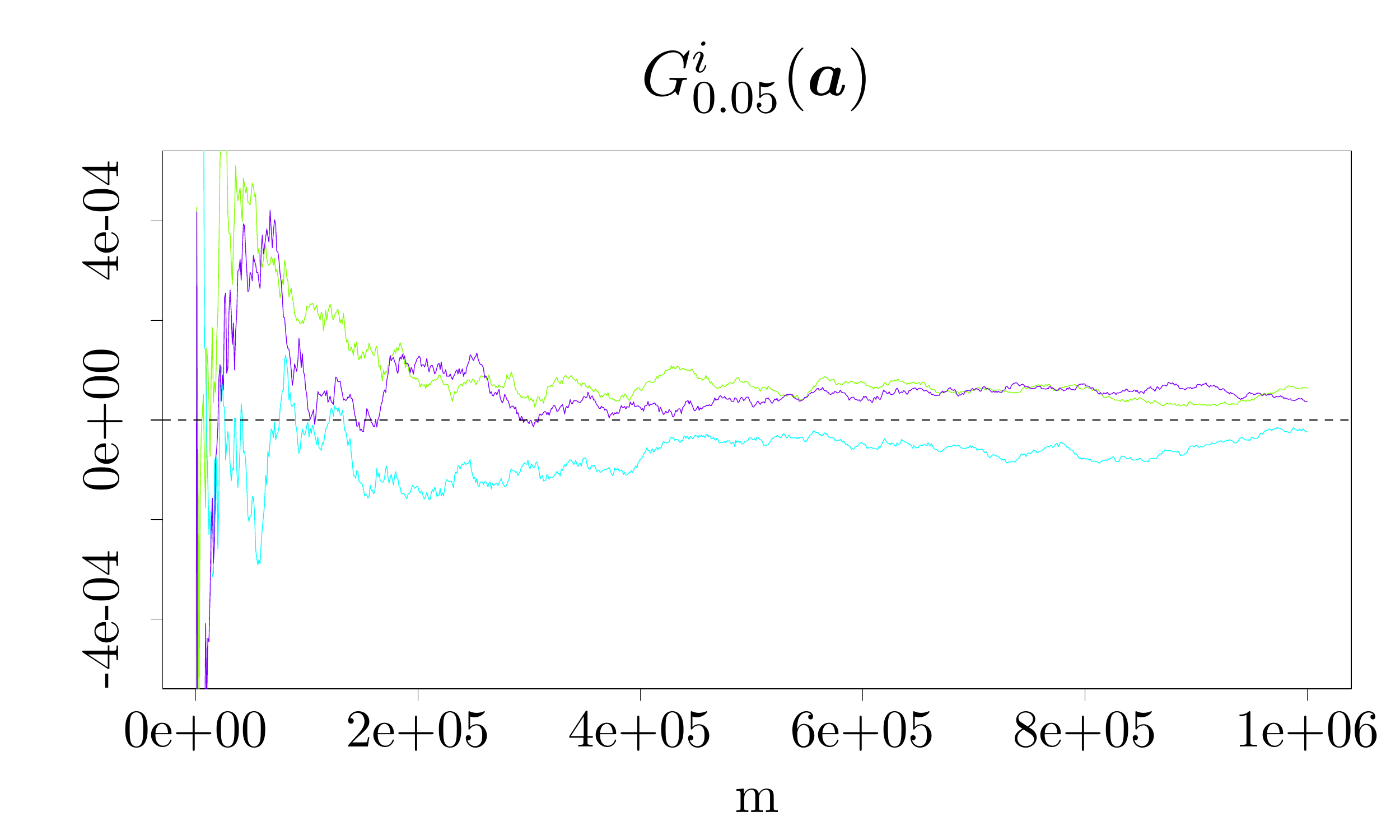}}
\scalebox{0.30}{\includegraphics{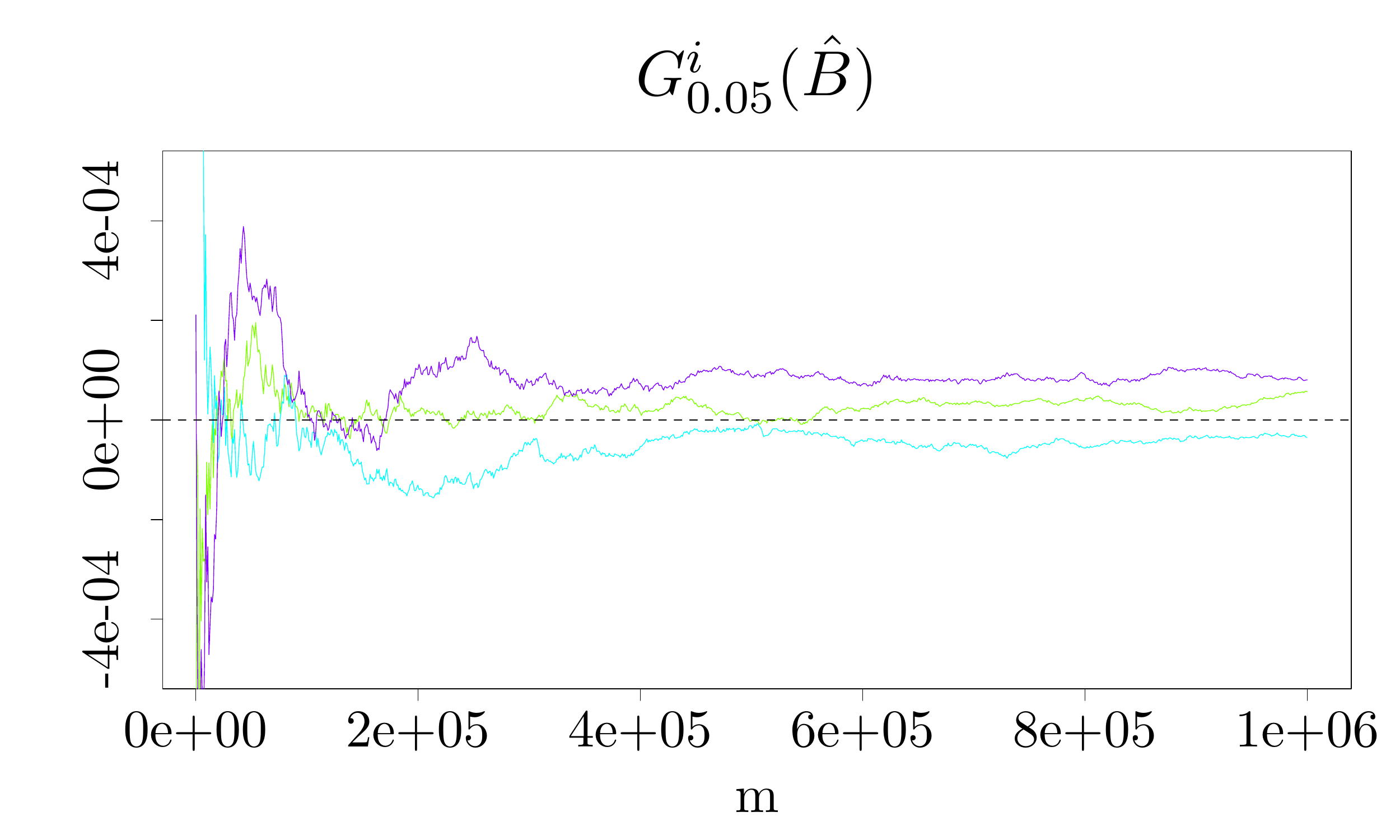}}\\
\scalebox{0.30}{\includegraphics{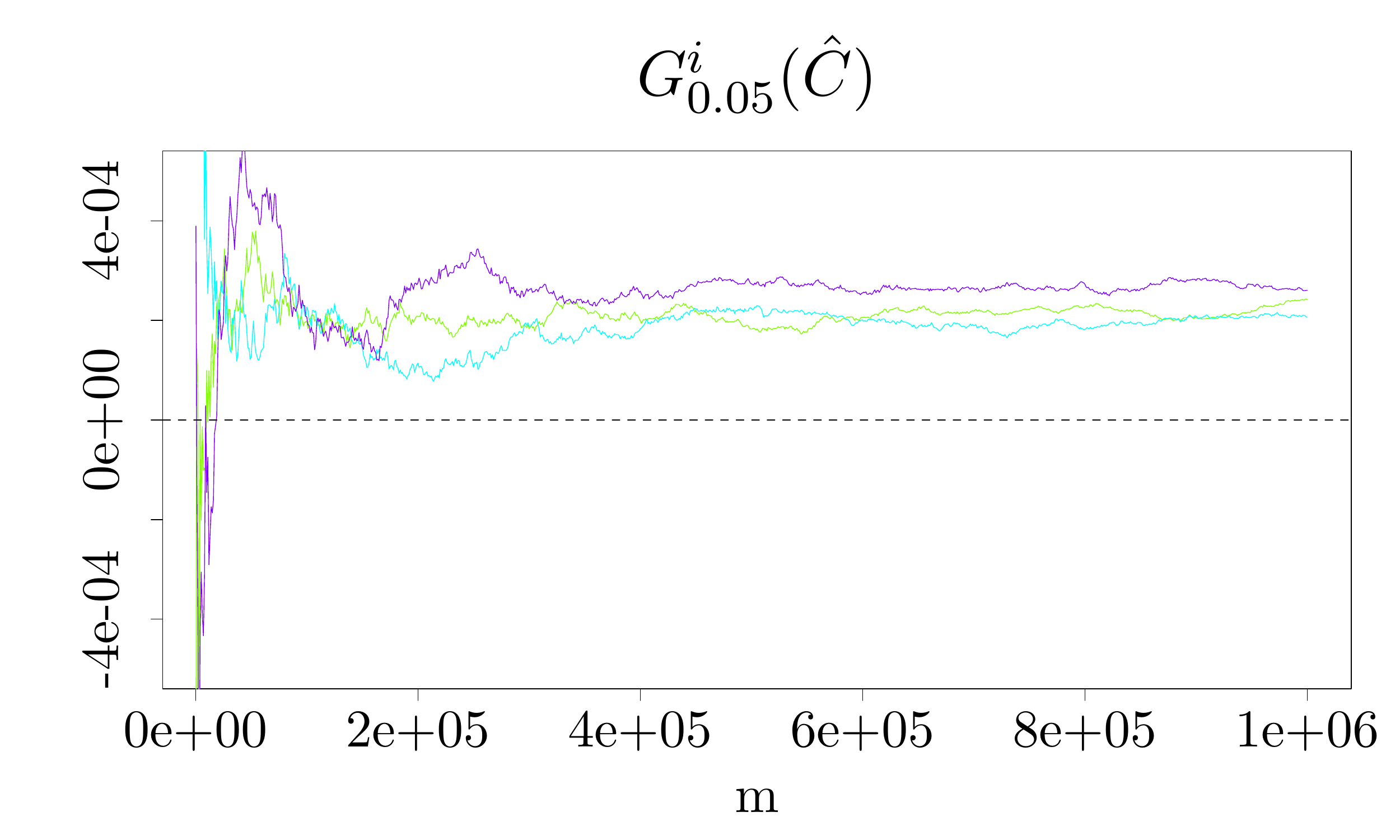}}
\scalebox{0.30}{\includegraphics{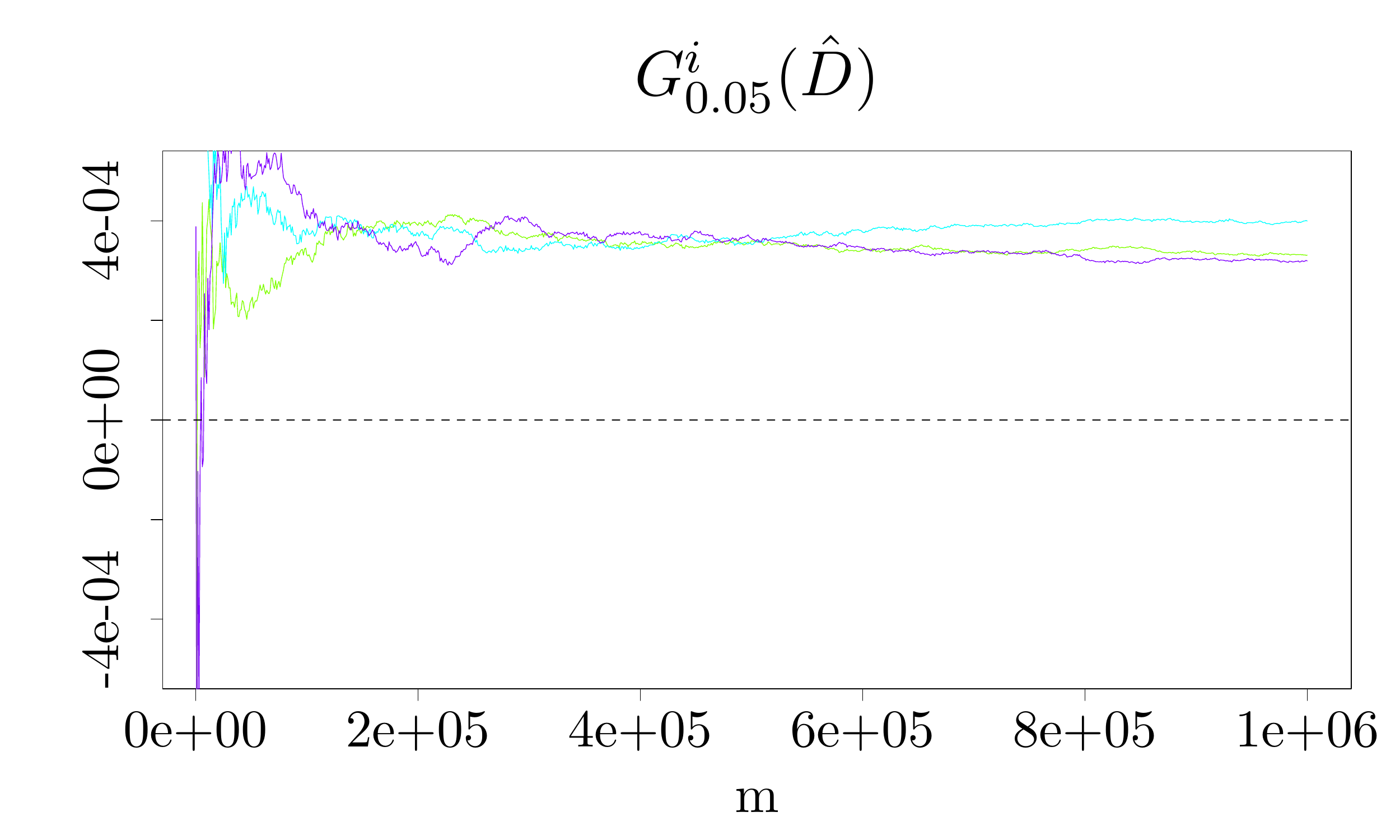}} \\
\scalebox{0.30}{\includegraphics{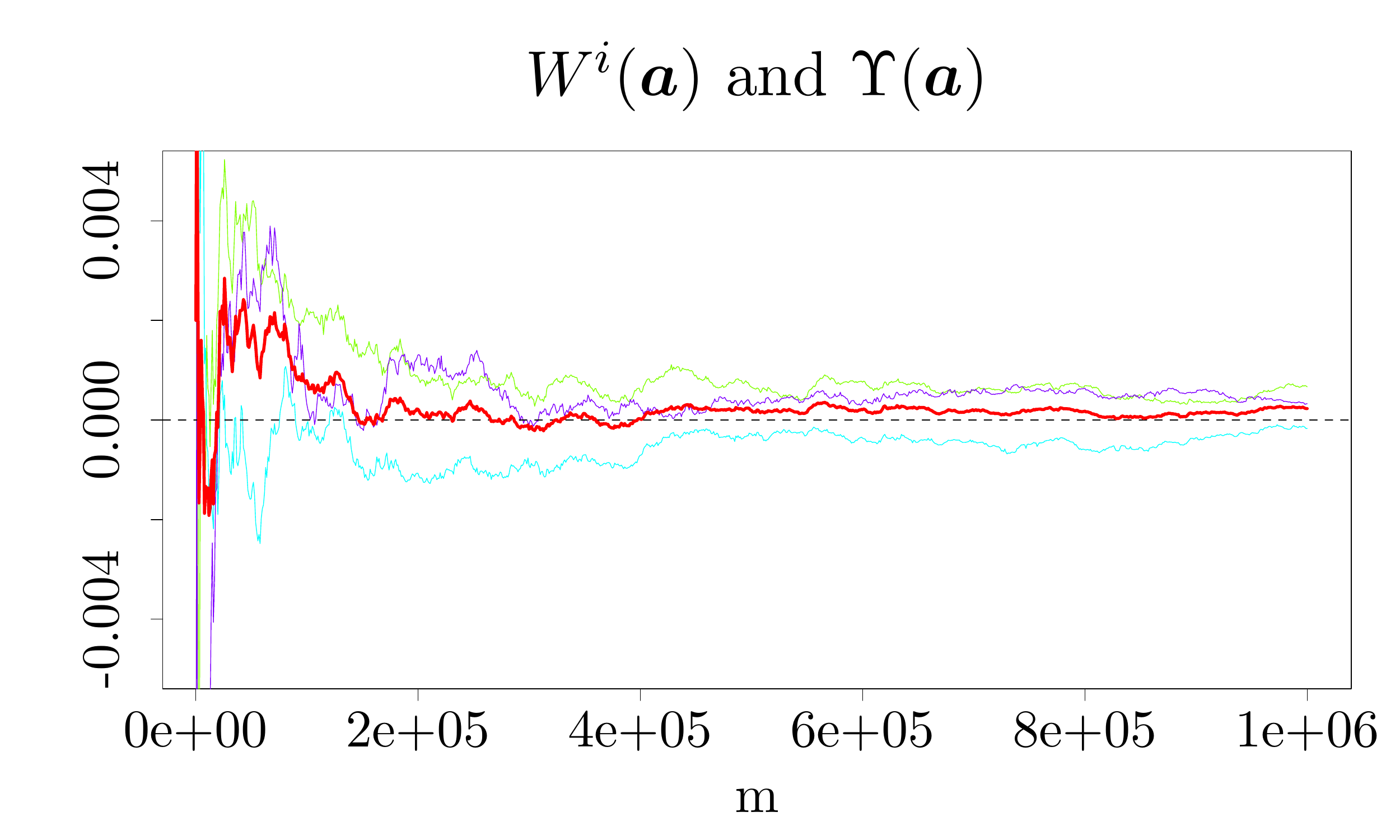}}
\scalebox{0.30}{\includegraphics{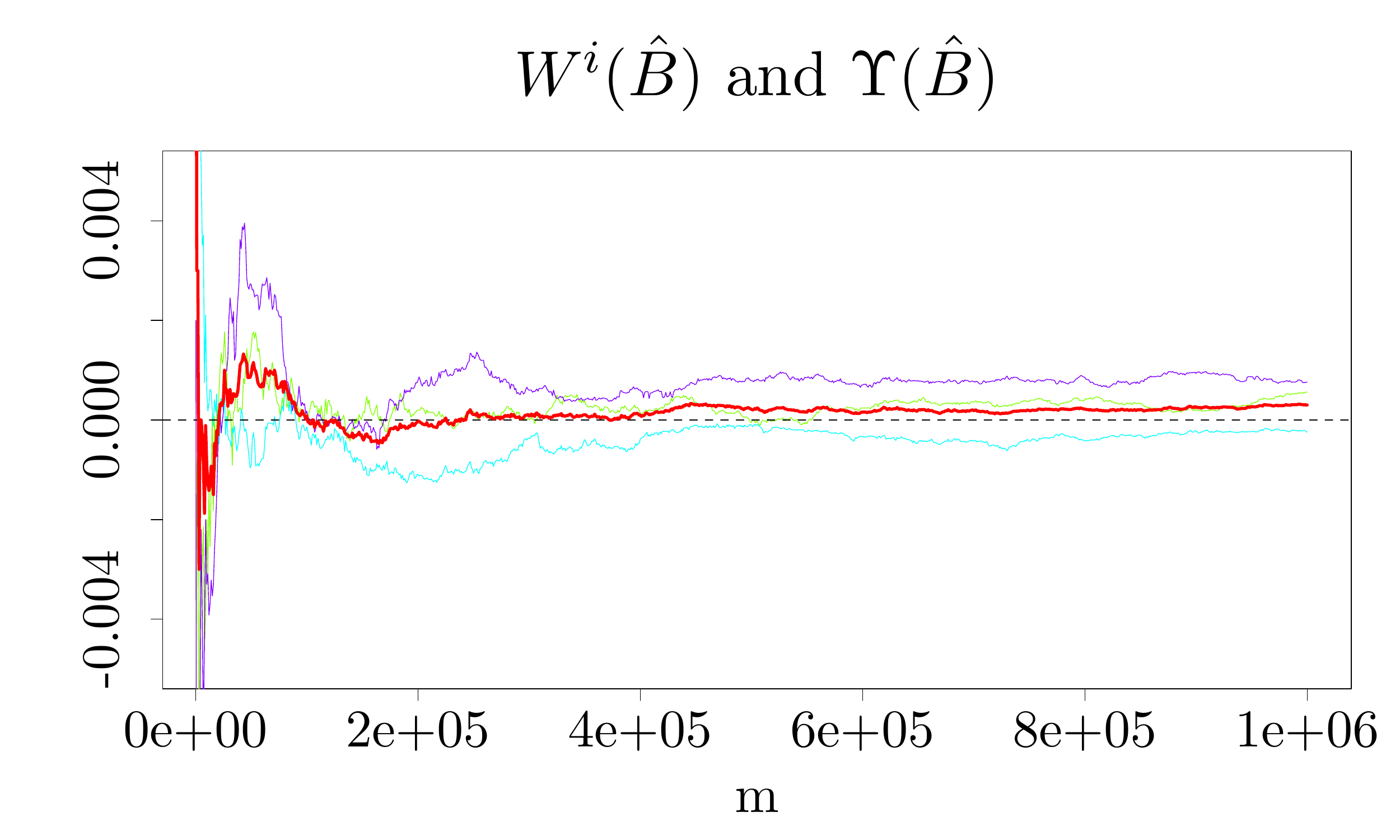}}\\
\scalebox{0.30}{\includegraphics{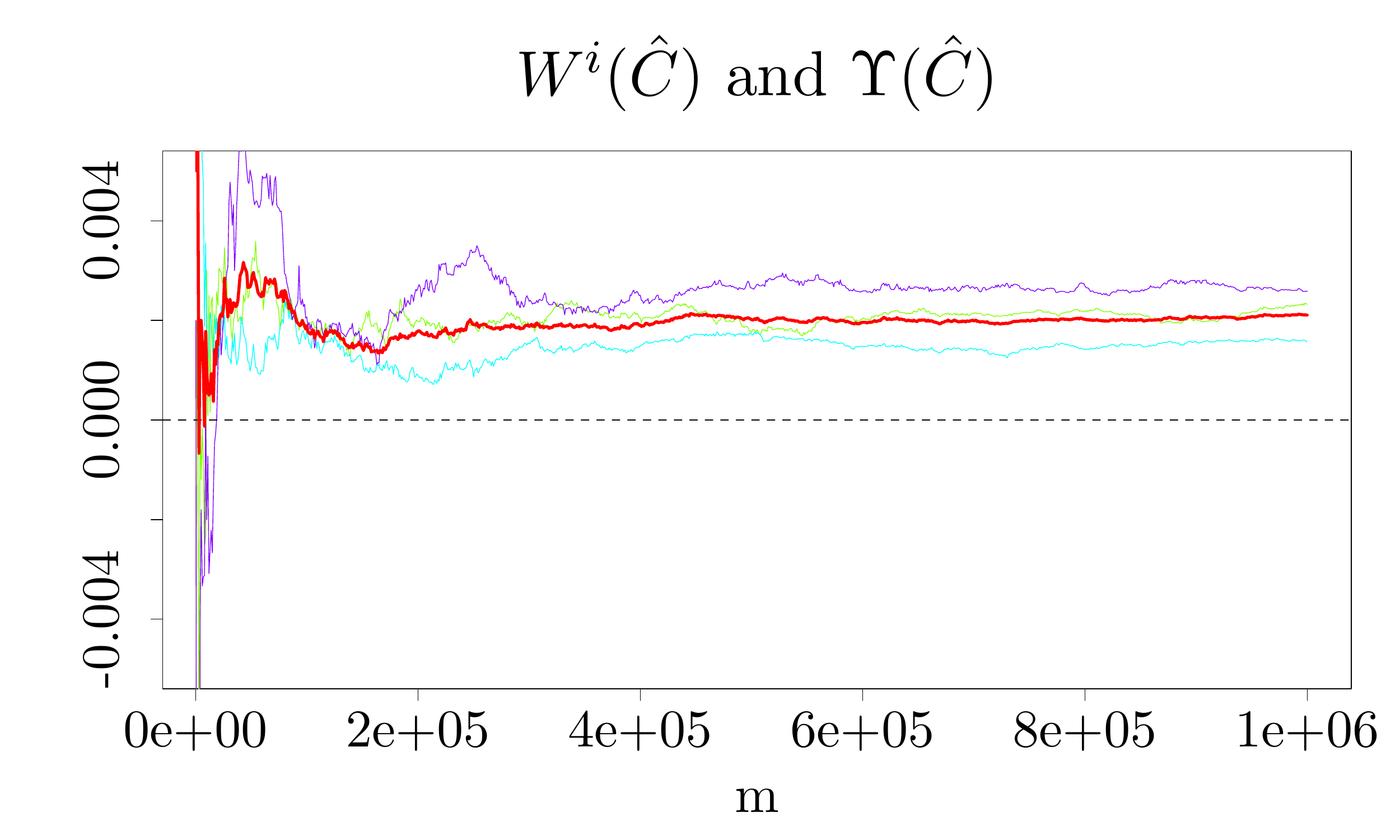}}
\scalebox{0.30}{\includegraphics{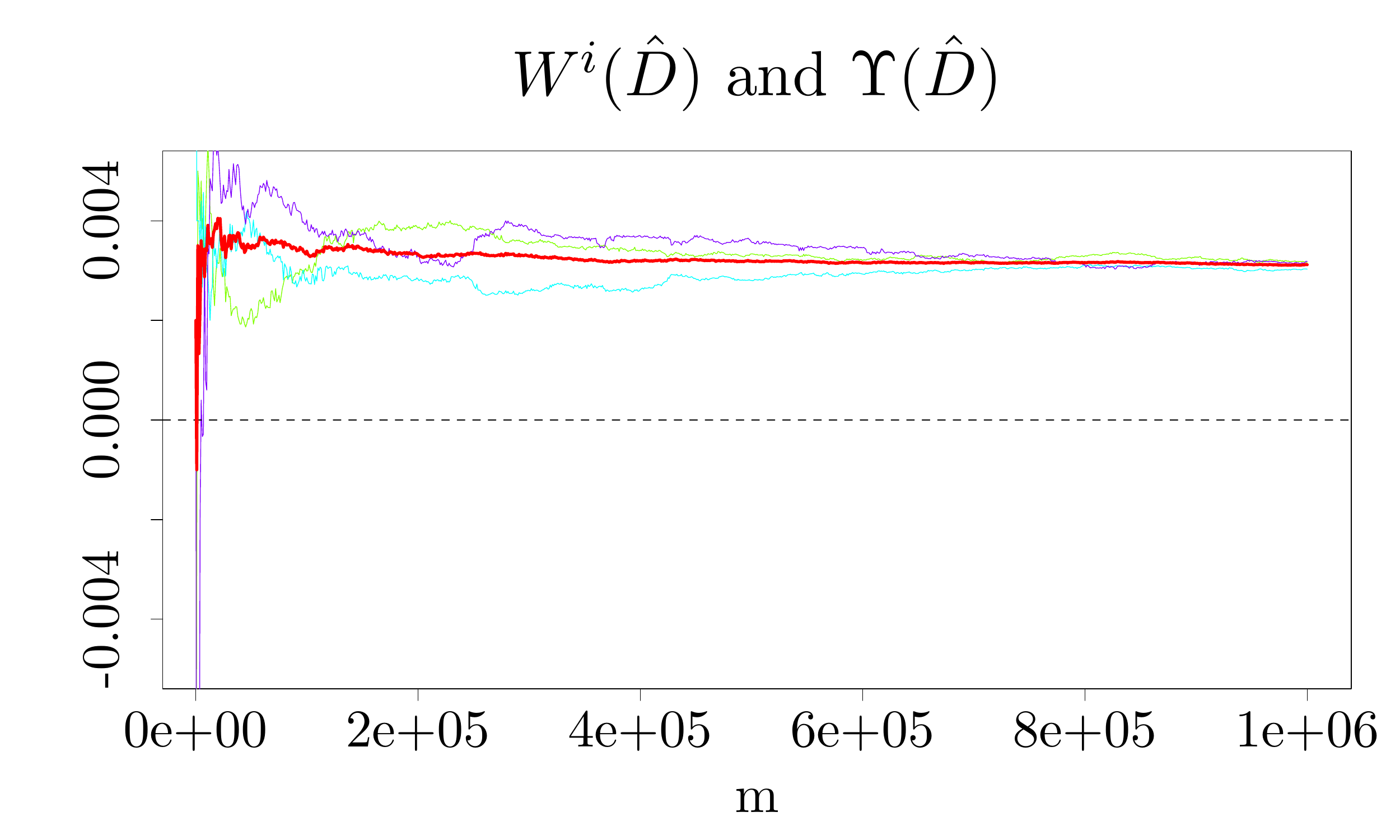}} \\
\end{center}

\caption{Example~\ref{ex:3} (small sample size): We fix the learning period of size $n=250$ and consider increasing lengths of backtesting intervals, i.e.~we let $m$ run. On the first two rows we plot  $G_{0.05}^i$, $i=1,2,3,$ for the true allocation $a$, and the estimated risk allocations $\hat B$, $\hat C$ and $\hat D$.
On the last two rows we plot $W^i,\ i=1,2,3,$. For allocations $a$ and $\hat B$  the measures are getting close to zero with increasing $m$, and $\Upsilon$ gets close to $0.05$,  confirming that these are fair allocations. For allocations $\hat C$ and $\hat D$ the opposite is true,  indicating that these are not fair allocations.
}
\label{F:asympt.G}
\end{figure}

\begin{figure}[htp!]
\begin{center}
\scalebox{0.30}{\includegraphics{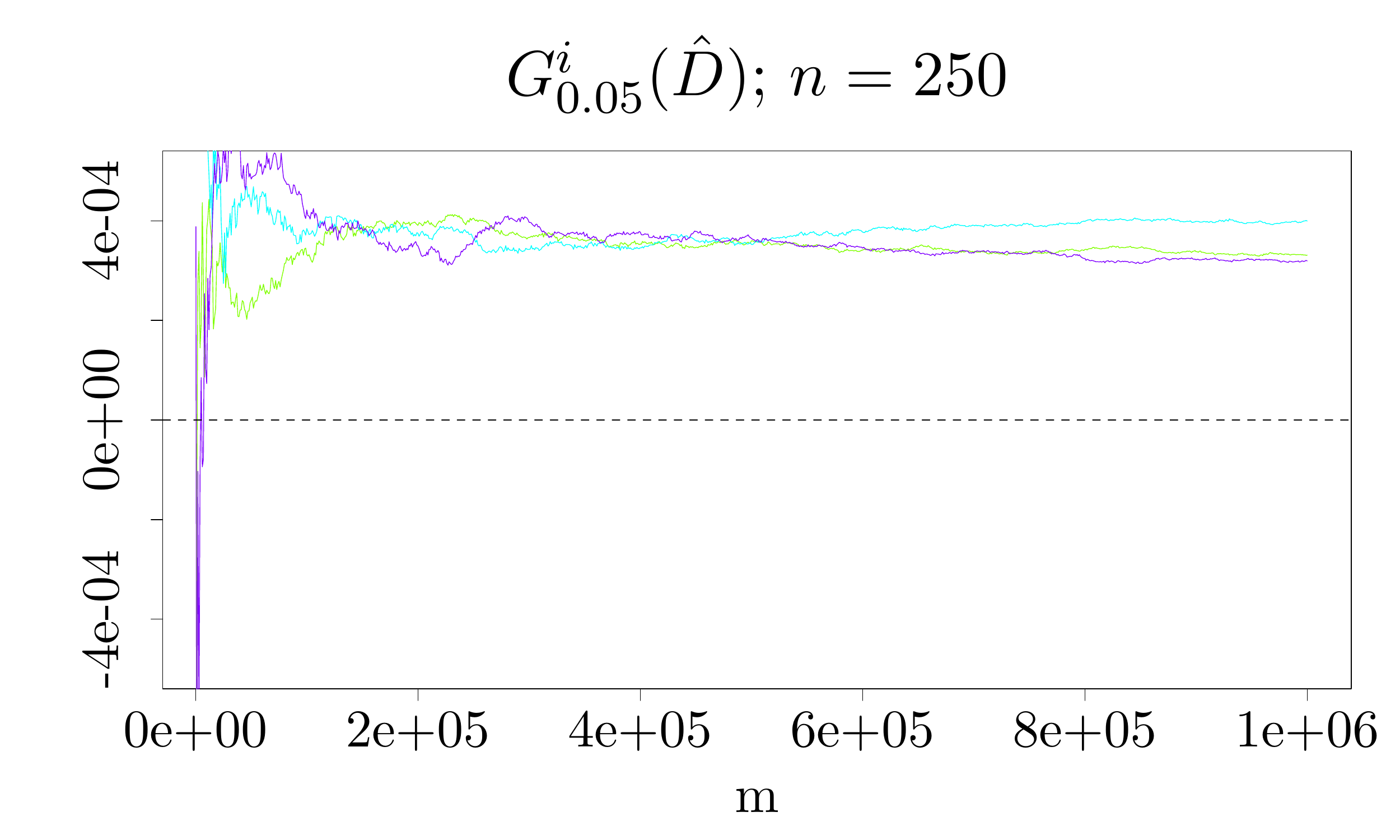}}
\scalebox{0.30}{\includegraphics{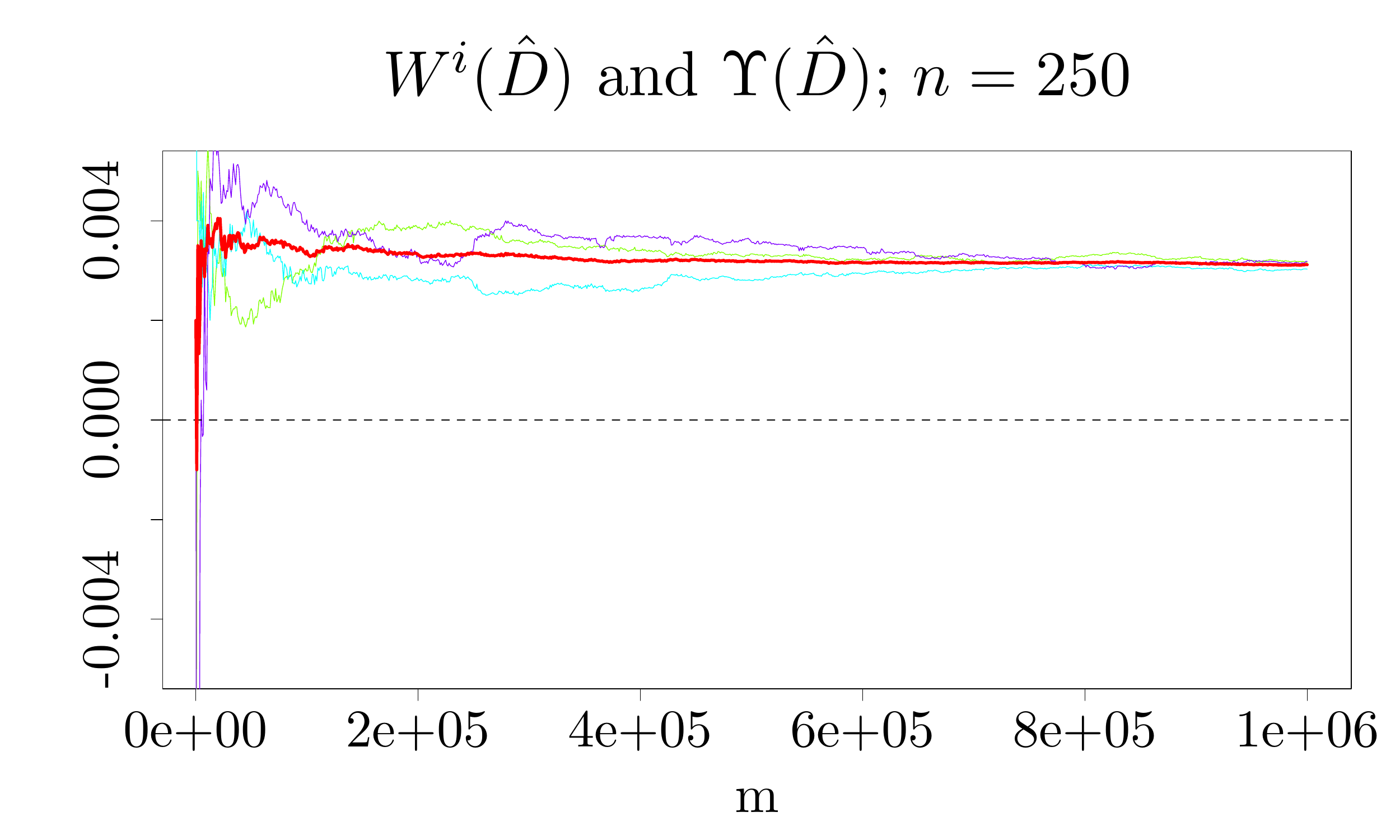}}\\
\scalebox{0.30}{\includegraphics{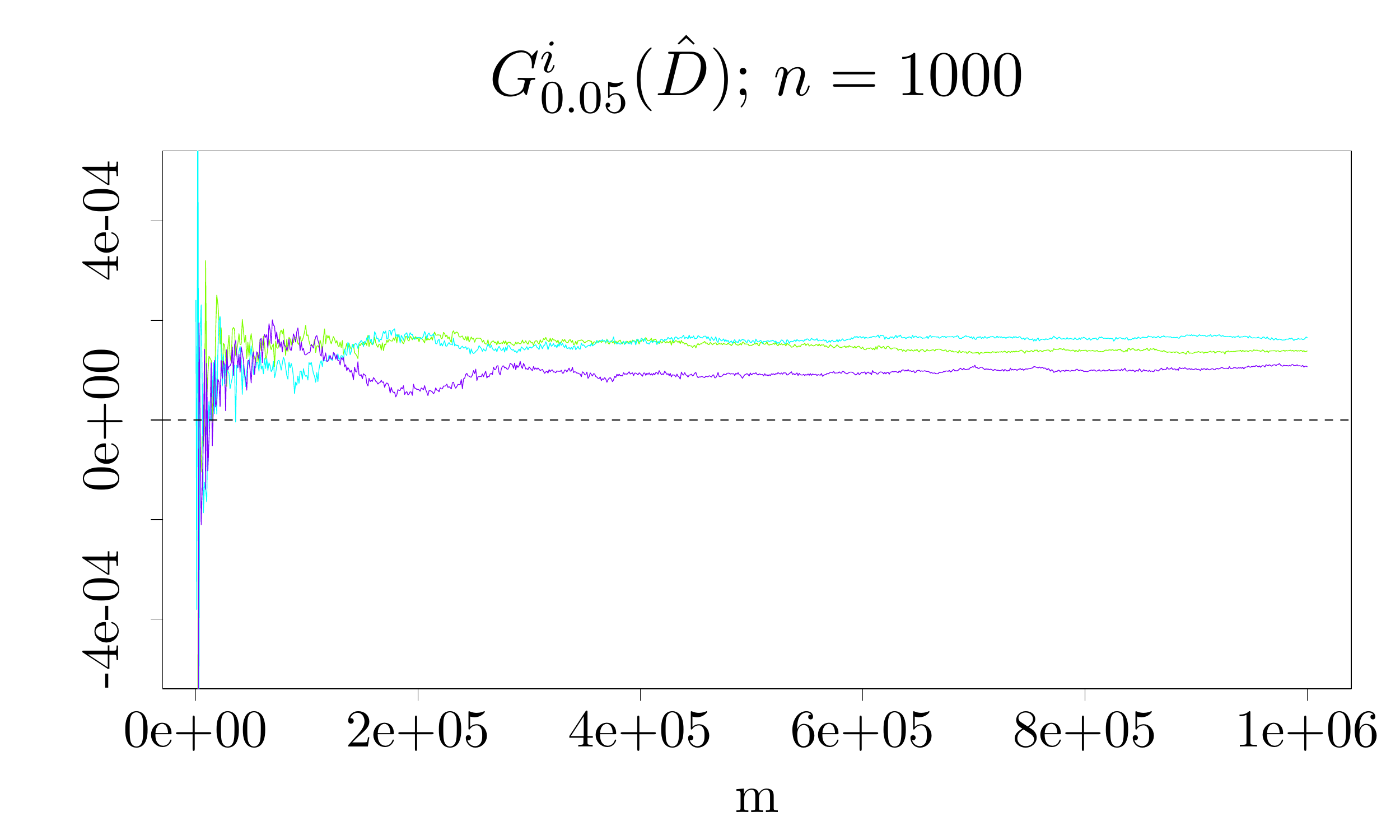}}
\scalebox{0.30}{\includegraphics{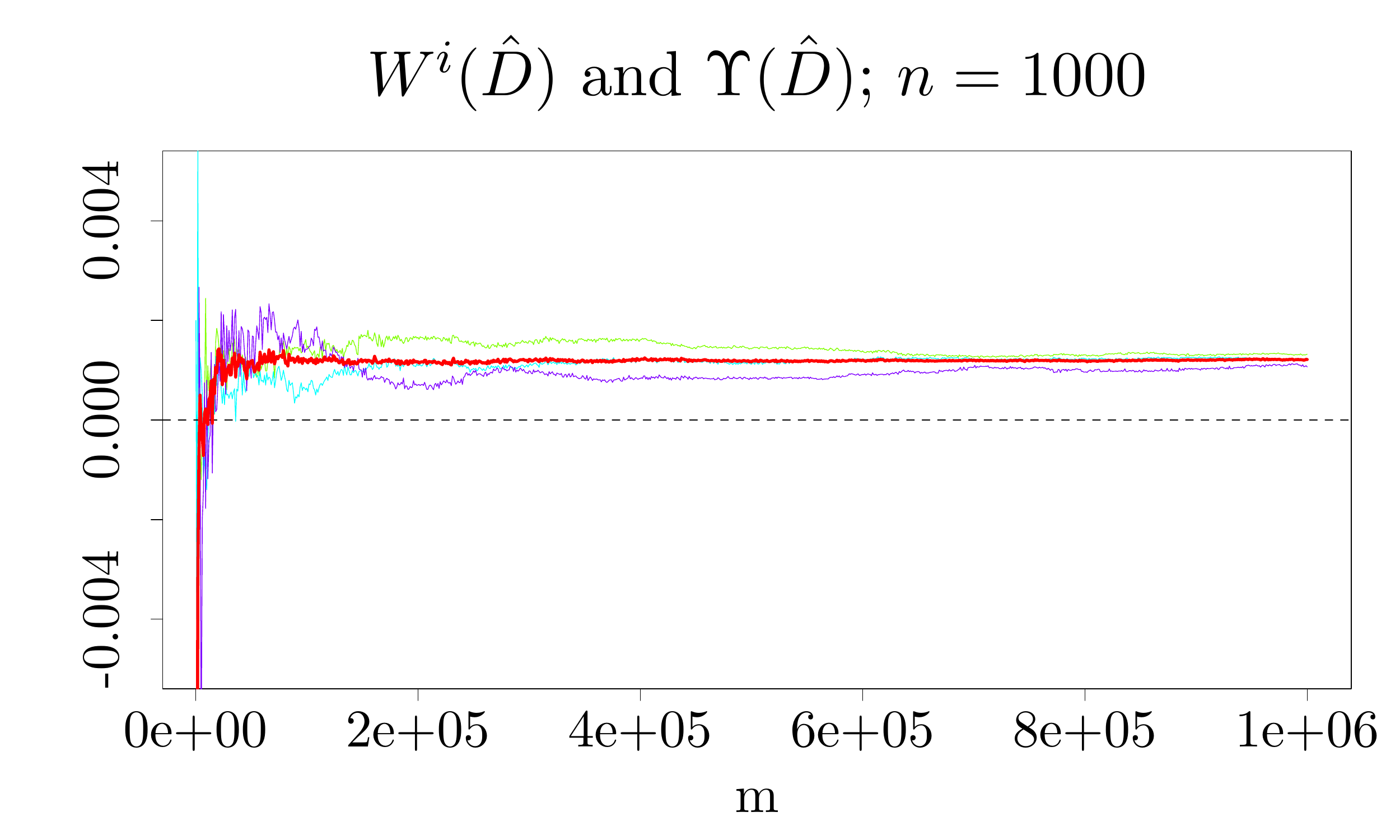}}\\
\scalebox{0.30}{\includegraphics{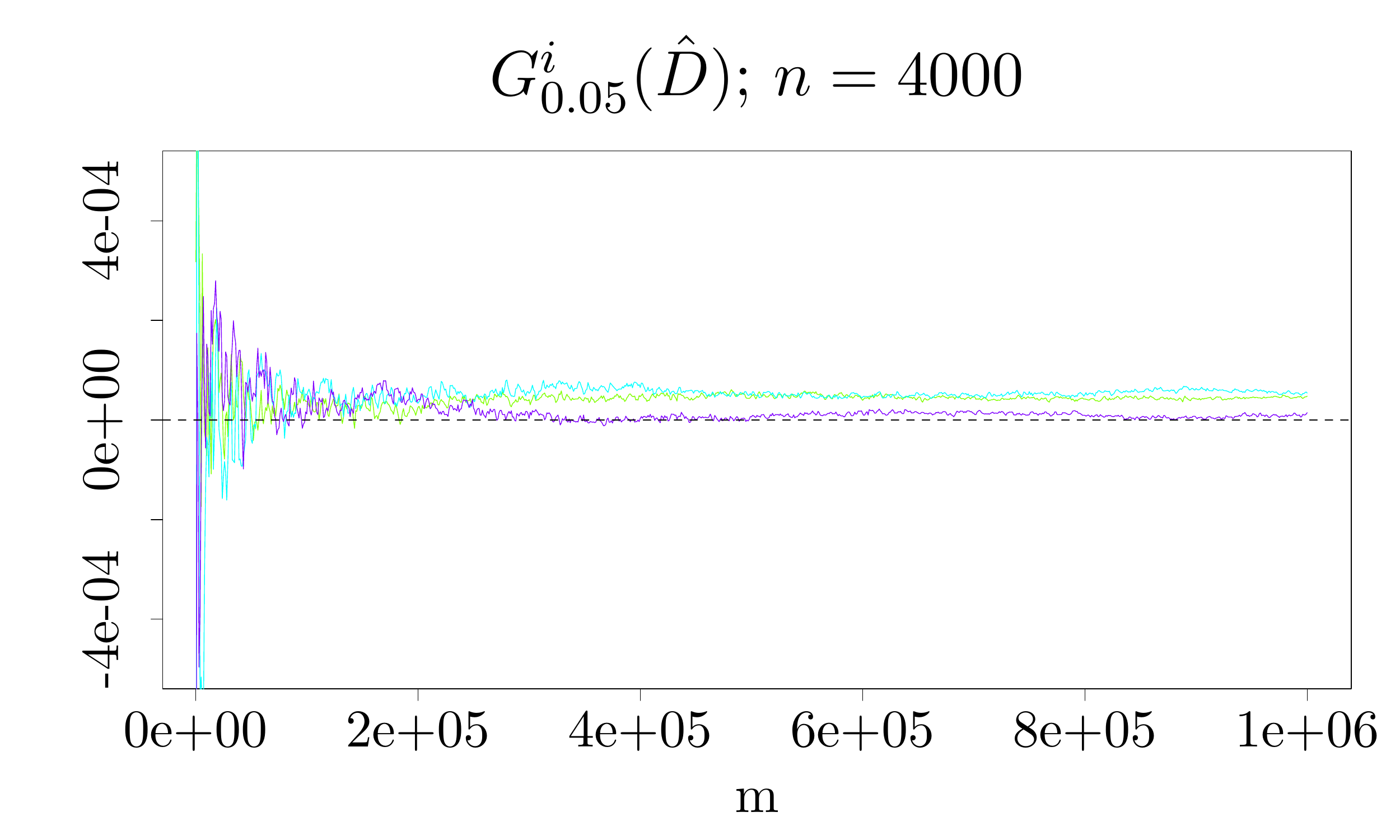}}
\scalebox{0.30}{\includegraphics{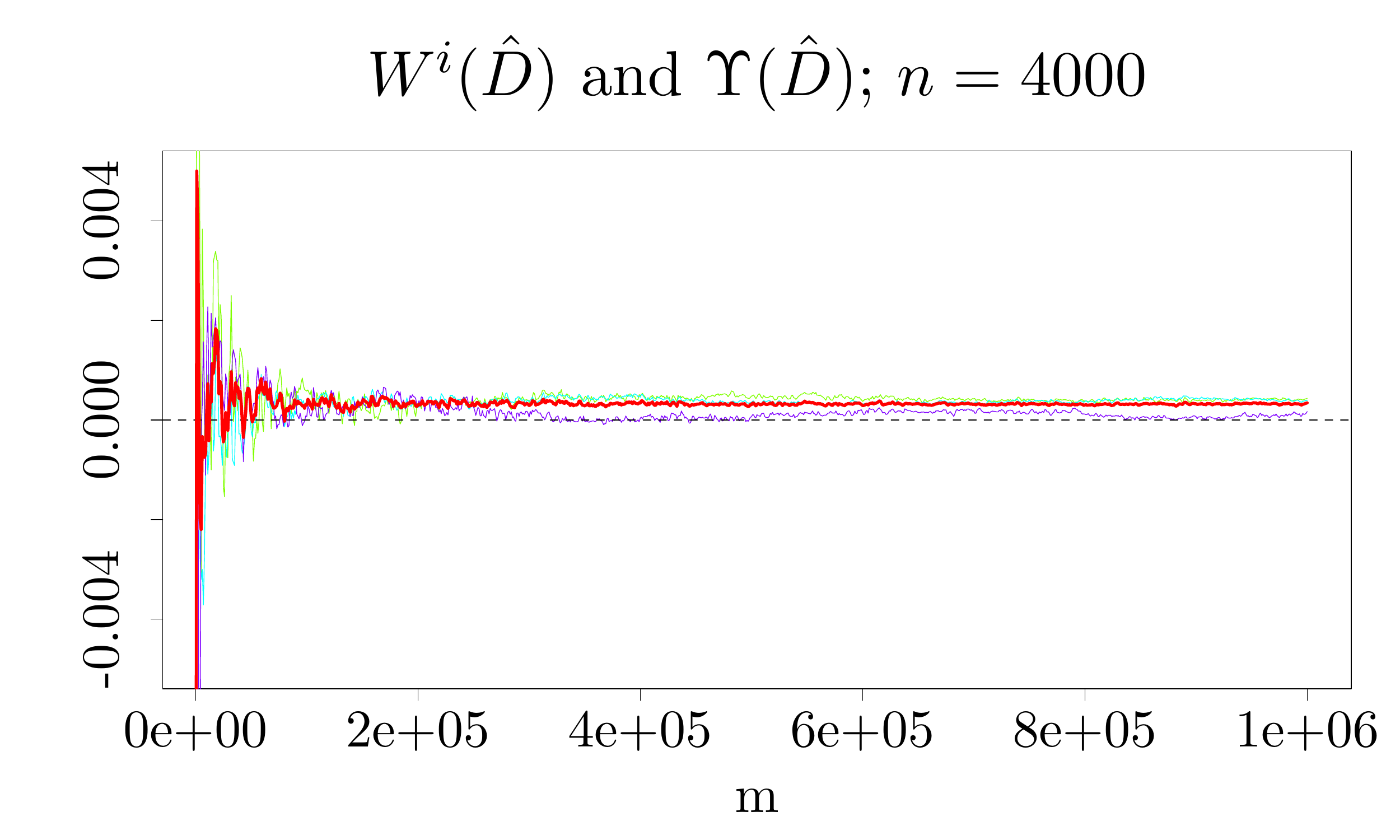}}
\end{center}
\caption{Example~\ref{ex:3} (asymptotic fairness): in the left panel we plot $G^i_{0.05}(\hat{D})$ as function of $m$, for different values of learning period, $n=250$ (top) to  $n=4000$ (bottom). These plots confirm that $G^i_{0.05}(\hat{D})$ get closer to zero with increasing $n$, which yields asymptotic fairness.  The pictures in the right panel contain values of $W^i(\hat{D})$ and $\Upsilon(\hat{D})$ as functions of $m$, and for $n=250, 1000$ and $4000$ (from top to bottom). We obtain similar results here, showing that  $W^i$  and $\Upsilon$ get closer to zero with increasing $n$.}
\label{F:consistency}
\end{figure}

\end{example}

\clearpage

\begin{figure}[htp!]
\begin{center}
\scalebox{0.20}{\includegraphics{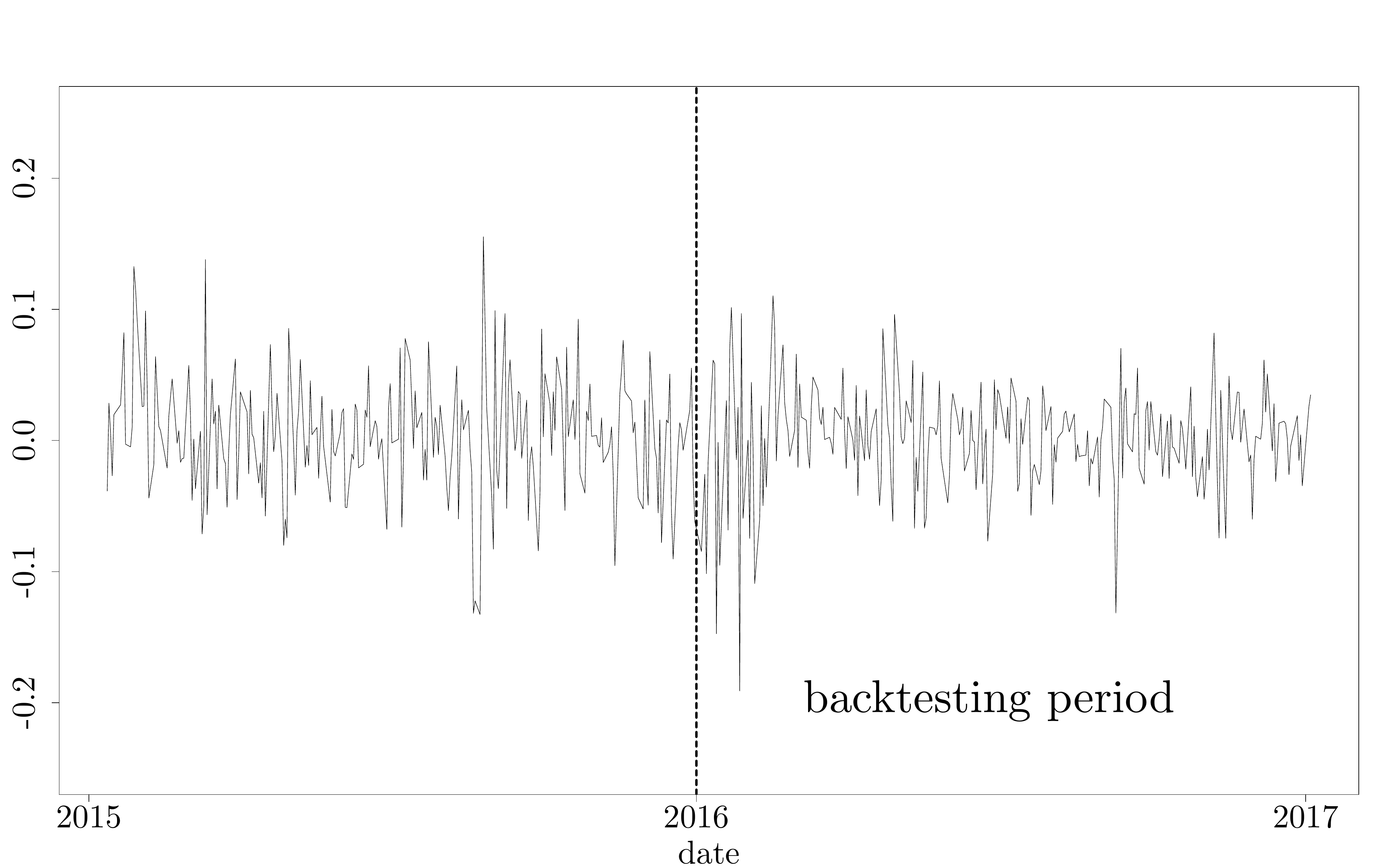}} \scalebox{0.20}{\includegraphics{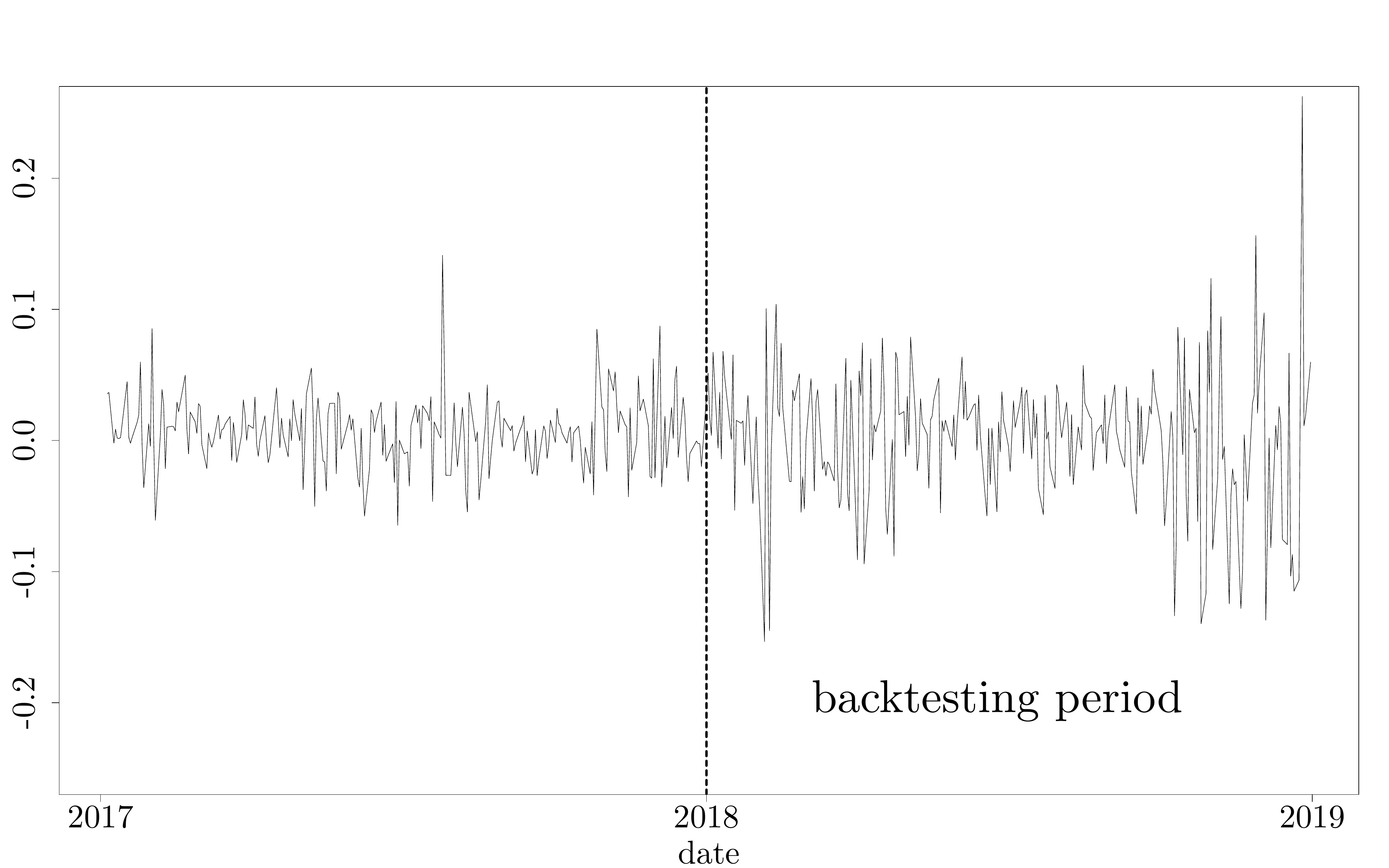}}
\end{center}
\caption{Aggregated portfolio P\&Ls split into \emph{Dataset~1} (January 2015 - December 2016, left panel) and \emph{Dataset~2} (January 2017 - December 2018, right panel).
The estimated volatilities are $0.0425\ (0.0456 / 0.0392) (first half / second half)$ in Dataset~1 and  in the right panel $0.0415 \ (0.0267 / 0.0522)$.
}
\label{fig:Ex4-0}
\end{figure}

\begin{example}[Market data example]\label{ex:4}
In this example we analyze the performance of the backtesting methodologies on  market data. We consider the same portfolio formation as in Example~\ref{ex:1}, by taking eight stocks (AAPL, AMZN, BA, DIS, HD, KO, JPM, and MSFT) from the S\&P 500 index, and form an equally weighted long-short portfolio. Namely, we hold a long position in the first six stocks, and a short position in the last two stocks, with nominal (absolute) value \$1 in each stock.
For this study, we use the daily stock returns, for the period January 2015 - December 2018. Throughout we set both, learning period ($n$) and backtesting period ($m$), equal to 250 days. We split the dataset into two subsets: January 2015 - December 2016 (\emph{Dataset~1}), and January 2017 - December 2018 (\emph{Dataset~2}).
As before, for each dataset, we use the standard 1-day rolling window and compare forecasted capital allocations with realized portfolio values.

One reason to split the data into these two time frames stems from the distinctively different patterns of the the aggregated P\&L of the portfolio; see Figure~\ref{fig:Ex4-0}.
Dataset~1 is more homogeneous, with slightly larger volatility in the first half. Specifically, the sample standard deviation of the aggregated P\&L portfolio for Dataset~1 is equal to 0.0425;  the sample standard deviation for the first half is 0.0456, and for the second half is 0.0392. Dataset~2 exhibits a higher volatility in the second half compared to its first half and compared to Dataset~1; the standard deviation for the first half is  0.0267, and for the second half is  0.0522. As  we will show later, these differences will propagate into the capital risk allocation and they will be picked up by the backtesting procedure. Similar to the previous examples, for both datasets we will use the risk allocation estimators $\widehat{C}$ and $\widehat{D}$, and we will use both backtesting procedures proposed in Section~\ref{sec:backtesting}. We also performed the Jarque-Bera normality test for the aggregated portfolio P\&Ls for both datasets, which was rejected at significance level 0.01.

In the following, we will analyze each dataset separately. A first overview is presented in Figure~\ref{fig:Ex4-1}, where the first two columns (left panel) correspond to Dataset~1, and the rightmost two columns (right panel) to Dataset~2.

\smallskip \noindent
\textbf{Dataset~1}, January 2015 - December 2016, Figure~\ref{fig:Ex4-1}, left panel, and Table~\ref{table:ex4-2}.
The aggregated (total) risk of the portfolio is displayed in the first row of Figure~\ref{fig:Ex4-1}, which was computed by using estimators $\widehat{C}$ and $\widehat{D}$.  The aggregated portfolio risk seems to be well estimated by both $\widehat C$ and $\widehat D$. The noticeable slight decrease in time of the aggregated risk is partially due to the lower volatility of the returns in the second part of the Dataset~1. The estimated risk capital allocations are presented in the second row, and the backtesting statistics $G,G^i$, and $W,W^i$ are graphically displayed in rows 3-5 of Figure~\ref{fig:Ex4-1} and the numerical values are presented in Table~\ref{table:ex4-2}.

\begin{table}[t!]
\centering
\scalebox{0.8}{
\begin{tabular}{lllllllllll} \toprule
Dataset 1    & $X_1$ & $X_2$ & $X_3$ & $X_4$ & $X_5$ & $X_6$ & $X_7$ & $X_8$ &  & $S$ \\ \midrule
 $G^i_{0.05}(\hat C)$ & -0.00107 &  0.00221 &  0.00110 & -0.00036 & -0.00052 &  0.00255 &  0.00653 & -0.00797 & $G_{0.05}$ & 0.00247 \\
 $W^i(\hat C)$ & -0.002 &  0.010 &  0.010 & -0.002 & -0.026 &  0.014 & \textbf{-0.050} &  0.018 & $\Upsilon$ & 0.056 \\
   \midrule
 $G^i_{0.05}(\hat D)$ & -0.00653 & -0.00855 & -0.00692 & -0.00525 & -0.00419 &  0.00102 &  0.01707 &  0.00339 & $G_{0.05}$ & -0.00996 \\
 $W^i (\hat D)$  & -0.022 & -0.022 & -0.026 & -0.030 & -0.026 &  0.006 & \textbf{-0.050} & -0.010 & $\Upsilon$ & 0.044 \\
   \bottomrule
\end{tabular}
}
\caption{Summary of backtesting statistics for Example~\ref{ex:4}, Dataset~1, split into the first period (first two rows) and the second period (last two rows): in the first columns we show $G^i_{0.05}$ and $W^i$, $i=1,\dots,8$, for the estimated risk allocations $\hat C$ and $\hat D$, corresponding to backtesting the fair allocation.
Overall, the capital allocations are well estimated by both $\widehat{C}$ and $\widehat{D}$, with exception of the seventh constituent, for which $W^7_{0.05}(\hat C)=-0.05$ and $W^7_{0.05}(\hat D)=-0.05$. For $i=7$ the correlation difference is noticeably higher than for the rest of the sample which might be a result of a structural change.
The last column shows the aggregated quantities $G_{0.05}$ and the risk level shift $\Upsilon$. Here, $\Upsilon$ is close to $\alpha=0.05$, as expected.}
\label{table:ex4-2}
\end{table}

Overall, the capital allocations are well estimated by both\footnote{We note that while data is not normally distributed, the estimator $\widehat{C}$ performed similarly well as the nonparametric estimator $\widehat{D}$.} $\widehat{C}$ and $\widehat{D}$, with exception of the seventh constituent, for which $W^7_{0.05}(\hat C)=-0.05$ and $W^7_{0.05}(\hat D)=-0.05$. To see whether this is a problem with the estimator or a result of time-correlation structure change we checked the sample correlations between each constituent and the portfolio for two disjoint subsets. The results are presented in Table~\ref{table:ex4-3}. One could see that for $i=7$ the correlation difference is noticeably higher than for the rest which might be a result of a structural change.
Consequently, we believe that the proposed  backtesting procedures correctly identified a wrong allocation in this particular case.

\begin{table}[tp!]
\centering
\scalebox{0.8}{
\begin{tabular}{rrrrrrrrr} \toprule
& $X_1$ & $X_2$ & $X_3$ & $X_4$ & $X_5$ & $X_6$ & $X_7$ & $X_8$ \\ \midrule
01/2015 -- 12/2015 &  0.70 & 0.66 & 0.75 & 0.66 & 0.68 & 0.58 & -0.54 & -0.38\\
01/2016 -- 12/2016&  0.59 &0.65 & 0.56 & 0.6 & 0.59 & 0.52 & -0.18 & -0.33\\
\midrule
Difference &  0.11 & 0.01 & 0.18 & 0.06 & 0.09 & 0.06 & {\bf -0.36} & -0.04\\ \bottomrule
\end{tabular}
}
\caption{Estimated correlations between each portfolio constituent ($X_i$) and aggregated portfolio ($S$), for two separate time periods for Dataset~1 in Example~\ref{ex:4}.
One can that the biggest difference is observed for $i=7$; this might indicate a structural change, a possible explanation of the results in Table \ref{table:ex4-2}.
}
\label{table:ex4-3}
\end{table}

\smallskip \noindent
\textbf{Dataset~2}, January 2017 - December 2018, Figure~\ref{fig:Ex4-1}, right panel, and Table~\ref{table:ex4-4}. Due to the increase of the volatility in the second half of the Dataset~2, the aggregated portfolio risk increases throughout the backtesting period; see Figure~\ref{fig:Ex4-1}, first row. In the second row of the same figure we present the nominal value of the allocated risk among constituents computed by using risk allocation estimators $\widehat{C}$ and $\widehat{D}$. In contrast to Dataset~1, the backtesting results for Dataset~2 reveal a significant underestimation of the aggregated  risk. This can be seen by noticing that the values of $G_{0.05}$ and $\Upsilon$, for both $\widehat{C}$ and $\widehat{D}$, are far from zero; see  last column in Table~\ref{table:ex4-4}, or the third and fourth rows of Figure~\ref{fig:Ex4-1}, right panel.
The graph of function $\beta\to G_\beta$ is plotted in the third row of Figure~\ref{fig:Ex4-1}, solid red line, and the value of $\Upsilon$ corresponds to the red vertical line in the last row. Comparing these plots with the corresponding plots from previous examples and datasets, we also conclude that the aggregated risk is significantly underestimated. Inevitably, this error propagates to the risk allocation estimation, as shown in the plots from rows 3-5. Clearly, the values of $G^i_{0.05}$ and $\Upsilon$ are significantly different from zero (see also Table~\ref{table:ex4-4}), in comparison to those from Dataset~1 and the previous examples. On the other hand, arguably, the risk allocation using the nonparametric estimators $\widehat{D}$ performs better than that one using $\widehat{C}$; see for instance the values of $W^i(\widehat{D})$ and $G^i_{0.05}(\widehat{D})$ versus $W^i(\widehat{C})$ and $G^i_{0.05}(\widehat{C})$. Finally, we note that, for the estimator $\widehat{D}$, the zeros of functions $\beta\to G^i_\beta(\widehat{D}),\ i=1,\ldots,8,$ are essentially the same as the zero of the function $\beta\to G_\beta(\widehat{D})$, indicating that the risk allocation itself (as proportion of the total risk) is done properly, and failure of the backtesting procedure is due to underestimation of the total risk.

\begin{table}[t]
\centering
\scalebox{0.8}{
\begin{tabular}{lllllllllll} \toprule
Dataset 2    & $X_1$ & $X_2$ & $X_3$ & $X_4$ & $X_5$ & $X_6$ & $X_7$ & $X_8$ &  &  \\
  $G^i_{0.05}(\hat C)$ &  0.01187 &  0.02618 &  0.01512 &  0.00562 &  0.01229 &  0.00404 & -0.01356 & -0.01785 & $G_{0.05}$ & 0.0437 \\
  $W^i(\hat C)$ & 0.234 & 0.210 & 0.146 & 0.166 & 0.222 & 0.110 & 0.278 & 0.274 & $\Upsilon$ & 0.204 \\
   \midrule
  $G^i_{0.05}(\hat C)$ &  0.00680 &  0.01669 &  0.01633 &  0.00301 &  0.00876 &  0.00181 & -0.00809 & -0.01106 & $G_{0.05}$ & 0.03425 \\
  $W^i(\hat D)$ & 0.158 & 0.130 & 0.102 & 0.078 & 0.074 & 0.042 & 0.118 & 0.086 & $\Upsilon$ & 0.156 \\
   \bottomrule
\end{tabular}
}
\caption{Summary of backtesting statistics for Example~\ref{ex:4}, Dataset~2, split into the first period (first two rows) and the second period (last two rows):  in the first columns we show $G^i_{0.05}$ and $W^i$, $i=1,\dots,8$, for the estimated risk allocations $\hat C$ and $\hat D$, corresponding to backtesting the fair allocation.
In contrast to Dataset~1, the backtesting results for Dataset~2 reveal a significant underestimation of the aggregated  risk.
Inevitably, this error propagates to the risk allocation estimation. Clearly, the values of $G^i_{0.05}$ and $\Upsilon$ are significantly different from zero. The risk allocation using the nonparametric estimators $\widehat{D}$ performs better than that one using $\widehat{C}$. }
\label{table:ex4-4}
\end{table}

\setlength{\columnsep}{1.5cm}
\setlength{\columnseprule}{0.2pt}
\begin{figure*}
\begin{multicols}{2}
{Dataset 1}\par\medskip
\scalebox{0.18}{\includegraphics{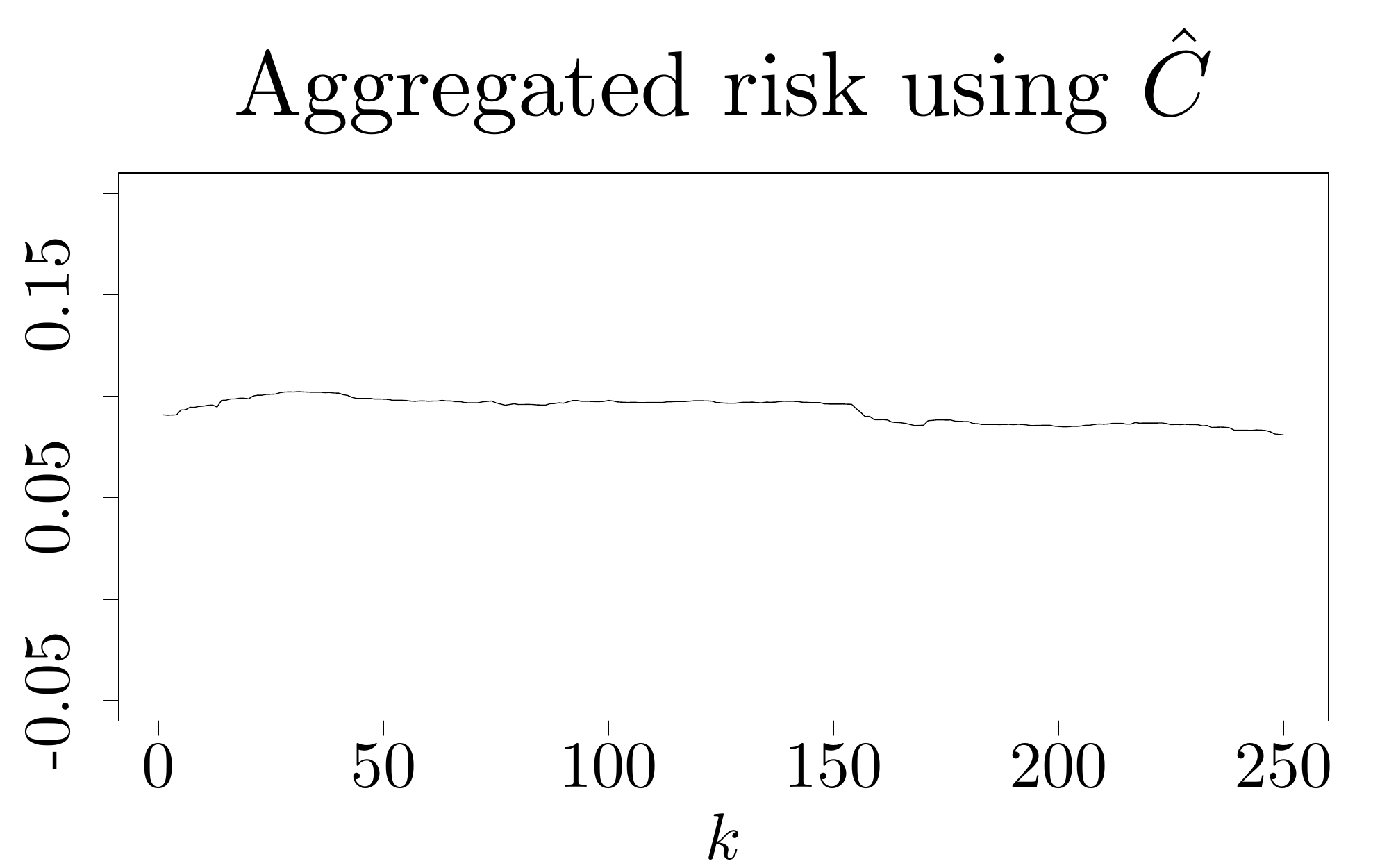}}
\scalebox{0.18}{\includegraphics{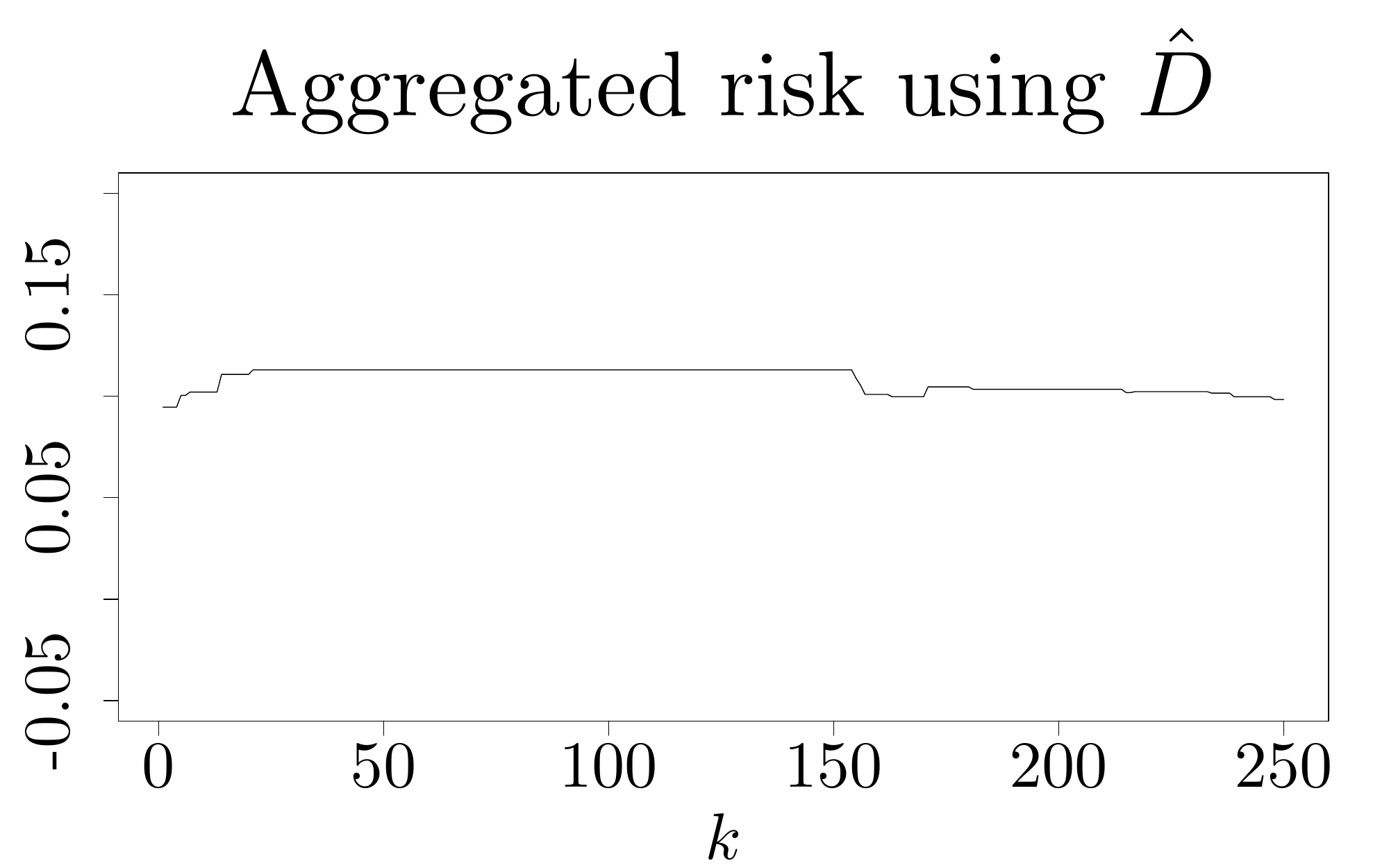}}\\
\scalebox{0.18}{\includegraphics{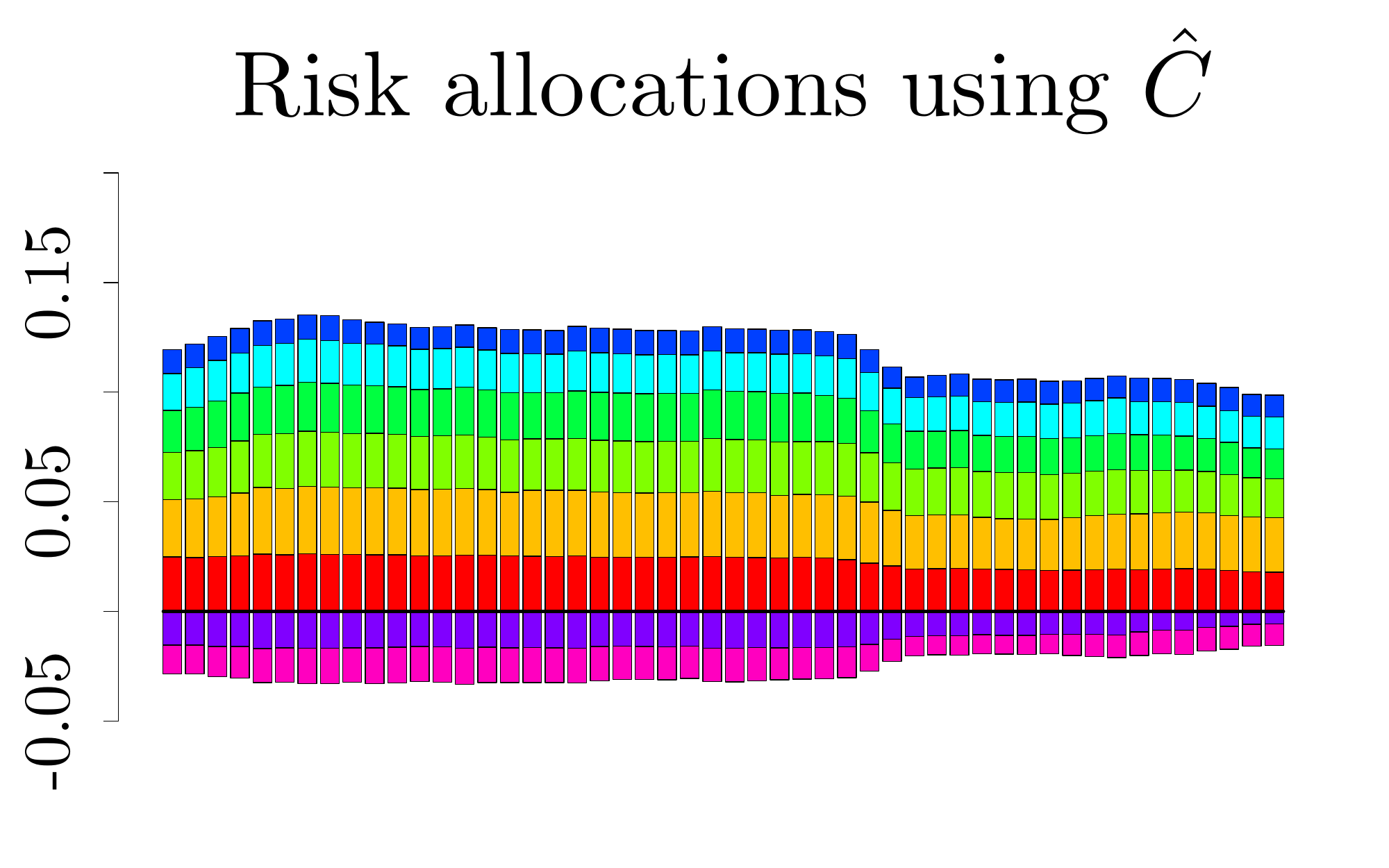}}
\scalebox{0.18}{\includegraphics{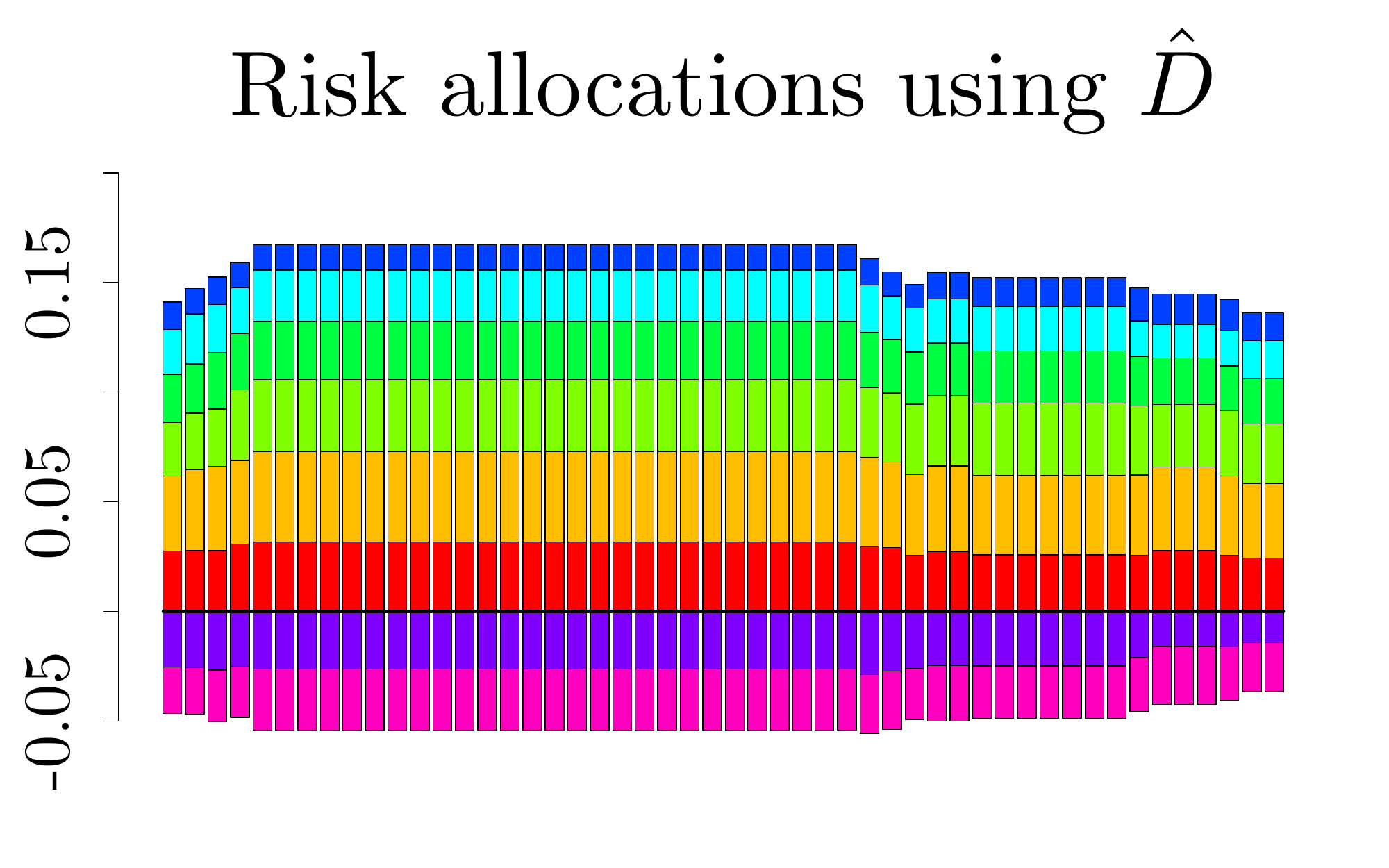}}\\
\scalebox{0.18}{\includegraphics{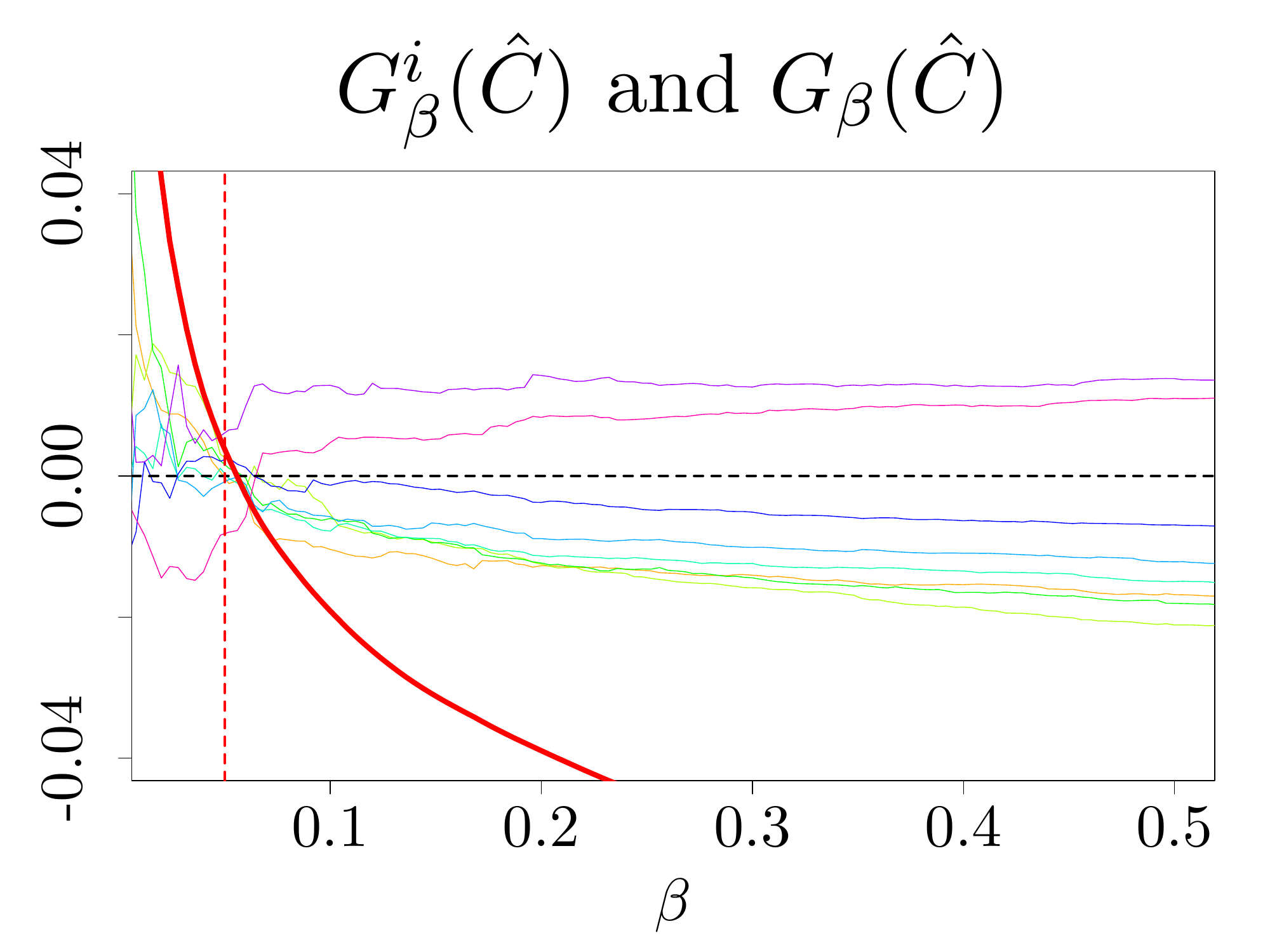}}
\scalebox{0.18}{\includegraphics{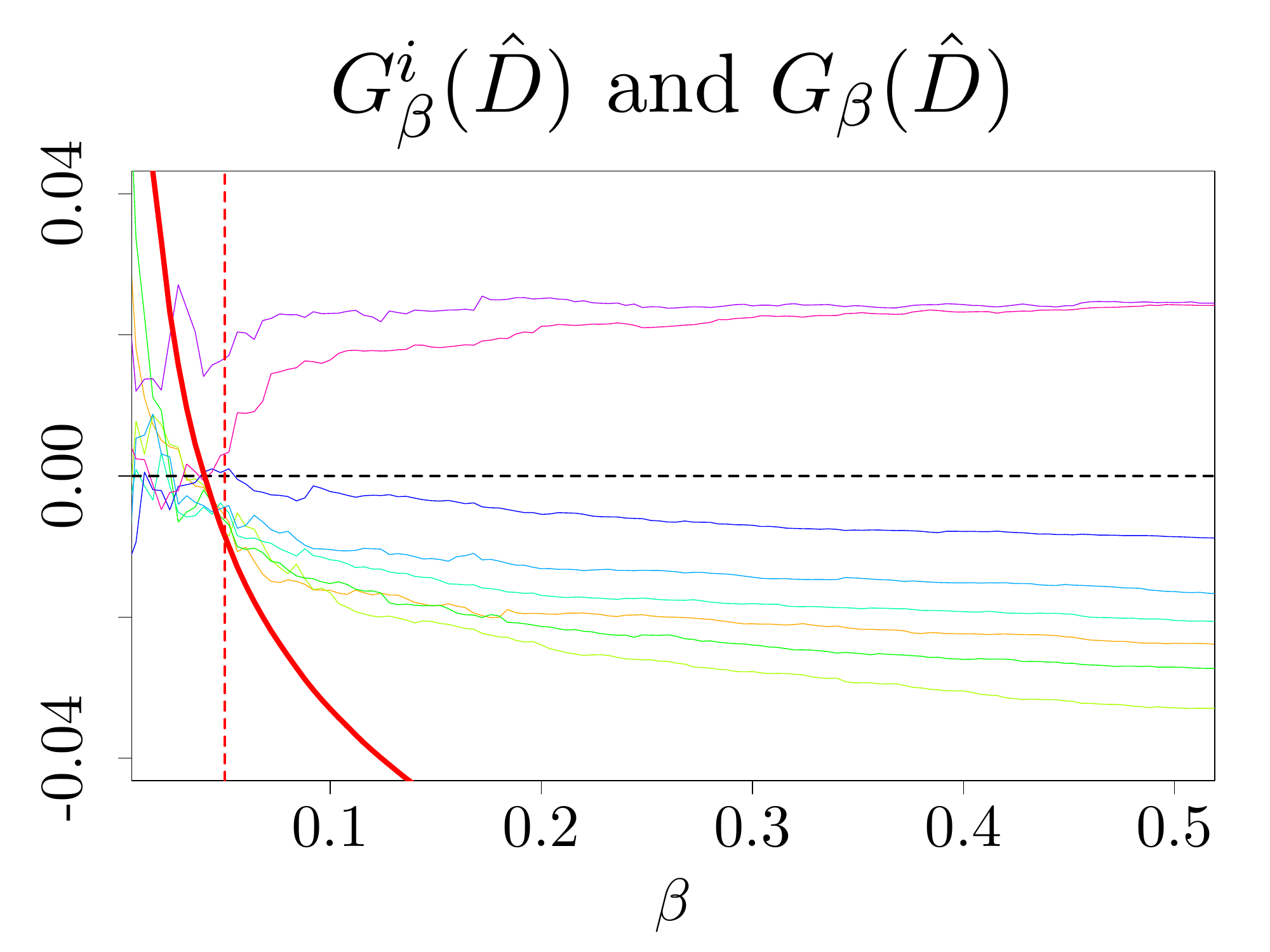}}\\
\scalebox{0.18}{\includegraphics{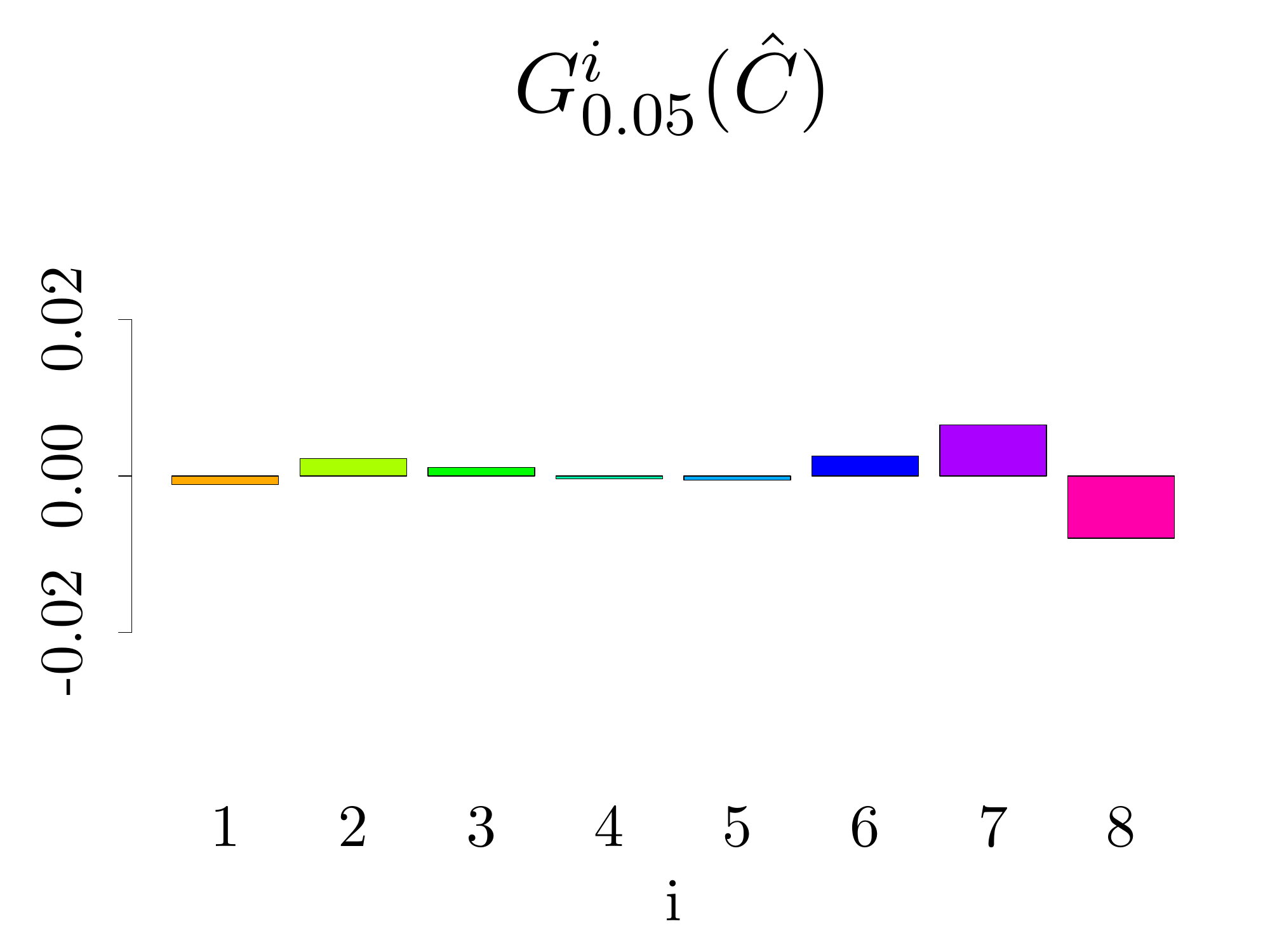}}
\scalebox{0.18}{\includegraphics{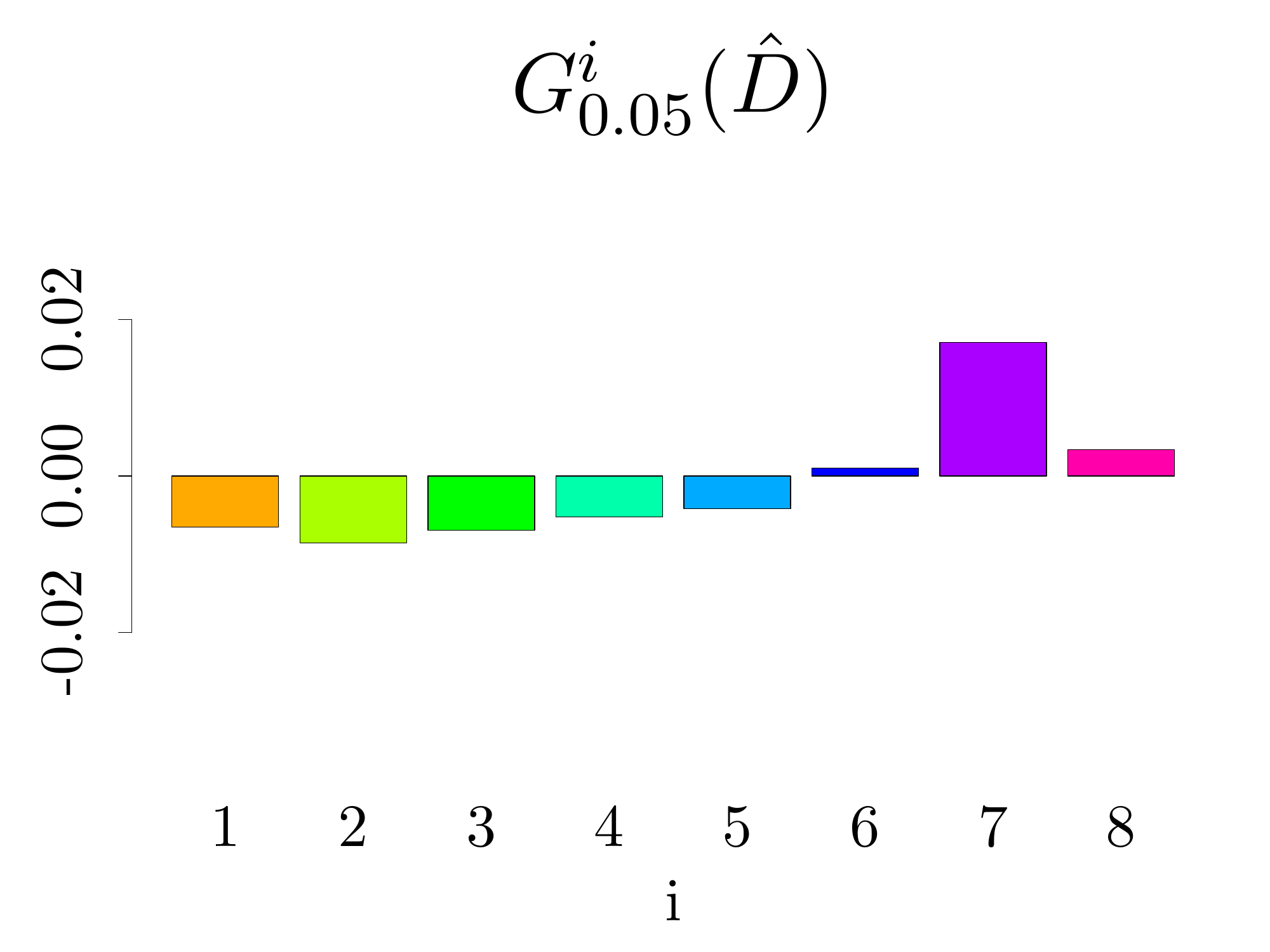}}\\
\scalebox{0.18}{\includegraphics{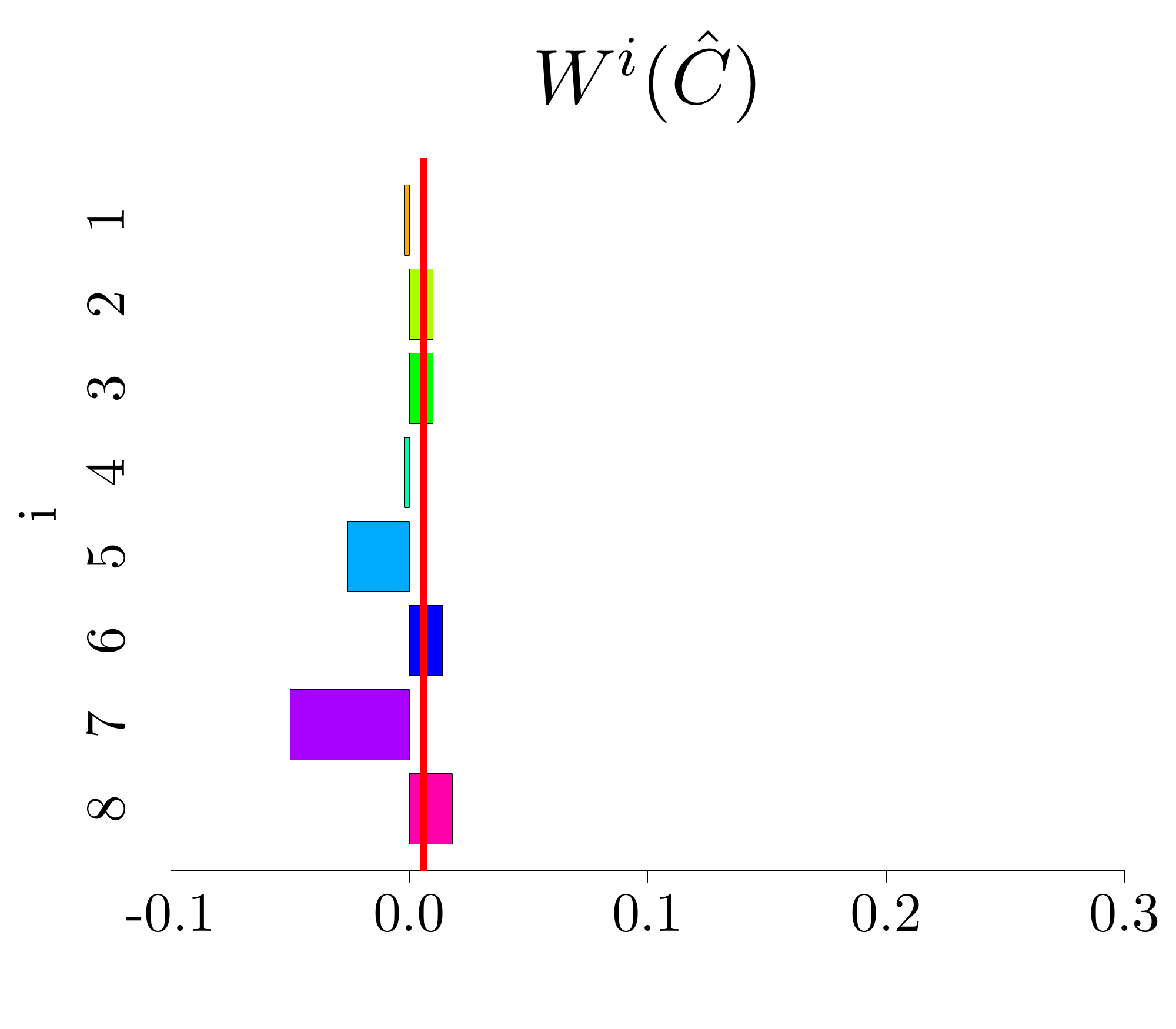}}
\scalebox{0.18}{\includegraphics{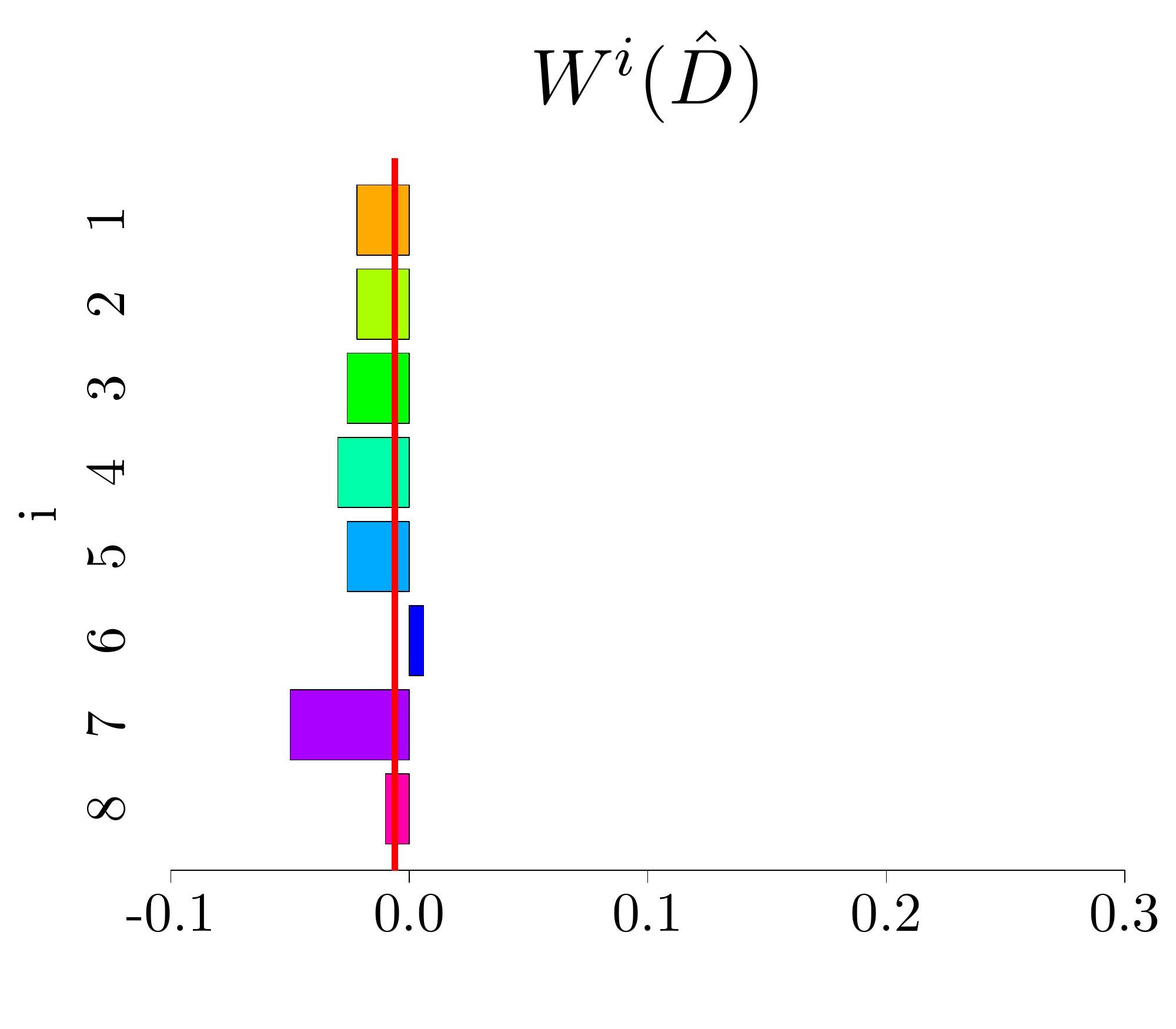}}\\
{Dataset 2}\par\medskip
\scalebox{0.18}{\includegraphics{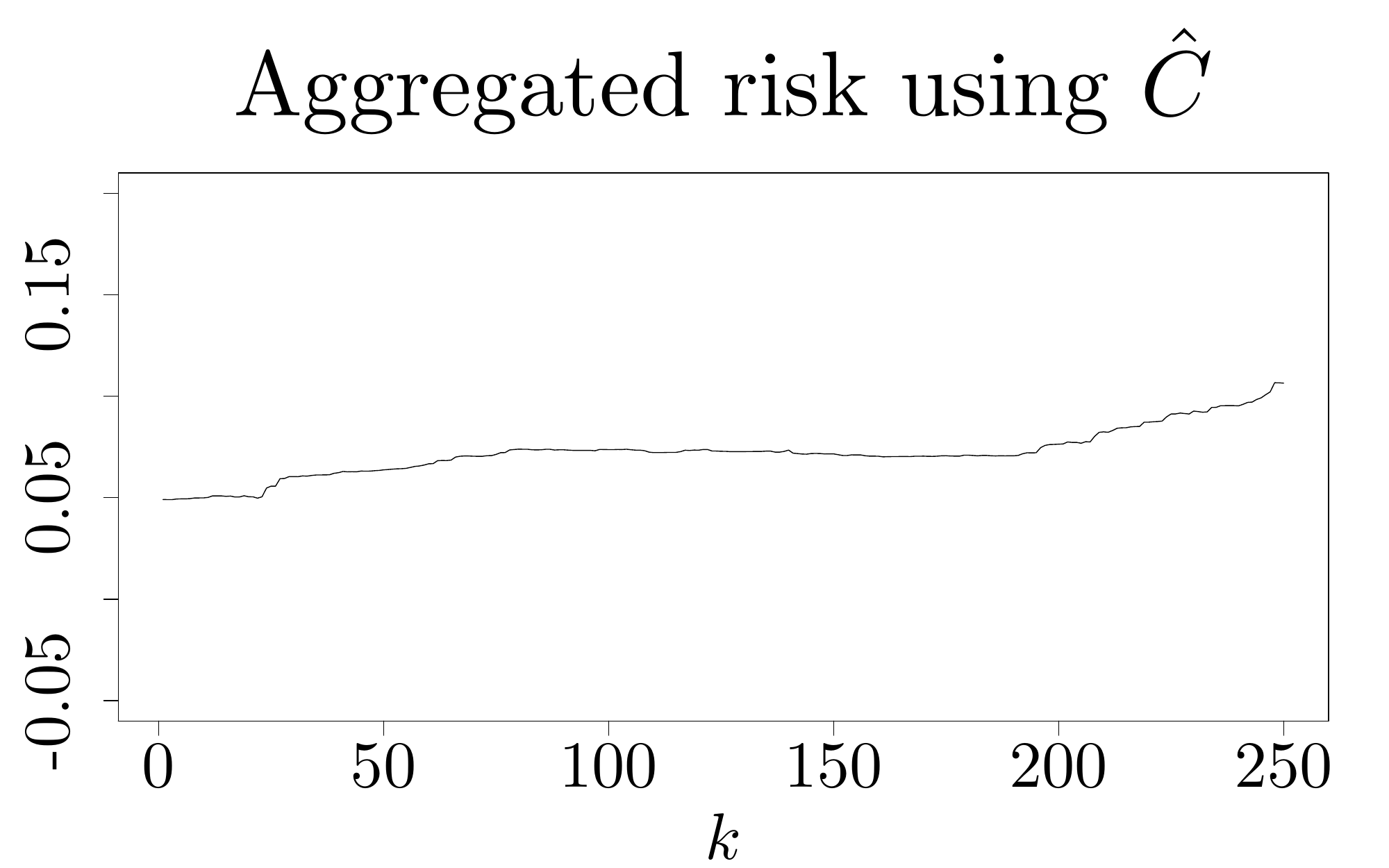}}
\scalebox{0.18}{\includegraphics{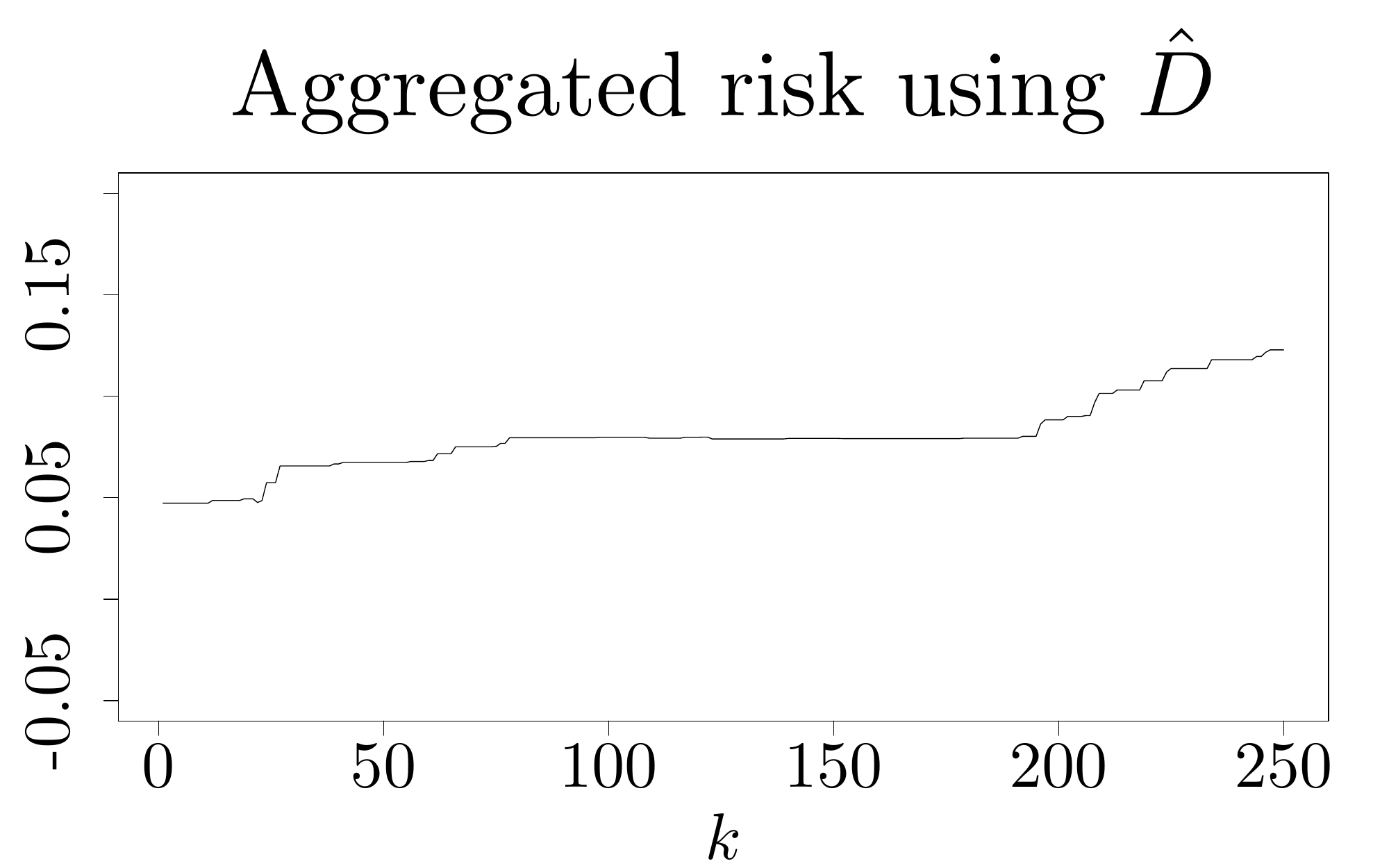}}\\
\scalebox{0.18}{\includegraphics{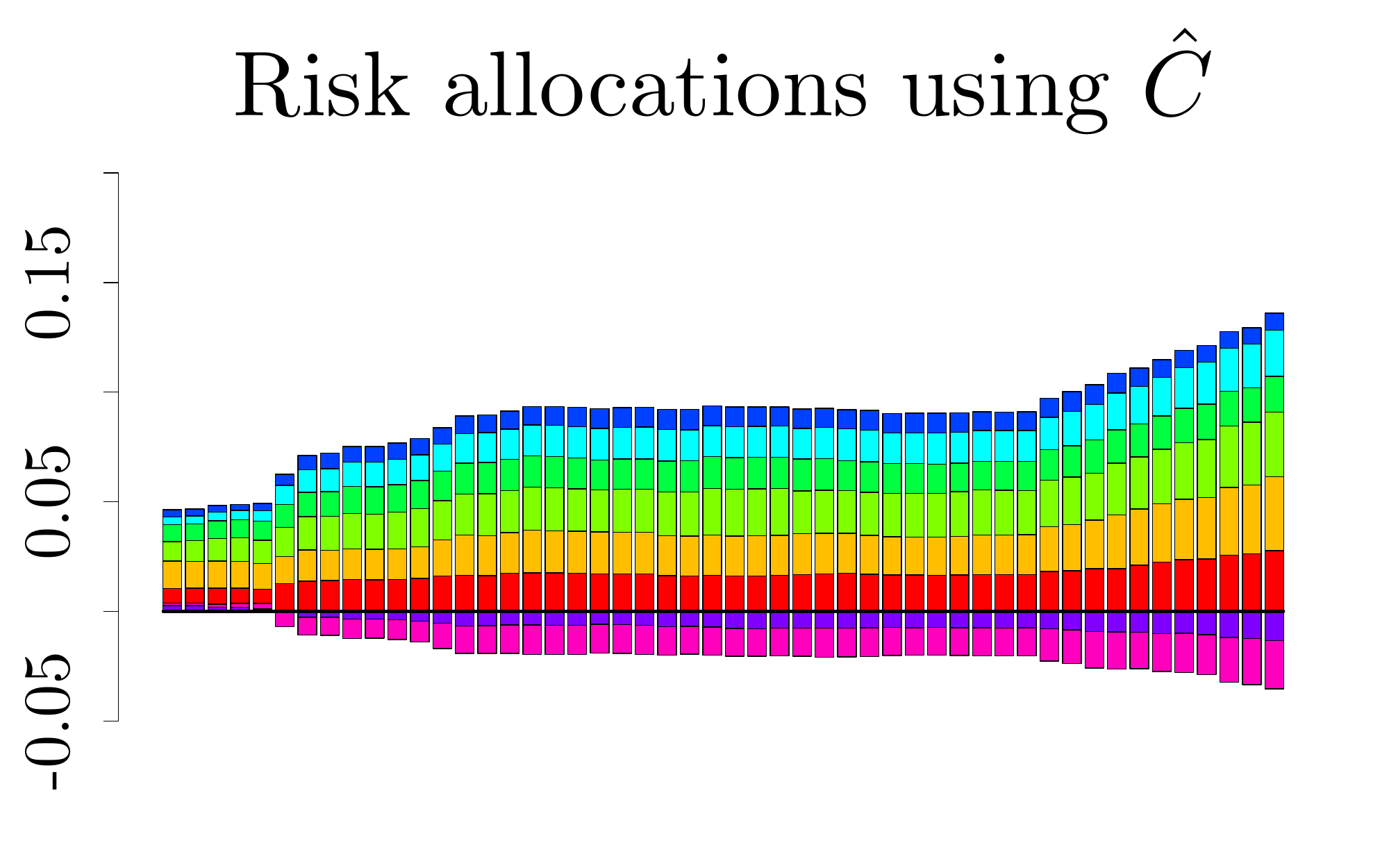}}
\scalebox{0.18}{\includegraphics{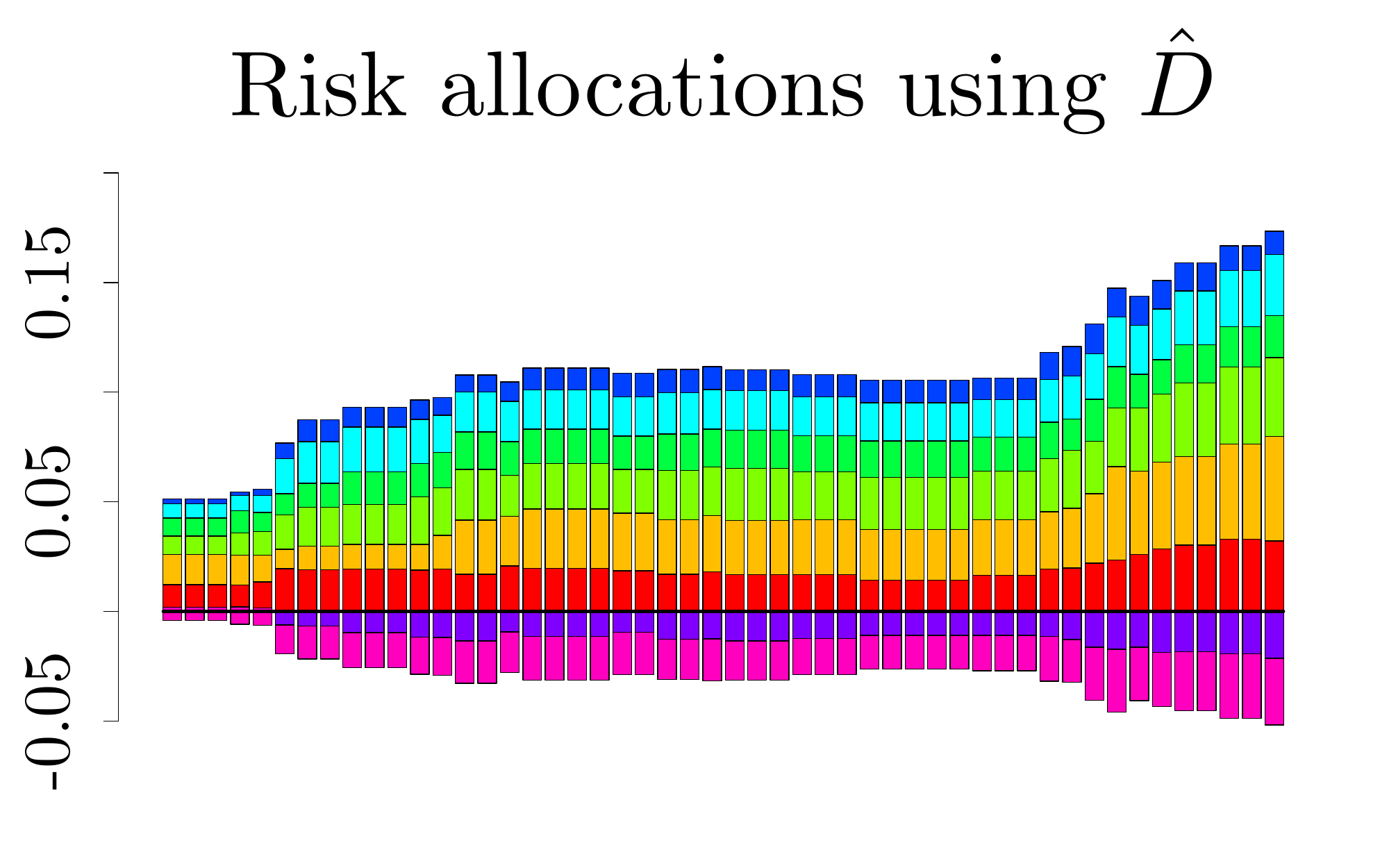}}\\
\scalebox{0.18}{\includegraphics{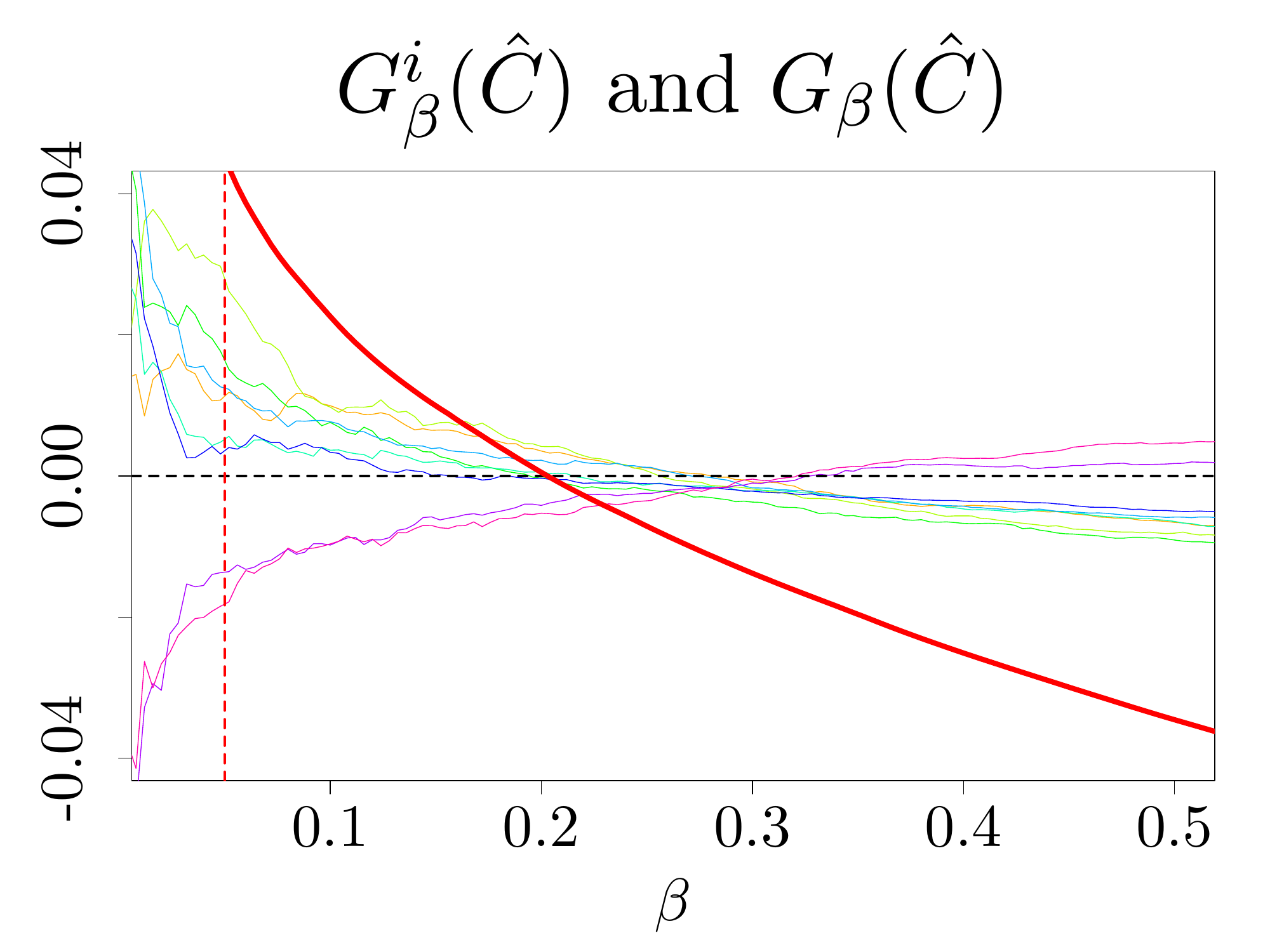}}
\scalebox{0.18}{\includegraphics{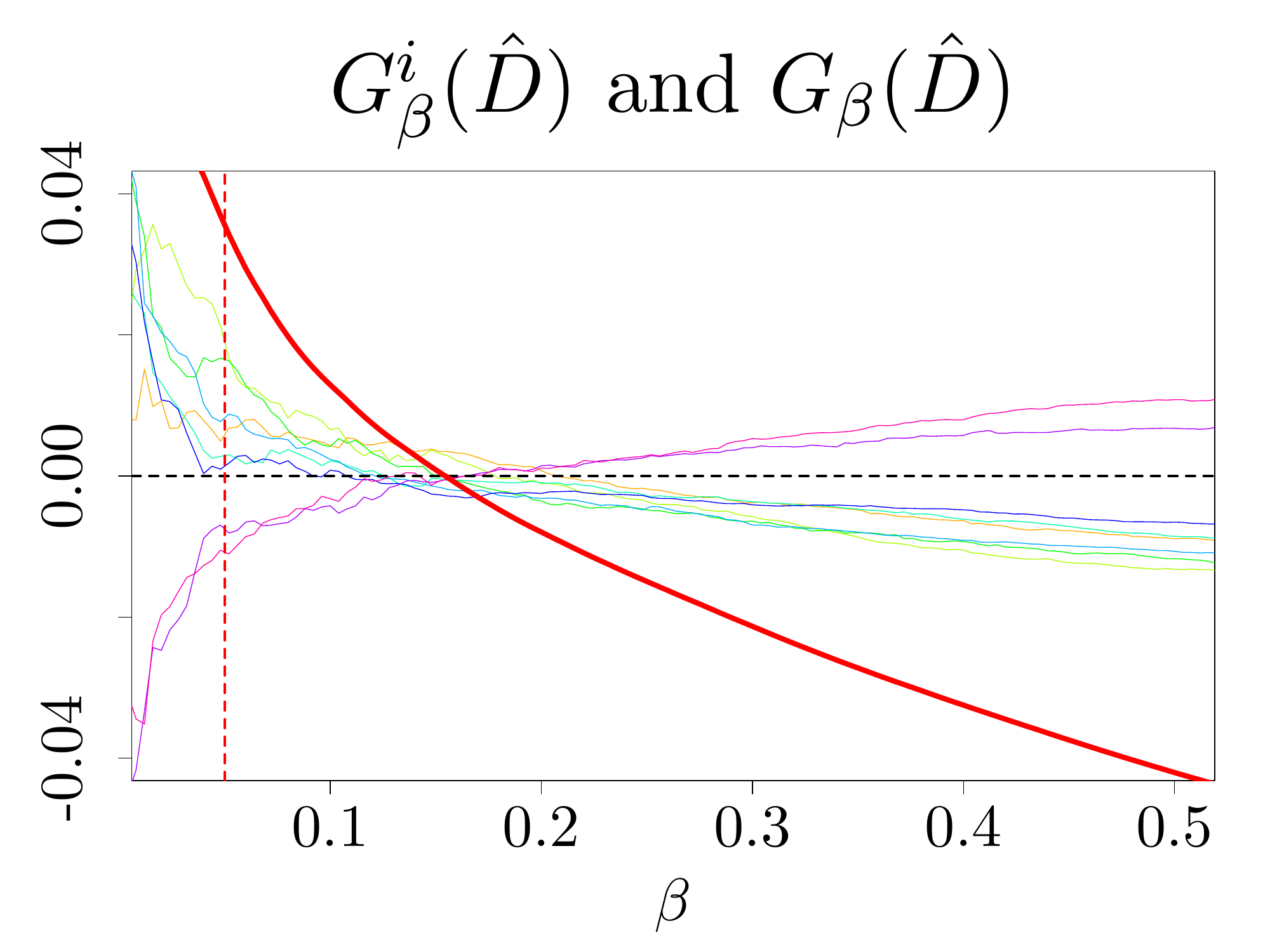}}\\
\scalebox{0.18}{\includegraphics{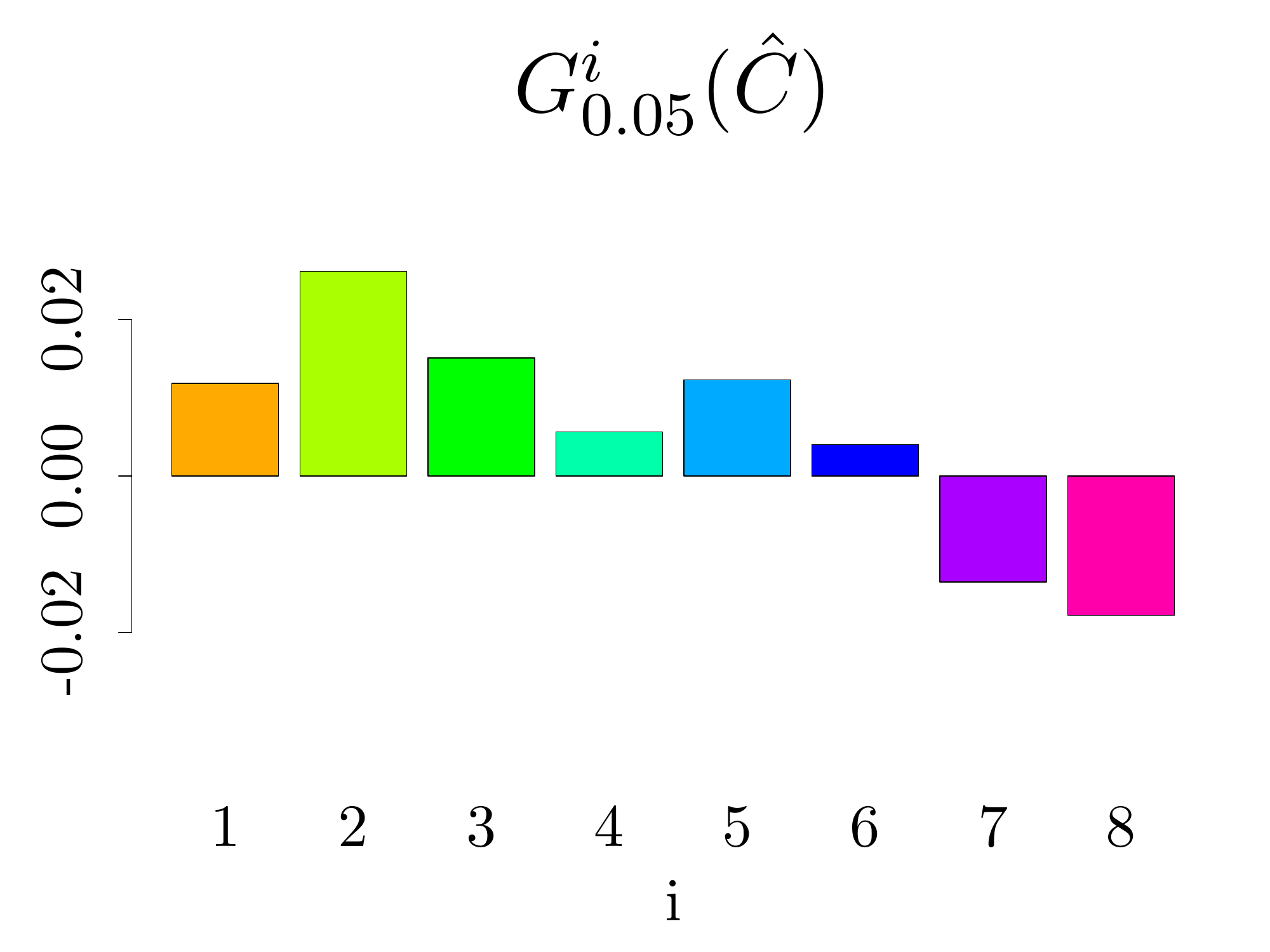}}
\scalebox{0.18}{\includegraphics{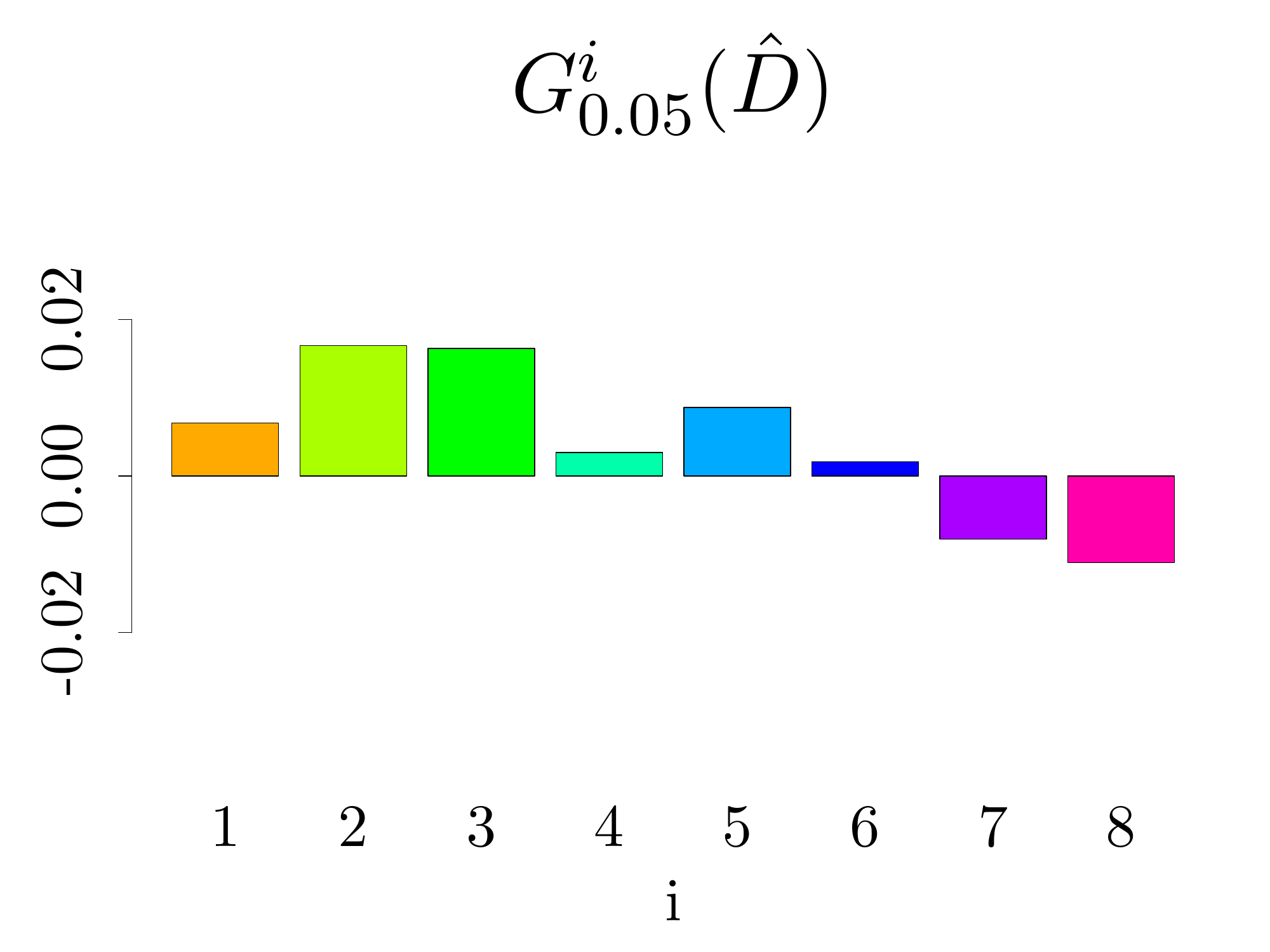}}\\
\scalebox{0.18}{\includegraphics{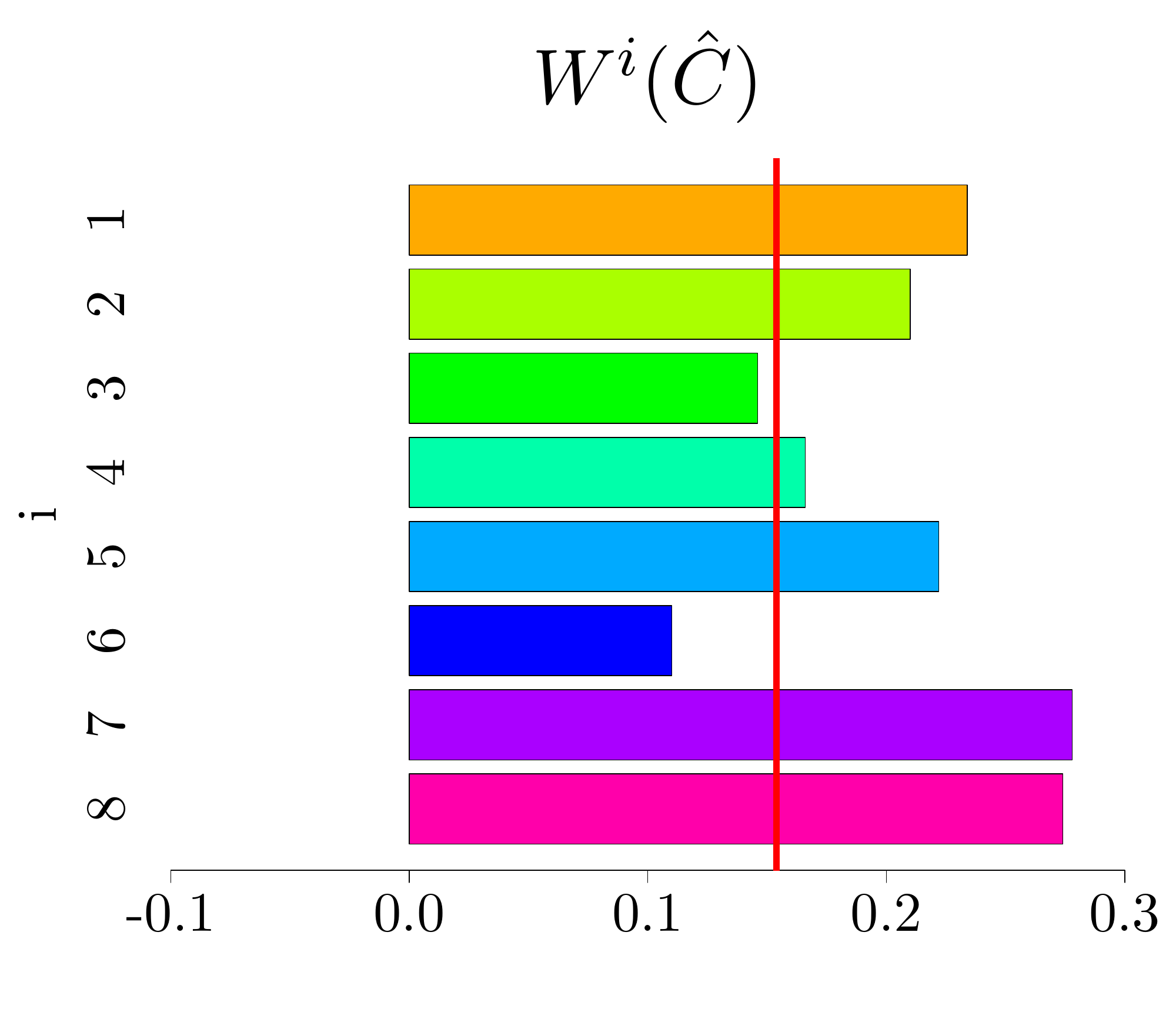}}
\scalebox{0.18}{\includegraphics{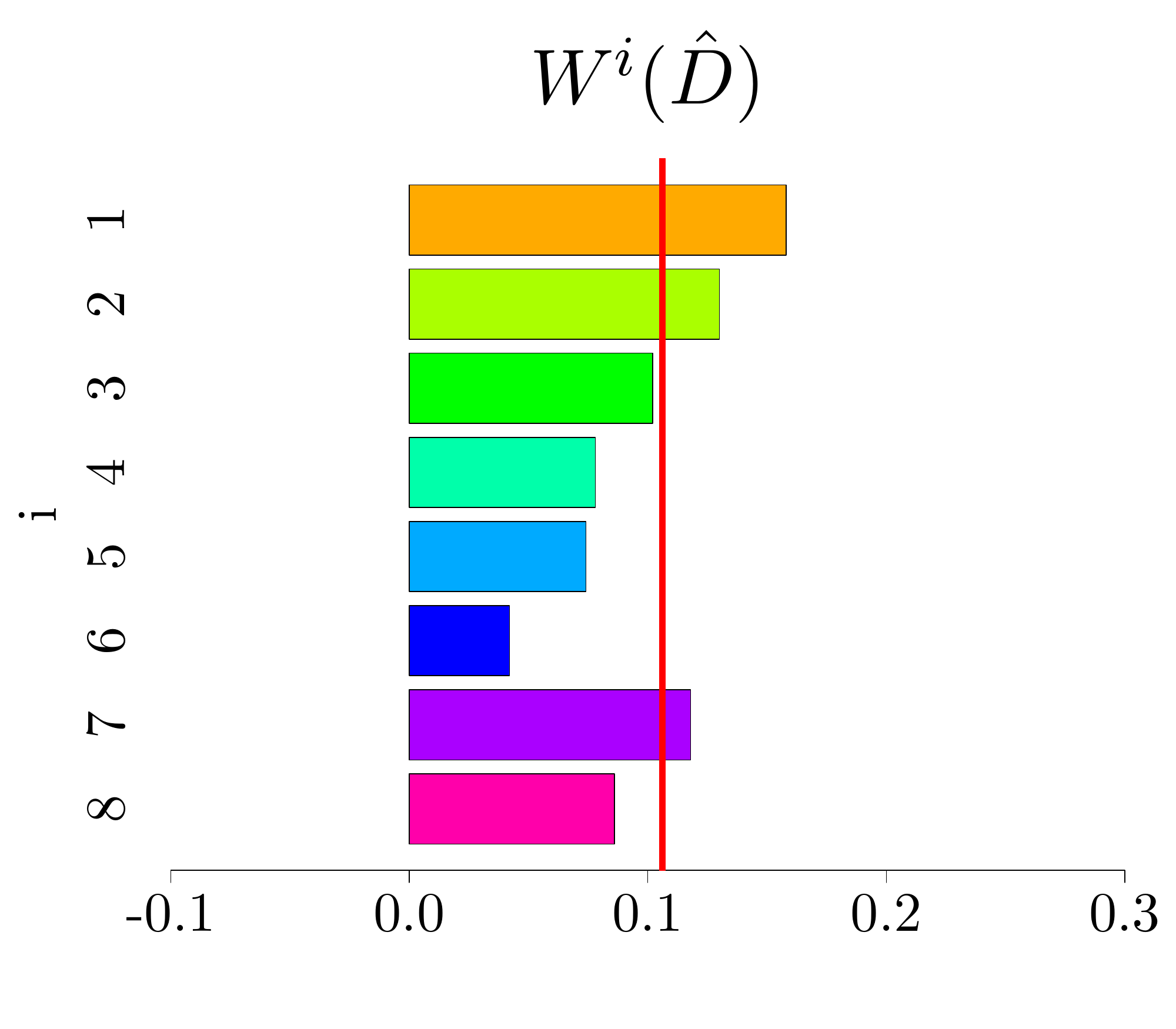}}
\end{multicols}
\caption{Aggregated risk, risk allocation and the backtesting metrics for portfolios in Example~\ref{ex:4}.
Two columns on the left (left panel) correspond to Dataset~1, while rightmost two columns (right panel) correspond to Dataset~2.
The results are obtained by using the risk allocation estimator $\widehat{C}$ and the nonparametric risk allocation estimators $\widehat{D}$.
}
\label{fig:Ex4-1}
\end{figure*}

\end{example}

\section*{Acknowledgments}
Tomasz R. Bielecki and Igor Cialenco acknowledge support from the National Science Foundation grant  DMS-1907568.
Marcin Pitera acknowledges support from the National Science Centre, Poland, via project 2016/23/B/ST1/00479.
The authors would also like to thank the anonymous referees, the associate editor and the editor for their helpful comments and suggestions which improved greatly the final manuscript.

\bibliographystyle{alpha}


\end{document}